\documentclass[a4paper, 3p, times, final, 10pt, numbers, sort&compress]{elsarticle} % 40 p
\usepackage{amsmath}
\usepackage[english]{babel}
\usepackage{graphicx}
\usepackage[usenames,dvipsnames]{color} 
\usepackage[footnotesize,bf,sf,center]{caption}
\usepackage{float}
\usepackage[dvips]{epsfig}
\usepackage{calc}
\usepackage{soul}
\usepackage{amsmath,amsfonts,amssymb,amsthm}
\usepackage{subfig}
\usepackage{setspace}
\usepackage{mathrsfs}
\usepackage{dsfont}

\newtheorem{theorem}{Theorem}[section]

\usepackage{algorithm}
\usepackage[noend]{algpseudocode}

\begin{document}
\newcommand{\RN}[1]{\textup{\uppercase\expandafter{\romannumeral#1}}}

\newcommand{\ve}[1]{\text{\boldmath${#1}$}} % vectors
\newcommand{\te}[1]{\text{\boldmath${#1}$}} % tensors

\newcommand{\ma}[1]   {\mathrm{#1}}             % FE matrices
\newcommand{\fevec}[1]{\underline{\mathrm{#1}}} % FE vectors

% Domain and boundary definitions
\newcommand{\Oft}{\Omega_{f}^t}
\newcommand{\Ost}{\Omega_{s}^t}
\newcommand{\hOf}{\hat{\Omega}_{f}}
\newcommand{\hOs}{\hat{\Omega}_{s}}
\newcommand{\hInt}{\hat{\Sigma}}
\newcommand{\Intt}{\Sigma^t}
\newcommand{\hGDs}{\hat{\Gamma}_{{D,s}}}
\newcommand{\hGNs}{\hat{\Gamma}_{{N,s}}}
\newcommand{\hGRs}{\hat{\Gamma}_{{R,s}}}
\newcommand{\GDf}{{\Gamma}_{{D,f}}^t}
\newcommand{\GNf}{{\Gamma}_{{N,f}}^t}
\newcommand{\GpseudoNf}{{\Gamma}_{{\tilde{N},f}}^t}
\newcommand{\GRf}{{\Gamma}_{{R,f}}^t}
\newcommand{\hns}{{\hat{\ve{n}}}_s}
\newcommand{\nf}{{\ve{n}}_f}
\newcommand{\dtale}[1]{\left. \frac{\partial}{\partial t} #1 \right|_{\mathcal{A}_t}}

\newcommand{\revised}[1]{#1} % \textcolor{blue}{#1}}

% -------------------------------------------------------------------------------------------
\begin{frontmatter}

  \renewcommand\arraystretch{1.0}
      \title{\textbf{{Efficient split-step schemes for fluid--structure interaction involving incompressible
      \revised{generalised Newtonian flows}}}}
    \author{
   {\bf Richard Schussnig}
     \footnote{{\sf Email address:} {\tt schussnig@tugraz.at}, 
     {\sf corresponding author}}, \ \ \  
   {\bf Douglas R. Q. Pacheco} $^2$ \ \ \
     %{\bf Olaf Steinbach} $^3$\ \ \
     %\footnote{Email address: {\tt o.steinbach@tugraz.at}}  \ \ \ 
        and \ \ \ 
   {\bf Thomas-Peter Fries} $^3$\\
     % \footnote{Email address: {\tt guenter.brenn@tugraz.at}} \\
   {\small ${}^{1,3}$ Institute of Structural Analysis, Graz University of Technology, Graz, Austria} \\
   {\small ${}^{2}$ Institute of Applied Mathematics, Graz University of Technology, Graz, Austria} \\
   {\small ${}^{1,2,3}$ Graz Center of Computational Engineering, Graz University of Technology, Graz, Austria}
   }
  \date{}
% -------------------------------------------------------------------------------------------
\begin{abstract}

Blood flow, dam or ship construction and numerous other problems in biomedical and general engineering involve incompressible flows interacting with elastic structures. Such interactions heavily influence the deformation and stress states which, in turn, affect the engineering design process. Therefore, any reliable model of such physical processes must consider the coupling of fluids and solids. However, complexity increases for non-Newtonian fluid models, as used, e.g.,~for blood or polymer flows. In these fluids, subtle differences in the local shear rate can have a drastic impact on the flow and hence on the coupled problem. \revised{There, existing (semi-)implicit solution strategies based on split-step or projection schemes for Newtonian fluids are not applicable, while extensions to non-Newtonian fluids can lead to substantial numerical overhead depending on the chosen fluid solver.} To address these shortcomings, we present here a higher-order accurate, added-mass-stable fluid--structure interaction scheme centered around a split-step fluid solver. We compare several implicit and semi-implicit variants of the algorithm and verify convergence in space and time. Numerical examples show good performance in both benchmarks and an \revised{idealised} setting of blood flow through an abdominal aortic aneurysm \revised{considering physiological parameters}.

\end{abstract}

\begin{keyword}
Fluid--structure interaction\sep
non-Newtonian fluid\sep
split-step scheme\sep
time-splitting method\sep
semi-implicit coupling\sep
incompressible viscous flow
%
% MSC codes here, in the form: \MSC code \sep code
% or \MSC[2008] code \sep code (2000 is the default)
\MSC 
74F10\sep
76A05\sep
76M10\sep 
74L15\sep
76D05
% 76M10: Finite element methods applied to problems in fluid mechanics
% 76A05: Non-Newtonian fluids
% 74F10: Fluid-solid interactions
% 74L15: Biomechanical solid mechanics
% 76D05: Navier--Stokes equations
\end{keyword}

\end{frontmatter}

\onehalfspacing

\section{Introduction}
\label{sec:intro}

Fluid--structure interaction (FSI) problems are characterised by a strong mutual dependence of fluid flow and structural deformation, exchanging momentum at the interface. Fluid forces acting on the solid cause deformation and induce strains, thereby influencing the stress state in the solid phase. A moving or deforming solid, in turn, alters the fluid flow domain and thus has a large impact on the flow quantities. Countless applications of FSI are found in science, engineering and biomedicine, ranging from airfoils or whole wind-turbines \cite{Machairas2018,Bazilevs2011}, bridge-decks \cite{Helgedagsrud2019}, offshore engineering \cite{Shiels2001,Bearman2011}, insect flight \cite{Chu2021} to blood flow through the circulatory system \cite{Bazilevs2006,Crosetto2011a,Crosetto2011b,Kuttler2010,Torii2006,Schussnig2021PAMM}, human phonation \cite{Thomson2005} or respiration \cite{Wall2008}. Consequently, the development of suitable models and solution procedures has been an active area of research over the past 40 years, which led to great advances in the field. Among the most popular numerical techniques to handle \revised{flows in moving domains -- a central part of any FSI scheme --} are arbitrary Lagrangian-Eulerian (ALE) \cite{Heil2004,Hughes1981,LeTallec2001,Leuprecht2002,Forti2017,Quaini2007,Quarteroni2000,Donea1982}, immersed boundary \cite{Fauci2006,Fogelson2004,Griffith2012,Griffith2009,Brandsen2021,Hesch2014} and fictitious-domain methods \cite{Baaijens2001,VanLoon2004,Boffi2017, Wang2017}, either tracking or capturing the motion of the fluid--structure interface.

On the fluid--structure interface, coupling conditions are enforced via monolithic or partitioned approaches.
In monolithic schemes, all balance equations are considered in one single system of equations, assembling contributions from all involved variables into the same matrix. This leads to an inherent tight coupling of physical fields, but unfortunately comes with an increased implementation effort, possibly unusual data structures and more involved preconditioners. Several variants to enforce the interface conditions include, e.g., Langrange multipliers \cite{Mayr2020, Crosetto2011a, Langer2018, Langer2016, Gerstenberger2008}, Nitsche's method \cite{Massing2015, Schott2019, Burman2020}, mortar techniques \cite{Hesch2014,Kloppel2011}, penalty approaches \cite{Kim2016,Vire2015,Vire2016} or formulations enforcing interface conditions through particular function space choices \cite{Hron2006,Richter2015,Wick2013,Schussnig2019,Jodlbauer2019,Balmus2020}. Ad hoc parallel preconditioners can be designed depending on the chosen setup \cite{Langer2016,Barker2010,Wu2014,Tezduyar2006,Crosetto2011a,Heil2004,Gee2011,Langer2015,Jodlbauer2019,Mayr2020,Aulisa2018}.

Partitioned coupling schemes, on the other hand, alleviate the development of efficient preconditioners by iteratively enforcing the interface conditions. They allow a ``reuse'' of already well-advanced solution algorithms tailored to specific applications and ease the inclusion of complex physics or discretisation methods in the individual fields \cite{Degroote2013,Hou2012,Hosters2018,Hilger2021}.
Using separate solvers and exchanging updates to the interface variables allows simpler software design, but shifts some of the intricacy to an outer coupling procedure. The added-mass effect \cite{Causin2005,Forster2007} is notorious for severely impairing performance for certain parameter combinations using standard partitioned schemes such as the serial staggered method \cite{Lesoinne1998} and even implicit partitioned methods as demonstrated, e.g., in \cite{Kassiotis2011,Kirby2007,Kuttler2008}. To counteract the hampered convergence behaviour in such cases, several remedies have been presented ranging from simple and effective Aitken relaxation \cite{Kuttler2008}, \revised{interface (quasi-)Newton or Newton--Krylov solvers} \cite{Gerbeau2003,Michler2005,Fernandez2005,Degroote2009,Spenke2020}, using Robin interface conditions \cite{Badia2009,Badia2008,Gerardo-Giorda2010} or artificial compressibility \cite{Degroote2010,Degroote2011,Bogaers2015}.

Nonetheless, performing several coupling iterations between fluid and solid phases remains costly. This was first alleviated by \citet{Fernandez2007}, introducing the concept of semi-implicit FSI. Therein, the fluid flow problem is solved via a projection or split-step scheme, decoupling velocity and pressure unknowns. Additionally, the fluid domain deformation is extrapolated from previous time steps, which then allows for coupling only the fluid pressure and solid deformation implicitly. 
Based on this rationale of avoiding the implicit coupling of all components, a rapid development was seen in the following years (see, e.g.,
\cite{Quaini2007,Badia2008b,Astorino2010,Breuer2012,Lozovskiy2015,He2015,Landajuela2017,Naseri2018,Fernandez2020}). In numerous challenging settings, these techniques were demonstrated to yield accurate and stable results while substantially increasing performance. Fully explicit treatment of interface conditions in FSI for problems with large added-mass effect was presented for shells (see, e.g., \cite{Fernandez2020,Nobile2008,Guidoboni2009,Lukacova-Medvidova2013,Banks2014,Fernandez2015}) and also for three-dimensional continua \cite{Burman2009,Burman2014,Banks2014b,Serino2019,Serino2019b}. Unfortunately, those fully explicit coupling schemes for bulk solids are as of now either limited to simple constitutive behaviour or only first-order accurate in time unless multiple correction steps are performed. 

Although FSI has been a major research topic in the past decades, approaches specifically targeting generalised Newtonian fluid flow are still scarcely found in literature \cite{Lukacova-Medvidova2013,Janela2010,Zhu2017,Zhu2019,Amani2020}. \revised{And while partitioned FSI approaches and performant acceleration schemes do allow a straight-forward incorporation of complex constitutive laws, the promising concept of semi-implicit FSI builds upon methods decoupling fluid velocity and pressure, which then need to be capable of considering non-Newtonian models.} 
Within such fluid solvers segregating velocity and pressure, necessary projection methods need to be free of nonphysical pressure boundary layers, which was not the case in the first generation of projection methods (see the work of \citet{Guermond2005,Guermond2006} for an excellent discussion on such schemes). An alternative solution was presented in \cite{Johnston2004,Liu2009}, replacing the continuity equation by a pressure Poisson equation (PPE) equipped with fully consistent boundary conditions. Further using extrapolated pressure and convective velocities in the fluid's balance of linear momentum, as often done also in mixed velocity-pressure formulations \cite{Simo1994,TUREK1996,Elman2011,Ingram2013,Schussnig2021}, allows to decouple velocity components and pressure unknowns completely. The PPE framework was recently extended to the generalised Newtonian case \cite{Pacheco2020,Pacheco2021b}, enabling an efficient parallel solution using \revised{open-source finite element libraries and linear algebra packages as black boxes}.

In this context, we present a novel split-step framework for partitioned FSI with incompressible, generalised Newtonian fluid flows and three-dimensional continua. Fluid velocity and pressure are completely decoupled using higher-order and possibly adaptive time-stepping and extrapolation formulae in a split-step scheme. This allows for equal-order, standard $C^0$-continuous interpolation, which is a major advantage compared to similar projection-based methods. Additionally, we derive fully consistent boundary and coupling conditions to preserve accuracy on the fluid--structure interface. Semi-implicit variants of the scheme are designed in an added-mass-stable way by implicitly coupling merely the solid displacement and fluid pressure. The remaining subproblems are treated in an explicit fashion, which enhances performance without degrading accuracy or stability. Moreover, we improve mass conservation through so-called divergence damping and achieve great flexibility concerning the rheological law, such that modifying the fluid material behaviour is as simple as exchanging the right-hand side of the viscosity projection step. All resulting linear systems are easily tackled using off-the-shelf black-box preconditioning techniques available as open-source scientific software, making the scheme an attractive alternative to available methods.

The remainder of this paper is organised as follows: Individual field equations are presented in Section~\ref{sec:subproblems}, where time integration is carried out using higher-order extrapolation and backward-differentiation formulae of order $m$ (BDF-$m$) for the fluid phase and the generalised-$\alpha$ scheme for the solid phase. In Section~\ref{sec:coupled_problem}, fully implicit and semi-implicit variants of the coupling scheme are presented, enforcing consistent Robin conditions at the fluid--structure interface. The computational performance of the schemes is assessed in Section~\ref{sec:computational_results} with i) numerical tests of temporal and spatial convergence in two cases with analytic solutions \cite{Serino2019,Serino2019b}, ii) the classical pressure pulse benchmark in three space dimensions (see, e.g., \cite{Formaggia2001,Janela2010,Langer2018}), and iii) a final numerical experiment in the context of aortic blood flow, highlighting the potential of the presented approach for practical application.

\section{Fluid and structure models in an ALE framework}
\label{sec:subproblems}
The computational domain at time $t$, denoted by $\Omega^t\subset \mathbb{R}^d$ with \mbox{$d=2\text{ or }3$}, is composed of the fluid and solid subdomains $\Oft$ and $\Ost$ with the moving interface $\Intt = \partial\Oft \cap \partial\Ost$. Further, let us introduce for any function $g(\ve{x},t)$ with $\ve{x}\in\Omega^t$ its counterpart $\hat{g}(\hat{\ve{x}},t) = g(\ve{x},t)$ living in the reference configuration $\hat{\Omega}$. The transformations from reference to current domains are defined as
\begin{align*}
	\mathcal{A}_t:\hOf \rightarrow \Oft\,,\, 
	\mathcal{A}_t =\hat{\ve{x}} + \ve{d}_f(\hat{\ve{x}},t)\,,
	\qquad
	\mathcal{L}_t:\hOs \rightarrow \Ost\,,\,
	\mathcal{L}_t =\hat{\ve{x}} + \ve{d}_s(\hat{\ve{x}},t)
	,
	%\label{eqn:map_def}
\end{align*}
with deformations $\ve{d}_f$ and $\ve{d}_s$, following the arbitrary Lagrangian-Eulerian \cite{Hughes1981,Donea2017} and total Lagrangian approaches \cite{holzapfel2000,Bonet2008}. This gives rise to the deformation gradients $\te{F}_{\! f}$ and $\te{F}_{\! s}$ with respective Jacobians $J_f$ and $J_s$:
\begin{align*}
\te{F}_{\! f} = \te{I} + \nabla \ve{d}_f\,,\,
J_f =\det \te{F}_{\! f} \,,
\qquad
\te{F}_{\! s} = \te{I} + \nabla \ve{d}_s\,,\,
J_s =\det \te{F}_{\! s}
.
%\label{eqn:map_def}
\end{align*}
\subsection{Mesh update}
\label{sec:subproblem_mesh}
The Lagrangian transformation $\mathcal{L}_t$ is naturally defined by the deformation itself, but in the fluid domain, the mapping $\mathcal{A}_t$ is constructed, e.g., via harmonic extension:
\begin{eqnarray}
-\nabla \cdot \left( c \nabla \ve{d}_f \right) =&0 & \quad \text{in } \hOf \, ,\label{eqn:harm_ext}\\
\ve{d}_f =&\ve{d}_s & \quad \text{on } \hInt \, ,\label{eqn:harm_ext_int_cond}\\
\ve{d}_f =&\ve{0} & \quad \text{on } \partial\hOf\setminus\hInt
,
\end{eqnarray}
with a possibly nonlinear stiffening parameter $c = c_0 \left(J_f + J_f^{-1}\right)$ and suitably chosen $c_0$ \cite{Wick2011,Stein2003}. This is only the simplest choice out of a wide variety of existing methods (cf. \cite{Stein2003,JOHNSON1994,Wick2011,Shamanskiy2021}), and may be further decomposed into equations in individual components of $\ve{d}_f$. Having introduced the basic setting, let us proceed with the formulation of balance equations in the individual subdomains.
\subsection{Structure models}
\label{sec:subproblem_solid}
The balance of linear momentum in the reference configuration of the solid domain $\hOs$ is expressed in terms of the solid displacement $\ve{d}_s$ as
\begin{eqnarray}
	\rho_s \frac{\partial^2}{\partial t^2} \ve{d}_s - \nabla \cdot \te{P} =& \ve{0} & \quad \text{in } \hOs \,, \label{eqn:solid_mom_bal}
\end{eqnarray}
with the solid's density $\rho_s$ and the first Piola--Kirchhoff stress tensor $\te{P}$, omitting body forces for brevity. Initial conditions and boundary conditions on the non-overlapping Dirichlet, Neumann and Robin boundary sections, denoted by $\hGDs$, $\hGNs$ and $\hGRs$, are given by
\begin{eqnarray}
\ve{d}_s =& \ve{d}_0 & \quad \text{at } t=0 \,,\\
\frac{\partial}{\partial t}\ve{d}_s =& \dot{\ve{d}}_0 & \quad \text{at } t=0\,,\\
\ve{d}_s =& {\ve{g}}_{s} & \quad \text{on } \hGDs \,, \label{eqn:solid_D}\\
\te{P} \hat{\ve{n}}_s =& {\ve{t}}_{s} & \quad \text{on } \hGNs \,, \label{eqn:solid_N}\\
\eta_s^R \frac{\partial}{\partial t} \ve{d}_s + \te{P} \hat{\ve{n}}_s =& {\ve{h}}_s & \quad \text{on } \hGRs \label{eqn:solid_R}
\,,
\end{eqnarray}
with $\hns$ denoting the unit outward normal in the reference configuration and the Robin parameter $\eta_s^R>0$. Constitutive equations linking stress and strain measures in the solid are herein formulated in terms of the second \mbox{Piola--Kirchhoff} stress tensor $\te{S}$, additionally using the relation $\te{P} = \te{F}_{\! s} \te{S}$. Assuming isotropic, linear elastic material behavior leads to the St.~Venant--Kirchhoff model
\begin{align}
	\te{S} = \lambda_s \text{tr}\left(\te{E}\right) \te{I} + 2 \mu_s \te{E}	
	\, , \quad \text{with }
	\te{E} = \frac{1}{2} \left(\te{C} - \te{I}\right)
	\, \text{and }
	\te{C} = \te{F}_{\! s}^T\te{F}_{\! s}
	,
	\label{eqn:solid_const_rel_StVenant}
\end{align}
with Lam{\'e} parameters 
\begin{align*}
	\mu_s = \frac{E_s}{2\left(1+\nu_s\right)}
	\quad \text{and} \quad
	\lambda_s = \frac{E_s\nu_s}{(1-2\nu_s)(1+\nu_s)}
\end{align*}
expressed in terms of Young's modulus $E_s$ and Poisson ratio $\nu_s$. \revised{Further assuming $|| \nabla \ve{d}|| \ll 1$, one obtains for the case of linear elasticity}
\begin{align}
	\te{P} 
	= \te{F}_{\! s} \te{S} \approx \te{S} 
	= \lambda_s \left(\nabla \cdot \ve{d}_s\right) \te{I}
	+ \mu_s \left[\nabla \ve{d}_s + \left(\nabla \ve{d}_s\right)^T\right]
	\, .
	%\, , \quad \text{with }
    %\nabla^S \ve{d}_s := \frac{1}{2}\left[\nabla \ve{d}_s + \left(\nabla \ve{d}_s\right)^T\right]	
	\label{eqn:solid_const_rel_linelast}
\end{align}
%
%being the symmetric displacement gradient. 
Other constitutive relations of particular interest in biomedical engineering are those describing rubber-like materials such as arterial tissue (see, e.g., \cite{holzapfel2000,Gasser2006,Bazilevs2010,Simo1998,Baumler2020}). Herein, a quasi-incompressible neo-Hookean model \cite{Simo1998}
\begin{align}
\text{\revised{$\te{P}_\mathrm{NH}$}} = \mu_s J_s^{-2/3} 
\left( 
	\te{F}_{\! s} - \frac{1}{3} I_1 \te{F}_{\! s}^{-T}
\right)
+ \frac{\kappa_b}{2} \left(J_s^2-1\right)\te{F}_{\! s}^{-T}
,
\label{eqn:solid_const_rel_qiNH}
\end{align}
with the invariant $I_1 = \text{tr}\left( \te{C} \right)$ and bulk modulus $\kappa_b = E/\left[3(1-2\nu)\right]$ is considered. This hyperelastic model is either used as is, choosing $\te{P}=\te{P}_{\mathrm{NH}}$, or together with contributions from dispersed collagen fibers. In the latter case, $\te{P}$ is decomposed into $\te{P}_{\mathrm{NH}}$ as defined in Equation~\eqref{eqn:solid_const_rel_qiNH} and an additional term, such that $\te{P}$ is given as \cite{Gasser2006}
\begin{gather}
	\te{P} = \te{P}_{\mathrm{NH}} + 
	\te{F}_{\! s}
	\sum_{i=4,6}
	\left[
	2 k_1 {G}_i \exp\left(k_2 {G}_i^2\right) 
	\frac{\partial {G}_i}{\partial \te{C}}
	\right],
	\label{eqn:solid_const_rel_qiHGO2006}
\\
\text{with }
    {G}_i = J_s^{-2/3} \left[\kappa_c I_1 + (1-3\kappa_c) I_i\right] - 1 \, ,
\nonumber\\
    \frac{\partial{G}_i}{\partial\te{C}}
    =
    J^{-2/3}_s
    \left[
    \kappa_c \te{I} + (1-3\kappa_c) \te{A}_i - \frac{1}{3} \te{C}^{-1}
    \left( \kappa_c I_1 + (1-3\kappa_c) I_i\right)
    \right]\nonumber
    .
\end{gather}
Here, we introduce fiber parameters $k_1$, $k_2$ and $\kappa_c$, the tensors $\te{A}_i = \hat{\ve{m}}_i \otimes \hat{\ve{m}}_i$ with mean fiber directions $\hat{\ve{m}}_i$, the corresponding invariants $I_i = \te{C}:\te{A}_i$ and $G_i$ for ease of notation \cite{DeVilliers2018}. Usually, in biomechanical applications, one defines the mean fiber directions $\hat{\ve{m}}_i$ relative to circumferential ($\hat{\ve{e}}_1$) and longitudinal ($\hat{\ve{e}}_2$) directions of the vessel via
\begin{gather}
	\hat{\ve{m}}_4 = \frac{\hat{\ve{e}}_1 + \hat{\ve{e}}_2 \tan \alpha_c }{||\hat{\ve{e}}_1 + \hat{\ve{e}}_2 \tan \alpha_c||}    
	\quad \text{and} \quad
	\hat{\ve{m}}_6 = \frac{\hat{\ve{e}}_1 - \hat{\ve{e}}_2 \tan \alpha_c }{||\hat{\ve{e}}_1 - \hat{\ve{e}}_2 \tan \alpha_c||}
	,
	\label{eqn:fiber_orientation}
\end{gather}
with $||\cdot||$ denoting the Euclidean norm, leading to $\pm \alpha_c$ describing the deviation from circumferential vessel direction. \revised{The vector fields $\hat{\ve{e}}_1$ and $\hat{\ve{e}}_2$ are herein constructed with a two step procedure: First, two auxiliary scalar Laplace equations are solved with boundary conditions on the in- and outlets and the fluid--structure interface, such that the gradients of the resulting fields approximate the longitudinal ($\hat{\ve{e}}_2$) and radial vessel orientation. Afterwards, the circumferential vessel direction ($\hat{\ve{e}}_1$) is constructed using the radial and longitudinal vectors. Then, \eqref{eqn:fiber_orientation} gives mean fiber directions rotated by $\pm \alpha_c$ from the circumferential vessel direction into longitudinal direction.} With these constitutive relations defined, differentiating between the individual stress tensors is not necessary, simply denoting the first Piola--Kirchhoff stress tensor by $\te{P}$ for all of the above material laws.

The structural equations are discretised in time for $t\in(0,T]$, decomposing the interval into $N_t$ steps with size $\Delta t = t^{n+1}-t^n$, $n=0,...,N_t$ and using Newmark formulae \cite{Newmark1959}
\begin{align}
\frac{\partial^2}{\partial t^2} \ve{d}_s (\hat{x},t^{n+1})
&=:
\ddot{\ve{d}}_s^{n+1}
\approx
\frac{1}{\beta \Delta t^2} \left( \ve{d}_s^{n+1} - \ve{d}_s^n\right)
-
\frac{1}{\beta \Delta t} \dot{\ve{d}}_s^n
+ \left(1 - \frac{1}{2 \beta}\right) \ddot{\ve{d}}_s^n
\,, \label{eqn:newmark_dtt_ds} \\
\frac{\partial}{\partial t} \ve{d}_s (\hat{x},t^{n+1})
&=:
\dot{\ve{d}}_s^{n+1}
\approx
\frac{\gamma}{\Delta t \beta} \left(\ve{d}_s^{n+1} - \ve{d}_s^n\right)
+
\left(1-\frac{\gamma}{\beta}\right) \dot{\ve{d}}_s^{n}
+
\Delta t\left(1-\frac{\gamma}{2\beta}\right) \ddot{\ve{d}}_s^{n}
\,, \label{eqn:newmark_dt_ds}
\end{align} 
together with the generalised-$\alpha$ method \cite{Chung1993}, to obtain the time-discrete form of momentum balance in terms of structural displacements $\ve{d}_s^{n+1}$ as
\begin{eqnarray}
	\rho_s \left(
	\alpha_m'\ddot{\ve{d}}_s^{n+1}
	+ \alpha_m \ddot{\ve{d}}_s^n
	\right)
	- \alpha_f' \nabla \cdot \te{P}\left(\ve{d}_s^{n+1}\right)
	- \alpha_f \nabla \cdot \te{P}\left(\ve{d}_s^n\right)
	= \ve{0} \quad \text{in } \hOs
	\,,
	\label{eqn:solid_mom_bal_time_discr}
\end{eqnarray}
where we introduce the shorthand-notations $\alpha_m' = 1-\alpha_m$ and $\alpha_f' = 1-\alpha_f$. The nonlinear terms in Equation~\eqref{eqn:solid_mom_bal_time_discr} are integrated via the generalised trapezoidal rule. Setting the parameters $\gamma$, $\beta$, $\alpha_m$ and $\alpha_f$ in the above time integration scheme determines accuracy and stability properties of the resulting scheme. Following \cite{Chung1993}, we achieve second-order accuracy and unconditional stability for linear problems by choosing
\begin{align}
\label{eqn:generalised_alpha_accuracy_stability_conditions}
\gamma = \frac{1}{2}-\alpha_m+\alpha_f
\, , \quad 
\beta = \frac{1}{4} \left(1-\alpha_m+\alpha_f\right)^2
.
\end{align}
The analysis in \cite{Chung1993} was extended by \citet{Erlicher2002}, proving second-order accuracy and energy stability in the high-frequency range depending on the algorithmic parameter $\rho_\infty$ even for nonlinear problems. The user-specified spectral radius in the high frequency limit $\rho_\infty$ is used to specify $\alpha_m$ and $\alpha_f$ according to Table~\ref{tab:gen_alpha_alphas}, resulting in the Newmark-$\beta$ (N-$\beta$) \cite{Newmark1959}, HHT-$\alpha$ \cite{Hilber1977}, WBZ-$\alpha$ \cite{Wood1980} or CH-$\alpha$ \cite{Chung1993} schemes.
\begin{table}[ht!]
	\centering
	\begin{tabular}{||c|c|c|c|c||} 
		\hline
		&&&&\\[-2ex]
		& Newmark-$\beta$ & HHT-$\alpha$ & WBZ-$\alpha$  & CH-$\alpha$ \\
		\hline\hline
		&&&&\\[-1ex]
		$\alpha_m$ &0 &0  &$\frac{\rho_\infty - 1}{1+\rho_\infty}$&$\frac{2\rho_\infty - 1}{1+\rho_\infty}$\\ [1.5ex]
		$\alpha_f$ &0 &$\frac{1-\rho_\infty}{1+\rho_\infty}$&0& $\frac{\rho_\infty}{1+\rho_\infty}$
		\\[1.4ex]
		\hline
	\end{tabular}
	\caption{Algorithmic parameters $\alpha_m$ and $\alpha_f$ in the generalised-$\alpha$ time integration scheme.}
	\label{tab:gen_alpha_alphas}
\end{table}
Given the time-discrete form of the momentum balance residual \eqref{eqn:solid_mom_bal_time_discr}, one proceeds by employing Newton's method (see, e.g., \cite{Bonet2008,holzapfel2000,DeVilliers2018}). In each step $k$ of Newton's method, the last iterate $\fevec{d}_s^k$ of the current time step's solution $\fevec{d}_s^{n+1}$, \revised{both being discrete finite element vectors approximating their respective continuous counterparts}, is updated via
\begin{align}
	\fevec{d}_s^{k+1} = \fevec{d}_s^k + \Delta \fevec{d}_s^k
	\quad\text{until}\quad
	||\fevec{d}_s^{k+1}-\fevec{d}_s^k||<\epsilon_N ||\fevec{d}_s^0||
	\label{eqn:newton_rel_tol}
\end{align}
is fulfilled. Therein, the increment $\Delta\fevec{d}_s^k$ is the solution of the Jacobian system \revised{written in matrix-vector form}
\begin{align}
	\ma{J}(\fevec{d}_s^k) \, \,\Delta\fevec{d}_s^k = - \fevec{r}(\fevec{d}_s^k)
	\, ,
	\label{eqn:solid_jacobian}
\end{align}
which is based on the standard problem of finding $\delta\ve{d}_s^k \in [H^1(\hOs)]^d$ with $\delta\ve{d}_s^k|_{\hGDs} = \ve{0}$, such that
\begin{align}
% dtt_d term
&\frac{\rho_s\alpha_m'}{\beta \Delta t^2}
\langle \ve{\varphi}, \delta\ve{d}_s^k\rangle_{\hOs}
% stress term
+ \alpha_f' \langle \nabla \ve{\varphi}, 
\frac{\partial}{\partial\ve{d}_s}%\mathcal{D}_{\ve{d}_s} 
\te{P}\left(\ve{d}_s^{k}\right) \rangle_{\hOs}
% Robin term
+ \eta_s^R\frac{\alpha_f'\gamma}{\beta\Delta t} \langle \ve{\varphi}, \delta\ve{d}_s^k \rangle_{\hGRs}
\nonumber\\
=&
% time term
- \rho_s \langle \ve{\varphi}, 
\alpha_m'\ddot{\ve{d}}_s^{n+1}\left(\ve{d}_s^k\right)
+ \alpha_m \ddot{\ve{d}}_s^n
\rangle_{\hOs}
% stress domain terms
- \langle \nabla \ve{\varphi} , \alpha_f' \te{P}\left(\ve{d}_s^{k}\right)
+ \alpha_f \te{P}\left(\ve{d}_s^n\right) \rangle_{\hOs} 
\nonumber\\
%
% stress Neumann boundary terms
&+ \langle \ve{\varphi}, \alpha_f' {\ve{t}}_s^{n+1} + \alpha_f {\ve{t}}_s^n \rangle_{\hGNs} 
% stress Robin boundary terms
+ \langle \ve{\varphi}, 
\alpha_f' \left[ {\ve{h}}_s^{n+1} - \eta_s^R \dot{\ve{d}}_s^{n+1}\left(\ve{d}_s^k\right) \right] 
+ \alpha_f \te{P}\left(\ve{d}_s^n\right) \hns \rangle_{\hGRs} 
\label{eqn:solid_mom_bal_weak_Newton}
\end{align}
for all $\ve{\varphi} \in [H^1(\hOs)]^d$, with $\ve{\varphi}|_{\hGDs} = \ve{0}$ and $\langle\cdot,\cdot\rangle_{\hOs}$ or $\langle\cdot,\cdot\rangle_{\hat{\Gamma}}$ denoting the $L^2(\hOs)$ or $L^2(\hat{\Gamma})$ inner products. Moreover, $\frac{\partial}{\partial\ve{d}_s} \te{P}$ %$\mathcal{D}_{\ve{d}_s}$ 
denotes the directional derivative of $\te{P}$ with respect to $\ve{d}_s$ (see, e.g., \cite{holzapfel2000,Bonet2008}), which is omitted here for brevity. Known Neumann (${\ve{t}}_s$) and Robin (${\ve{h}}_s$) boundary data are plugged into the boundary terms arising from integrating the stress-divergence terms by parts. This completes the solution procedure for the nonlinear elastodynamics equations with Neumann and Robin boundary terms, various material models and generalised-$\alpha$ time integration.
\subsection{Fluid models}
\label{sec:subproblem_fluid}
The Navier--Stokes equations for incompressible flow, comprised of the linear momentum balance and continuity equations, reads in ALE form (assuming zero body forces):
\begin{eqnarray}
	\rho_f \left[
	\dtale{\ve{u}_f}
	+
	\nabla \ve{u}_f \left(\ve{u}_f - \ve{u}_m \right)
	\right]
	-
	\nabla\cdot \te{\sigma}_f 
	=& \ve{0}
	&\quad \text{in } \Oft
	\,,
	\label{eqn:fluid_u_p_mom}
	\\
	\nabla\cdot\ve{u}_f =& 0 &\quad \text{in } \Oft
	\,,
	\label{eqn:fluid_u_p_cont}
	\\
	\ve{u}_f =& \ve{u}_0
	& \quad \text{at } t=0
	\,,
	\\
	\ve{u}_f =& \ve{g}_f &\quad \text{on } \GDf 
	\,,
	\\
	\te{\sigma}_f \nf =& {\ve{t}}_f &\quad \text{on } \GNf 
	\,,
	\label{eqn:fluid_BC_full_traction}
	\\
	\eta_f^R \ve{u}_f + \te{\sigma}_f \nf =& {\ve{h}}_f &\quad \text{on } \GRf
	\,,
	\label{eqn:fluid_u_p_robin}
\end{eqnarray}
with the fluid velocity $\ve{u}_f$, stress tensor $\te{\sigma}_f$, unit-outward normal in the current configuration $\nf$, Robin parameter $\eta_f^R>0$, ALE time-derivative $\dtale{\ve{u}_f}$ and mesh velocity \mbox{$\ve{u}_m := \frac{\partial}{\partial t} \ve{d}_f$}, the latter two of which are connected via
\begin{align*}
% see, e.g., Wall&Ramm, Comput Mech, 1998
	\dtale{\ve{u}_f} := \frac{\partial}{\partial t} \hat{\ve{u}}_f\left(\mathcal{A}_t^{-1}\left(\ve{x},t\right),t\right)
	+ \left(\nabla \ve{u}_f\right) \ve{u}_m 
	\, .
\end{align*}
The fluid stress tensor $\te{\sigma}_f$ of a generalised Newtonian (or, as often called, quasi-Newtonian) fluid is given by
\begin{align}
\te{\sigma}_f := - p_f \te{I} + 2 \mu_f\left(\dot{\gamma}(\ve{u}_f)\right) \nabla^S \ve{u}_f
\, ,
\quad\text{with}\quad
\nabla^S \ve{u}_f := \frac{1}{2}\left[\nabla \ve{u}_f + \left(\nabla \ve{u}_f\right)^T\right]
\label{eqn:fluid_full_sigma_f}
\end{align}
\revised{being the symmetric gradient of $\ve{u}_f$}, the fluid's pressure $p_f$ and dynamic viscosity $\mu_f$, which is most commonly expressed through a nonlinear map $\mu_f = \eta(\dot{\gamma})$, $\eta:\mathbb{R}_+\rightarrow\mathbb{R}_+^*$, dependent on the fluid's shear rate defined as
\begin{align*}
\dot{\gamma}(\nabla^S\ve{u}_f) := \sqrt{\frac{1}{2} \nabla^S\ve{u}_f : \nabla^S\ve{u}_f}
\, .
\end{align*}
Rheological models of great interest in biomedical and industrial applications may describe shear-thickening or shear-thinning behaviour, as of particular importance when considering flows of polymer melts or blood, and may be put into the general form \cite{Galdi2008}
\begin{align}
\eta(\dot{\gamma}) = \eta_\infty + (\eta_0 - \eta_\infty) 
\left[
\kappa_f + \left(\lambda_f \dot{\gamma}\right)^a
\right]^{\frac{n-1}{a}}
\, .
\label{eqn:fluid_dyn_visc_rheology}
\end{align}
Therein, $\eta_0$ and $\eta_\infty$ denote viscosity limits and the fitting parameters $\kappa_f$, $\lambda_f$, $a$ and $n$ can be used to retrieve the power-law ($\kappa_f = \eta_\infty = 0$), Carreau ($\kappa_f = 1$, $a=2$) or Carreau-Yasuda ($\kappa_f = 1$) models, but also the standard Newtonian model ($\eta_0 = \eta_\infty$). 

In the following, we extend the split-step scheme from \cite{Pacheco2021b} to a moving grid by replacing the standard velocity-pressure form \eqref{eqn:fluid_u_p_mom}--\eqref{eqn:fluid_u_p_robin} by the following set of equations to advance the fluid velocity and pressure in time:
\begin{align}
\rho_f \left[
\dtale{\ve{u}_f}
+
\nabla \ve{u}_f \left(\ve{u}_f - \ve{u}_m \right)
\right]
-
\nabla \cdot \left(
2 \mu_f \nabla^S \ve{u}_f
\right)
&= -\nabla p_f &&\text{in } \Oft
\,,
\label{eqn:fluid_splitstep_mom}
\\
\nabla\cdot \left[\rho_f \nabla\ve{u}_f\left(\ve{u}_f - \ve{u}_m\right) - 2\nabla^S\ve{u}_f \nabla\mu_f \right]
+ \left[\nabla\times\left(\nabla \times \ve{u}_f\right)\right]\cdot\nabla \mu_f
 &= -\Delta p_f &&\text{in } \Oft
 \,, \quad
\label{eqn:fluid_splitstep_ppe}
\end{align}
with additional consistent boundary and initial conditions given as
\begin{align}
\nabla \cdot \ve{u}_0 &= 0
&& \text{in } \Omega_f^{t=0}
\,, \label{eqn:fluid_split_step_div_u0}
\\
- \mu_f \nabla \cdot \ve{u}_f  
+ \nf \cdot \left(2\mu_f\nabla^S\ve{u}_f \nf - \ve{t}_f\right)
%- \nf\cdot \ve{t}_f
&= p
&&\text{on } \GNf \,,
\label{eqn:fluid_splitstep_BC_pressure_dirichlet_onGNf}
\\
- \mu_f \nabla \cdot \ve{u}_f  
+ \nf \cdot \left(2\mu_f\nabla^S\ve{u}_f \nf - \ve{h}_f + \eta_f^R \ve{u}_f \right)
%- \nf\cdot (\ve{h}_f - \eta_f^R \ve{u}_f)
&= p
&&\text{on } \GRf \,,
\label{eqn:fluid_splitstep_BC_pressure_dirichlet_onGRf}
\\
\nf\cdot\left\{
-\rho_f\left[
\dtale{\ve{u}_f} + \nabla \ve{u}_f \left(\ve{u}_f - \ve{u}_m\right)
\right]
- \mu_f \left[\nabla\times\left(\nabla\times\ve{u}_f\right)\right]
+ 2\nabla^S\ve{u}_f \nabla \mu_f
\right\}
&= \nf\cdot \nabla p_f &&\text{on } \GDf 
\,.
\label{eqn:fluid_splitstep_BC_pressure_neumann}
\end{align}
Altogether, this set of equations allow in their final form allow using $C^0$-continuous finite element discretisations and decouple the balance of linear momentum and continuity equations, yielding standard discrete problems, for which off-the-shelf black-box preconditioners can be employed. This is a straight-forward extension of \cite{Pacheco2021b} to moving domains, which itself considers fluid flows on fixed grids and is based on \cite{Liu2009} for the Newtonian case with open/traction boundary conditions and \cite{Johnston2004} for pure Dirichlet problems.
\begin{theorem}
	For sufficiently regular $p_f, \ve{u}_f$, \revised{$\ve{g}_f$}, $\ve{t}_f, \ve{h}_f$, systems \eqref{eqn:fluid_splitstep_mom}--\eqref{eqn:fluid_splitstep_BC_pressure_neumann} and \eqref{eqn:fluid_u_p_mom}--\eqref{eqn:fluid_u_p_robin} are equivalent.
    \label{theorem:equivalence}
\end{theorem}
\begin{proof}
    \revised{See appendix.}
\end{proof}

This split-step scheme is not plagued by spurious pressure boundary layers, since the pressure is recovered from a fully consistent PPE rather than updated as in classical pressure-correction methods \cite{Johnston2004,Liu2009,Pacheco2021b,Guermond2005}. However, Liu~\cite{Liu2009a} observed that stability can be improved significantly by performing a Leray projection, which is of particular importance when considering nonsmooth solutions. So, we aim to improve stability of the overall scheme and suppress accumulation of errors in mass conservation by solving the simple Poisson problem
\begin{align}
	-\Delta \psi &= - \nabla \cdot \ve{u}_f&&\text{in } \Oft \,,\\
	\nf \cdot \nabla \psi &= 0 &&\text{on }\GDf\,,\\
	\psi &= 0 &&\text{on }\partial\Oft\setminus\GDf\,.
	\label{eqn:leray}
\end{align}
and updating $\check{\ve{u}}_f := \ve{u}_f - \nabla {\psi}$ via projection. It is then easily verified that
\begin{align*}
	\nabla \cdot\check{\ve{u}}_f &= 0 && \text{in }\Oft \,,\\
	\check{\ve{u}}_f \cdot \nf &= \ve{u}_f \cdot \nf && \text{on } \GDf\,,\\
	\check{\ve{u}}_f \cdot \ve{\tau}_f &= \ve{u}_f \cdot \ve{\tau}_f && \text{on } \partial\Oft\setminus\GDf\,,
\end{align*} 
with any tangential vector $\ve{\tau}_f$. Both $\ve{u}_f$ and its (weakly) divergence-free counterpart $\check{\ve{u}}_f$ converge at the same rates~\cite{Guermond2006}, making them equally valuable options from an accuracy point of view. However, it is also clear that one either settles for improved mass conservation, considering $\check{\ve{u}}_f$, or chooses $\ve{u}_f$, fulfilling the boundary conditions exactly. There is an intermediate alternative available, though: As done by Liu~\cite{Liu2009} in the original scheme, we apply Leray projection on the past velocities only, effectively skipping the $L^2$-projection step to obtain $\check{\ve{u}}_f$, and thereby fulfill the Dirichlet boundary conditions on the velocity $\ve{u}_f$ exactly. This technique is also referred to as divergence damping~\cite{Jia2011,Li2020} and has been shown to effectively reduce mass conservation errors and improve overall stability, while being cheaper than standard Leray projection and preserving boundary conditions on the velocity. Then, we construct the time-discrete weak form of the split-step scheme using BDF-$m$ schemes of the form
\begin{align}
	\dtale{\ve{u}_f
	\left(t^{n+1}
	\right)
	} 
	\approx
	\alpha_0^m\ve{u}_f^{n+1} + \sum_{j=1}^m \alpha_j^m 
	\ve{u}_f^{n+1-j}
	%\left( \ve{u}_f^{n+1-j} -  \nabla \psi^{n+1-j}\right)
	\,,
	\label{eqn:dtale_bdf}
\end{align}
and higher-order accurate extrapolation % of velocity and pressure at $t=t^{n+1}$,
\begin{align}
	\ve{u}_f^{n+1} \approx \ve{u}_f^{\star} = \sum_{j=1}^m\beta_j^m\ve{u}_f^{n+1-j}
	\,,
	%
%	\quad \text{and}
%	\quad
%	%
%	p_f^{n+1} \approx {p}_f^{\star} = \sum_{j=1}^m\beta_j^m{p}_f^{n+1-j}
%	\,,
	\label{eqn:extrap}
\end{align}
exemplarily shown for $\ve{u}_f$ with coefficients $\alpha_j^m$ and $\beta_j^m$ according to Table~\ref{tab:bdf_extrap_coeff}, to effectively decouple balance of linear momentum and pressure Poisson equations.
\begin{table}[!htbp]
	\centering
	\begin{tabular}{||r|c|c|c||}
		\hline
		&&&\\[-2ex]
		$j$ & 0 & 1 & 2 \\
		\hline\hline
		&&&\\[-1ex]
		$\alpha_j^m$ 
		& $\frac{2 \Delta t^{n} + \Delta t^{n-1}}{\Delta t^{n} (\Delta t^{n}+\Delta t^{n-1})}$ 
		& $- \frac{\Delta t^{n}+\Delta t^{n-1}}{\Delta t^{n} \Delta t^{n-1}}$  & $\frac{\Delta t^{n}}{\Delta t^{n-1}(\Delta t^{n}+\Delta t^{n-1})}$ \\[1.4ex]
		$\beta_j^m$  & $-$ 
		& $1+\frac{\Delta t^{n}}{\Delta t^{n-1}}$ & $\frac{\Delta t^{n}}{\Delta t^{n-1}}$
		\\[1.4ex]
		\hline
	\end{tabular}
	\caption{Backward differentiation and extrapolation coefficients, order $m=2$~\cite{Hairer1993}.}
	\label{tab:bdf_extrap_coeff}
\end{table}
Since for generalised Newtonian fluids the viscosity depends on the shear rate, which itself is a function of the velocity gradient, standard Lagrangian finite elements cannot be applied in a straight-forward way due to increased regularity requirements on the velocity interpolant. Therefore, the viscosity $\mu_f$ is introduced as an additional unknown and recovered through a simple $L^2$-projection. Thus, in the split-step scheme in time step $n$, first update the domain position $\ve{d}_f^{n+1}$ and compute the mesh velocities $\ve{u}_m$ using exactly the same \mbox{BDF-$m$} formula as for $\ve{u}_f$~\eqref{eqn:dtale_bdf}. Then, extrapolate known velocities, pressures and viscosities from previous time steps to obtain $\ve{u}_f^\star$, $p_f^\star$ and $\mu_f^\star$ via~\eqref{eqn:extrap} to linearise/decouple momentum balance and pressure Poisson equations. In contrast to the scheme presented in~\cite{Pacheco2021b}, momentum balance and PPE steps are executed in reversed order, which is motivated by observations made, showing that updating the velocity with an implicitly coupled pressure increases stability in semi-implicit schemes of FSI as discussed later.
So, we first project the pressure Dirichlet boundary conditions~\eqref{eqn:fluid_splitstep_BC_pressure_dirichlet_onGNf} and~\eqref{eqn:fluid_splitstep_BC_pressure_dirichlet_onGRf} on the respective boundary segments via
\begin{align}
\zeta^{n+1} 
&= 
-\mu_f^{n+1} \nabla \cdot \ve{u}_f^{n+1} 
+ \nf \cdot \left(2\mu_f^{n+1} \nabla^S\ve{u}_f^{n+1} \nf 
- \ve{t}_f^{n+1} \right)
&& \text{on } \GNf 
\label{eqn:final_pressure_projection_N}
\,,
\\
\zeta^{n+1} 
&= 
-\mu_f^{n+1} \nabla \cdot \ve{u}_f^{n+1} 
+ \nf \cdot \left(2\mu_f^{n+1} \nabla^S\ve{u}_f^{n+1} \nf 
- \ve{h}_f^{n+1} + \eta_f^R \ve{u}_f^{n+1}\right) 
&& \text{on } \GRf
\,,
\label{eqn:final_pressure_projection_R}
\end{align}
such that the resulting quantity $\zeta^{n+1}$ is continuous on the whole combined boundary segment \mbox{$\Gamma^t_p := \partial\Oft\setminus\GDf$}. This intermediate step is necessary, since the pressure Dirichlet condition would be discontinuous otherwise, but is fortunately negligible in terms of computational cost. A suitable weak form of Equation~\eqref{eqn:fluid_splitstep_ppe} to find $p_f\in H^1(\Oft)$, $p_f|_{\Gamma_p^t}=\zeta$ is obtained by multiplying with $\varphi \in H^1(\Oft)$, $\varphi|_{\Gamma_p^t} = 0$ and integrating by parts to obtain
\begin{align*}
\langle \nabla \varphi, \nabla p_f\rangle_{\Oft}
=
\langle \varphi, \nf\cdot\nabla p_f\rangle_{\GDf} 
+ \langle \varphi, \left[\nabla\times\left(\nabla \times \ve{u}_f\right)\right]\cdot\nabla \mu_f \rangle_{\Oft} 
%\\
%&
+ \langle \varphi, 
\nabla\cdot \left[\rho_f \nabla\ve{u}_f\left(\ve{u}_f - \ve{u}_m\right) - 2\nabla^S\ve{u}_f \nabla\mu_f \right]
\rangle_{\Oft}
\,,
\end{align*} 
where we can insert the pressure Neumann condition~\eqref{eqn:fluid_splitstep_BC_pressure_neumann} and integrate by parts again to get
\begin{align*}
\langle \nabla\varphi, \nabla p_f\rangle_{\Oft}
=&
- \langle \varphi, \rho_f \nf\cdot \dtale{\ve{u}_f}\rangle_{\GDf} 
- \langle \varphi \nf, \mu_f \nabla\times\left(\nabla\times\ve{u}_f\right)\rangle_{\GDf} 
\\
&
+
\langle\varphi,
\left[\nabla\times\left(\nabla \times \ve{u}_f\right)\right]\cdot\nabla \mu_f 
\rangle_{\Oft}
%\\
%&
+
\langle \nabla \varphi, 
2\nabla^S\ve{u}_f \nabla\mu_f
-
\rho_f \nabla\ve{u}_f\left(\ve{u}_f - \ve{u}_m\right)
\rangle_{\Oft}
\,.
\end{align*}
This is further simplified using
\begin{align*}
\langle 
\varphi \nf, \mu_f \nabla\times\left(\nabla\times\ve{u}_f\right) 
\rangle_{\GDf} 
= 
\langle 
\nabla \varphi, \mu_f \nabla\times\left(\nabla\times\ve{u}_f\right) 
\rangle_{\Oft}
+
\langle 
\varphi, \nabla \mu_f \cdot \left[\nabla\times\left(\nabla \times \ve{u}_f\right)\right]
\rangle_{\Oft}
\end{align*}
\begin{align*}
\text{and}\qquad\qquad
\langle 
\nabla\varphi,
\mu_f\nabla\times\left(\nabla \times \ve{u}_f\right)
\rangle_{\Oft}
=
&
\langle
\nabla \varphi \times \nf,
\mu_f \nabla \times \ve{u}_f
\rangle_{\partial\Oft}
+
\langle \nabla \times \left(\mu_f\nabla\varphi\right),\nabla\times\ve{u}_f\rangle_{\Oft}
\\
=
&
\langle
\nabla \varphi \times \nf,
\mu_f \nabla \times \ve{u}_f
\rangle_{\partial\Oft}
+
\langle \nabla \varphi , \left[ \nabla \ve{u}_f - (\nabla \ve{u}_f)^T\right] \nabla \mu_f \rangle_{\Oft}
\,,
\end{align*}
which leads then in the time-discrete case, also replacing $\dtale{\ve{u}_f}$ by a BDF-$m$ approximation~\eqref{eqn:dtale_bdf} to the Dirichlet boundary data $\ve{g}_f$ given on $\GDf$, to the problem of finding the pressure $p_f^{n+1}\in H^1(\Oft)$, such that $p_f^{n+1}|_{\Gamma^t_p}=\zeta^{n+1}$ for all $\varphi\in H^1(\Oft)$, $\varphi|_{\Gamma^t_p}=0$ and
\begin{align}
\langle \nabla\varphi, \nabla p_f^{n+1}\rangle_{\Oft} 
=&
- 
\langle 
\varphi \nf, 
\rho_f \sum_{j=0}^{m} \alpha_j^m \ve{g}_f^{n+1-j}
\rangle_{\GDf} 
+ \langle \nf\times\nabla\varphi,\mu_f^{n+1}\nabla\times\ve{u}_f^{n+1}\rangle_{\GDf}
\nonumber\\
&
+
\langle\nabla\varphi, 
2 \left(\nabla \ve{u}_f^{n+1}\right)^T\nabla\mu_f^{n+1}
- \rho_f \nabla \ve{u}_f^{n+1} (\ve{u}_f^{n+1}-\ve{u}_m^{n+1})
\rangle_{\Oft}
\,.
\label{eqn:final_ppe}
\end{align}
The weak form of momentum balance then reads: Find $\ve{u}_f^{n+1}\in[H^1(\Oft)]^d$, such that $\ve{u}_f^{n+1}|_{\GDf} = \ve{g}_f^{n+1}$ and
\begin{gather}
	\rho_f 
	\langle 
		\ve{\varphi}
		,
		\alpha_0^m\ve{u}_f^{n+1} 
		+
		\nabla \ve{u}_f^{n+1} \left( \ve{u}_f^\star - \ve{u}_m^{n+1} \right)
	\rangle_{\Oft}
	+
	\langle 
		\nabla \ve{\varphi}
	    ,
	    2 \mu_f^\star \nabla^S \ve{u}_f^{n+1}  
	\rangle_{\Oft}
	\nonumber
	\\
    = 
    \langle 
		\nabla \ve{\varphi}
	    ,
        p_f^\star \te{I}  
	\rangle_{\Oft}
    -
	\rho_f 
	\langle 
		\ve{\varphi}
		,
		\sum_{j=1}^m \alpha_j^m \left( \ve{u}_f^{n+1-j} -  \nabla \psi^{n+1-j}\right)
	\rangle_{\Oft} 
	+ \langle \ve{\varphi} , \ve{t}_f^{n+1}\rangle_{\GNf}
	+ \langle \ve{\varphi} , \ve{h}_f^{n+1} - \eta_f^R \ve{u}_f^{n+1}\rangle_{\GRf}
	\label{eqn:final_mom_bal_full_stress_form}
\end{gather}
for all $\ve{\varphi} \in [H^1(\Oft)]^d$, $\ve{\varphi}|_{\GDf}=\ve{0}$, with the divergence suppression applied to the old time step velocities via $\nabla\psi^{n+1-j}$. Here, traction conditions arise naturally from integrating the full stress divergence by parts. Linearising the convective term is a widely applied technique to improve efficiency in transient problems of incompressible flow for both coupled velocity-pressure formulations and pressure-/velocity projection or splitstep schemes \cite{Simo1994,TUREK1996,Elman2011,Ingram2013,Schussnig2021,Guermond2006,Liu2009,Pacheco2021b}. 
In the case of generalised Newtonian fluids, i.e., when the viscosity is not constant, the next step is to find the dynamic viscosity $\mu_f^{n+1} \in H^1(\Oft)$ given the current velocity $\ve{u}_f^{n+1}$, such that
\begin{align}
\langle \varphi , \mu_f^{n+1}\rangle_{\Oft} = \langle \varphi, \eta\left(\dot{\gamma}(\nabla\ve{u}_f^{n+1}) \right) \rangle_{\Oft}
\label{eqn:final_visc_proj}
\end{align}
for all $\varphi\in L^2(\Oft)$, which is a step that can simply be skipped in the Newtonian case, since the viscosity is constant, i.e., $\mu_f^{n+1} \equiv \eta_\infty$ holds at any point time. Also, the variable $\psi^{n+1}$ used for divergence supression is updated using $\ve{u}_f^{n+1}$ by solving the standard Poisson problem of finding $\psi^{n+1}\in H^1(\Oft)$ such that $\psi^{n+1} = 0$ on $\Gamma^t_p$ and
\begin{align}
	\langle \nabla \varphi , \nabla \psi^{n+1} \rangle_{\Oft} 
	= 
	\langle \varphi , \nabla \cdot \ve{u}_f^{n+1} \rangle_{\Oft}
	\label{eqn:final_leray_proj}
\end{align}
for all $\varphi \in H^1(\Oft)$, with $\varphi|_{\Gamma^t_p} = 0$ to apply divergence suppression on the current time step's velocity to be used in the next time step's momentum balance equation.

In summary, it is thus possible to construct a weak form containing only first-order derivatives, rendering our beloved $C^0$-continuous, standard Lagrangian finite elements applicable to the problem at hand. In fact, we might even employ equal-order finite element pairs for velocity and pressure as already pointed out. The presented weak forms contain a generous set of boundary conditions, which will in the coupled FSI-problem (partly) depend on the solid subproblem's solution as shall be seen next.

\section{The coupled FSI problem}
\label{sec:coupled_problem}
The strong form of the FSI problem incorporating mesh, fluid and solid subproblems as discussed in Section~\ref{sec:subproblems} including only interface conditions for brevity reads
\begin{align}
% MESH
-\nabla \cdot ( c \nabla \ve{d}_f) 
&= 0 
&&\text{in } \hOf 
\,,
\\
% SOLID
\rho_s \frac{\partial^2}{\partial t^2} \ve{d}_s
- \nabla \cdot \te{P} 
&= 0 
&& \text{in } \hOs 
\,,
\\
% FLUID MOM
\rho_f \left[
\dtale{\ve{u}_f}
+
\nabla \ve{u}_f \left(\ve{u}_f - \ve{u}_m \right)
\right]
-
%\mu_f \Delta \ve{u}_f
%-
%2\nabla^S\ve{u}_f\nabla\mu_f
\nabla\cdot \left( 2 \mu_f \nabla^S \ve{u}_f \right)
&= -\nabla p_f 
&&\text{in } \Oft
\,,
\\
% FLUID PPE
\nabla\cdot \left[\rho_f \nabla\ve{u}_f\left(\ve{u}_f - \ve{u}_m\right) - 2\nabla^S\ve{u}_f \nabla\mu_f \right]
+ \left[\nabla\times\left(\nabla \times \ve{u}_f\right)\right]\cdot\nabla \mu_f
&= -\Delta p_f 
&&\text{in } \Oft
\,,
\\
% INTERFACE COND
\ve{d}_f 
&= \ve{d}_s 
&&\text{on } \hInt
\,,
\label{eqn:fsi_cont_disp}
\\
\ve{u}_f
&= 
\frac{\partial}{\partial t} \ve{d}_s
&&\text{on }\Intt 
\,,
\label{eqn:fsi_cont_velo}
\\
J_s^{-1}\te{P}\te{F}_{\! s}^T \nf 
&= \te{\sigma}_f \nf
&& \text{on }\Intt
\,,
\label{eqn:fsi_cont_traction}
\end{align}
where~\eqref{eqn:fsi_cont_disp}--\eqref{eqn:fsi_cont_traction} enforce the continuity of displacements, velocities and tractions on the fluid--structure interface. Note that the continuity of tractions in Equation~\eqref{eqn:fsi_cont_traction} is formulated on $\Intt$, using the Cauchy stresses and the current configuration's normal vectors $\nf$, but is easily rewritten as
\begin{align}
\te{P}\hns = J_f \te{\sigma}_f \te{F}_{\! f}^{-T} \hns 
\quad 
\text{on } \hInt
\,,
\label{eqn:fsi_cont_traction_REF}
\end{align}
to enforce balance of tractions in the reference configuration. The Robin--Robin (RR) coupling conditions (see, e.g.,\cite{Badia2008,Nobile2008} or \cite{Astorino2010} in a projection-based semi-implicit scheme), linearly combine the interface conditions enforcing continuity of velocities and normal tractions, i.e., \eqref{eqn:fsi_cont_velo} and \eqref{eqn:fsi_cont_traction} or \eqref{eqn:fsi_cont_traction_REF}, yielding in the respective configurations
\begin{align}
	\eta_f^R \ve{u}_f
	+
	\te{\sigma}_f \nf
	&=
	\eta_f^R \frac{\partial}{\partial t} \ve{d}_s
	+ 
	J_s^{-1}\te{P}\te{F}_{\! s}^T\nf 
	&&\text{on }\Intt
	\,, 
	\label{eqn:fsi_robin_f}
	\\
	\eta_s^R \frac{\partial}{\partial t} \ve{d}_s
	+ 
	\te{P}\hns 
	&=
	\eta_s^R \ve{u}_f
	+
	J_f \te{\sigma}_f \te{F}_{\! f}^{-T}\hns
	&&\text{on }\hInt
	\,,
	\label{eqn:fsi_robin_s}
\end{align}
with Robin parameters $\eta_f^R,\eta_s^R > 0$. This type of interface condition leads to a coupling algorithm with good convergence properties even in the case of high added-mass effects (cf. \cite{Astorino2010,Badia2008,Gerardo-Giorda2010,Nobile2012}), which is of particular importance in biomedical applications \cite{Causin2005,Forster2007,Kuttler2010}. As the basic algorithm, we perform an implicit single-loop coupling scheme (see, e.g., \cite{Kuttler2010,Matthies2003,Nobile2013,Nobile2014}), which is executed until convergence criteria of the form
\begin{gather}
	\frac{||\fevec{d}_s^{k+1} - \fevec{d}_s^k||}{||\fevec{d}_s^{k+1}||} 
	< 
	\epsilon_{rel}
	\quad
	\text{and}
	\quad
	||\fevec{d}_s^{k+1} - \fevec{d}_s^k|| 
	< 
	\epsilon_{abs}
	\label{eqn:fsi_convergence_criteria}
\end{gather}
are fulfilled. In the following, we will denote the last iterates by a superscript $k$, and the newly computed iterate by a superscript $k+1$.
Moreover, we directly present the RR scheme, which is obtained inserting interface conditions into the Robin terms of the respective subproblems. So, at each time step $n$, given the solutions from previous time steps $\ve{d}_f^{n}$, $\ve{d}_f^{n-1}$, $\ve{d}_s^{n}$, $\ve{d}_s^{n-1}$, $\dot{\ve{d}}_s^{n}$, $\ddot{\ve{d}}_s^{n}$, $\ve{u}_f^{n}$, $\ve{u}_f^{n-1}$, $\mu_f^{n}$, $\mu_f^{n-1}$, $p_f^{n}$ and $p_f^{n-1}$, the resulting coupling algorithm reads

\vspace{5mm}
\hrule
\begin{enumerate}
	%
	% LERAY PROJECTION ON PAST time step
	%
	\item \textit{Divergence suppression:} Update the Leray projection variable of the past time step's fluid velocity $\ve{u}_f^n$, $\psi^{n}\in H^1(\Oft)$, such that $\psi^{n}|_{\Gamma^t_p} = 0$ and
	\begin{align}
	\langle \nabla \varphi,\nabla\psi^{n}\rangle_{\Oft}
	= 
	\langle \varphi,\nabla\cdot\ve{u}_f^{n}\rangle_{\Oft}
	\qquad
	\forall \varphi\in H^1(\Oft),
	\text{ with } \varphi|_{\Gamma^t_p}=0.
	\label{eqn:fsi_leray}
	\end{align}
	%	
	% EXTRAP AND INIT
	%
	\item \textit{Extrapolation/initial guess:} 
	Compute $\ve{d}_s^\star$, $\mu_f^\star$, $\ve{u}_f^\star$ and $p_f^\star$ based on old time step solutions via~\eqref{eqn:extrap} and set $\ve{d}_s^{k}=\ve{d}_s^{\star}$, $\ve{u}_f^{k}=\ve{u}_f^{\star}$ and $p_f^{k}= p_f^{\star}$ as initial guess.
	
	\item \textit{Implicit coupling loop:}\\
	\textbf{WHILE} not converged according to Equation~\eqref{eqn:fsi_convergence_criteria} \textbf{DO}	
	\begin{enumerate}
		%
		% MESH
		%
		\item \textit{Mesh subproblem:} Update the domain $\Oft$ by
			finding $\ve{d}_f^{n+1}\in \revised{[H^1(\hOf)]^d}$, such that $\ve{d}_f^{k+1} = \ve{d}_s^k$ on $\hInt$, $\ve{d}_f^{k+1} = \ve{0}$ on $\partial\hOf\setminus\hInt$ and  	
			\begin{align}
			\langle c (\ve{d}_f^n) \nabla \ve{\varphi}, \nabla \ve{d}_f^{k+1} \rangle_{\hOf} 
			= 
			0
			\,
			\qquad
			\forall \ve{\varphi} 
			\in \revised{[H^1(\hOf)]^d}, \text{ with } \ve{\varphi}|_{\partial\hOf} = \ve{0}
			\label{eqn:fsi_mesh}
			\end{align}
			and linearised stiffening parameter $c(\ve{d}_f^n)$.
		%
		% MESH VELO
		%
		\item \textit{Mesh velocity update:}
			Compute $\ve{u}_m^{k+1}$ via the BDF-$m$ formula~\eqref{eqn:dtale_bdf}.
		%
		% VISC PROJ
		%
		\item \textit{Viscosity projection:} Find $\mu_f^{k+1} \in H^1(\Oft)$, such that
		\begin{align}
		\langle \varphi , \mu_f^{k+1}\rangle_{\Oft} = \langle \varphi, \eta\left(\dot{\gamma}(\nabla\ve{u}_f^{k}) \right) \rangle_{\Oft}
		\qquad 
		\forall \varphi\in L^2(\Oft)
		\,.
		\label{eqn:fsi_visc_proj}	
		\end{align}
		%
		% PRESS PROJ
		%
		\item \textit{Pressure boundary projection: }Update the pressure Dirichlet condition by projecting $\zeta^{k+1}$ 
		on $\GNf$ using 
		%on $\Gamma^t_p:=\partial\Oft\setminus\GDf$ using
%		on $\Gamma^t_p := \partial \Oft \setminus \GDf = \GpseudoNf \cup \Intt$ using
		\begin{align}
		%
%		\zeta^{k+1}|_{\Intt}
%		=& 
%		-\mu_f^{k} \nabla \cdot \ve{u}_f^{k} 
%		+ \nf \cdot \left(
%		2\mu_f^{k} \nabla^S\ve{u}_f^{k} \nf 
%		- \ve{h}_f^{k} + \eta_f^R \ve{u}_f^{k}
%		\right) 
%		\,,
%		\nonumber \\
		\zeta^{k+1}|_{\GNf} 
		=& 
		-\mu_f^{k+1} \nabla \cdot \ve{u}_f^{k} 
		+ \nf \cdot \left(
		2\mu_f^{k+1} \nabla^S\ve{u}_f^{k} \nf 
		- \ve{t}_f^{n+1}
		\right)
		\,.
		\label{eqn:fsi_press_proj}
		\end{align}
%		where $\ve{h}_f^{n+1} = \eta_s^R \dot{\ve{d}}_s^{n+1} + J_s^{-1}\te{P}\te{F}^T\nf$ is evaluated with the last iterate $\ve{d}_s^k$.
		%
		% PPE
		%
		\item \textit{Pressure Poisson step:} Find $p_f^{k+1}\in H^1(\Oft)$, such that $p_f^{k+1}|_{\GNf} = \zeta^{k+1}$ and
		\begin{align}
		\langle \nabla\varphi, \nabla p_f^{k+1}\rangle_{\Oft} 
		= 
		& 
		\langle\nabla\varphi, 
		2 \left(\nabla \ve{u}_f^{k}\right)^T\nabla\mu_f^{k+1}
		- \rho_f \nabla \ve{u}_f^{k} (\ve{u}_f^{k}-\ve{u}_m^{k+1})
		\rangle_{\Oft}
		\nonumber\\
		&
		+
		\langle \nf\times\nabla\varphi,\mu_f^{k+1}\nabla\times\ve{u}_f^{k}
		\rangle_{\Intt\cup\GDf}
		-
		\langle 
		\varphi \nf, 
		\rho_f \sum_{j=0}^{m} \alpha_j^m \ve{g}_f^{n+1-j}
		\rangle_{\GDf}
		-
		\langle 
		\varphi \nf, 
		\rho_f \ddot{\ve{d}}_s^{n+1}
		\rangle_{\Intt}
		\label{eqn:fsi_ppe}
		\end{align}
		holds for all $\varphi\in H^1(\Oft)$ with $\varphi|_{\GNf}=0$, using the last solid iterate $\ve{d}_s^k$ to compute $\ddot{\ve{d}}_s^{n+1}$.
		%
		% SOLID MOM
		%
		\item \textit{Solid momentum:} Solve the nonlinear solid momentum balance equation via Newton's method, where in each step $l$, \revised{the discrete finite element vectors} $\fevec{d}_s^{l+1} = \fevec{d}_s^l + \Delta\fevec{d}_s^l$ are updated. The increment $\Delta \fevec{d}_s^l$ \revised{is the discrete vector of nodal unknowns approximating} $\delta\ve{d}_s^l \in [H^1(\hOs)]^d$, for which $\delta \ve{d}_s^l|_{\hGDs}=\ve{0}$ and
		\begin{align}
		% dtt_d term
		&\frac{\rho_s\alpha_m'}{\beta \Delta t^2}
		\langle \ve{\varphi}, \delta\ve{d}_s^l\rangle_{\hOs}
		% stress term
		+ \alpha_f' \langle \nabla \ve{\varphi}, 
		\frac{\partial}{\partial\ve{d}_s}%\mathcal{D}_{\ve{d}_s} 
		\te{P}\left(\ve{d}_s^{l}\right) \rangle_{\hOs}
		% Robin term
		+ \eta_s^R\frac{\alpha_f'\gamma}{\beta\Delta t} \langle \ve{\varphi}, \delta\ve{d}_s^l \rangle_{\hInt}
		\nonumber\\
		=&
		%
		% stress Neumann boundary terms
		\langle \ve{\varphi}, \alpha_f' {\ve{t}}_s^{n+1} + \alpha_f {\ve{t}}_s^n \rangle_{\hGNs} 
		% stress Robin boundary terms
		+ \langle \ve{\varphi}, 
		\alpha_f' \left[ {\ve{h}}_s^{n+1} 
		- \eta_s^R \dot{\ve{d}}_s^{n+1} \left(\ve{d}_s^l\right) 
		\right] 
		+ \alpha_f \te{P}\left(\ve{d}_s^n\right) \hns \rangle_{\hInt} 
		\nonumber\\
		%
		% time term
		&- \rho_s \langle \ve{\varphi}, 
		\alpha_m'\ddot{\ve{d}}_s^{n+1} (\ve{d}_s^l)
		+ \alpha_m \ddot{\ve{d}}_s^n
		\rangle_{\hOs}
		% stress domain terms
		- \langle \nabla \ve{\varphi} , \alpha_f' \te{P}\left(\ve{d}_s^{l}\right)
		+ \alpha_f \te{P}\left(\ve{d}_s^n\right) \rangle_{\hOs}
		\label{eqn:fsi_solid_mom}
		\end{align}
		holds for all $\ve{\varphi}\in[H^1(\hOs)]^d$ with $\ve{\varphi}|_{\GDf} = \ve{0}$, using the last iterate $\ve{d}_s^l$ to evaluate the time derivatives $\ddot{\ve{d}}_s^{n+1}$, $\dot{\ve{d}}_s^{n+1}$ and Robin data \mbox{$\ve{h}_s^{n+1} = \eta_s^R \ve{u}_f^{k} + J_f \te{\sigma}_f (\ve{u}_f^k, p_f^{k+1}, \mu_f^{k+1}) \te{F}_{\! f}^{-T}\hns$}. Set $\fevec{d}_s^{k+1} = \fevec{d}_s^{l+1}$, once the relative convergence criterion  $||\fevec{d}_s^{l+1}-\fevec{d}_s^l||<\epsilon_N ||\fevec{d}_s^k||$ is fulfilled.
		%
		%
		% FLUID MOM
		%
		\item \textit{Fluid momentum:} Solve the linearised momentum equation in $\Oft$, i.e., find $\ve{u}_f^{k+1}\in[H^1(\Oft)]^d$, such that $\ve{u}_f^{k+1}|_{\GDf} = \ve{g}_f^{n+1}$ and
		\begin{gather}
		\rho_f 
		\langle 
    		\ve{\varphi}
    		,
    		\alpha_0^m\ve{u}_f^{k+1}
    		+
    		\nabla \ve{u}_f^{k+1} \left( \ve{u}_f^k - \ve{u}_m^{k+1} \right)
		\rangle_{\Oft}
		+
		\langle
		    \nabla\ve{ \varphi}
		    ,
		    2 \mu_f^{k+1} \nabla^S\ve{u}_f^{k+1}
	    \rangle_{\Oft}
		\nonumber
		\\
		=
		\langle\nabla\ve{\varphi}, p_f^{k+1} \te{I} \rangle_{\Oft}
		- \rho_f 
		\langle 
		\ve{\varphi}
		,
		\sum_{j=1}^m \alpha_j^m \left( \ve{u}_f^{n+1-j} -  \nabla \psi^{n+1-j}\right)
		\rangle_{\Oft} 
		+ \langle \ve{\varphi} , \ve{t}_f^{n+1}\rangle_{\GNf}
		+ \langle \ve{\varphi} , \ve{h}_f^{n+1} - \eta_f^R \ve{u}_f^{k+1} 
		\rangle_{\Intt}
		\, 
		\label{eqn:fsi_fluid_mom}
		\end{gather}
		for all $\ve{\varphi} \in \revised{[H^1(\Oft)]^d}$, with $\ve{\varphi}|_{\GDf} = \ve{0}$, and
		using the Leray projection acting on the past time step's fluid velocities via $\psi^{n+1-j}$
		and the updated Robin condition \mbox{$\ve{h}_f^{n+1} = \eta_f^R \dot{\ve{d}}_s^{n+1} + J_s^{-1}\te{P} (\ve{d}_s^{k+1})\te{F}^T_s\nf$}. 
		%
		%\item \textit{Aitken acceleration:} ...
	\end{enumerate}			
	\textbf{END DO}
	\item \textit{Update time step data:} Set $\ve{d}_f^{n+1}=\ve{d}_f^{k+1}$, $\ve{d}_s^{n+1}=\ve{d}_s^{k+1}$,
	$\ve{u}_f^{n+1}=\ve{u}_f^{k+1}$,
	${p}_f^{n+1}={p}_f^{k+1}$
	and
	${\mu}_f^{n+1}=\mu_f^{k+1}$.
\end{enumerate}
\hrule
\vspace{5mm}
Note here that the treatment of Robin boundary conditions in the momentum balance and PPE steps is not simply assigning Robin boundary conditions to the fluid subproblem, but rather enforcing Robin conditions on $\Intt$ in the fluid momentum equation and treating the interface as a Dirichlet boundary for all steps related to the fluid pressure. This combination is equivalent to the strategy adopted by \cite{Astorino2010} and is herein solely based on numerical observations. The sequence of fluid steps and viscosity projection turned out to be the most stable choice when confronted with large time steps and sudden jumps in fluid pressure boundary conditions as present in the pressure pulse benchmark in Section~\ref{sec:examples_pressure_pulse}.

The RR coupling algorithm as introduced above includes the standard Dirichlet--Neumann (DN) coupling scheme in the asymptotic limit, when $\eta_f^R\rightarrow\infty$ and $\eta_s^R = 0$. In the discrete setting, however, $\eta_f^R$ has to be assigned a bounded value, which motivates including the interface Dirichlet condition on the fluid velocity in a more direct way. Thus, in the DN case, the interface is treated as part of the fluid Dirichlet boundary $\Intt\subset\GDf$ together with setting $\eta_s^R=0$, effectively leading to small changes in the function space definitions only, but not introducing any additional terms. To counteract decreased convergence for high added-mass effects, Aitken's acceleration is applied to relax the discrete solution vector $\tilde{\fevec{x}}^{k+1}$ in iteration $k$ with a recursively defined $\omega_k$ \cite{Kuttler2008}
\begin{gather}
\fevec{x}^{k+1} = \omega_k \tilde{\fevec{x}}^{k+1} + (1-\omega_k) \fevec{x}^k
\quad
\text{with}
\quad
\omega_k = -\omega_{k-1} 
\frac{\fevec{r}^{k}\cdot(\fevec{r}^{k+1} - \fevec{r}^k)}
	 {|| \fevec{r}^{k+1} - \fevec{r}^k ||^2}
\nonumber\\
\text{and}
\quad
\fevec{r}^{k+1} = \tilde{\fevec{x}}^{k+1}-\fevec{x}^k  
\,,
\label{eqn:aitken}
\end{gather}
where $||\cdot||$ denotes the Euclidean norm. \revised{Clearly, more advanced acceleration schemes (see, e.g., \cite{Gerbeau2003,Michler2005,Fernandez2005,Spenke2020}) might be employed, but herein we restrict the discussion to the reportedly effective Aitken relaxation for the sake of brevity.}
Another option to increase efficiency drastically is to give up on fully implicit coupling of all the involved subproblems, but rather settling for a semi-implicit variant of the scheme. This option is directly accessible having formulated the fully implicit algorithms by simply moving the mesh and fluid momentum subproblems and the viscosity projection out of the coupling loop, similar to the methods proposed by~\cite{Fernandez2007,Astorino2010,Breuer2012,Quaini2007,Lozovskiy2015,He2015,Naseri2018} for Newtonian fluids and~\cite{Amani2020} considering a non-Newtonian, viscoelastic fluid. Interestingly, the sequence of substeps in the fluid phase had to be changed in order to yield satisfying results. First, the fluid pressure and solid displacement are implicitly coupled, before explicitly treating the fluid balance of linear momentum and viscosity projection rather than performing the update on $\ve{u}_f$ before the implicit loop. This way, the dependence on an extrapolated, non-coupled pressure is eliminated, which results in improved robustness of the scheme. An additional tuning possibility is available via the convergence criterion in the solid's Newton scheme, where one may tweak the tolerance $\epsilon_N$ or even exit after a fixed number of steps in the nonlinear solver in the spirit of \cite{LANGER2015b,Nobile2014}. Using suitable higher-order extrapolation schemes, temporal accuracy is preserved, while the fully implicit coupling of pressure and structural displacements is sufficient to obtain a stable method as numerically observed and proven for simplified model problems (cf.~\cite{GRANDMONT2001,Fernandez2007,Nobile2013,Lozovskiy2015}). In a nutshell, the following distinct features and benefits arise in the proposed scheme when compared to related methods:
\begin{enumerate}[(i)]
	\item Higher-order and possibly adaptive time-stepping schemes are available based on standard time integration and extrapolation formulae.
	\item Equal-order finite element pairs can be employed, which would be unstable in the classical coupled velocity-pressure formulation of the Navier--Stokes equations.
	\item Exchanging the rheological model of the generalised Newtonian fluid is as simple as changing the right-hand side of the projection step.
	\item The semi-implicit design reduces computing times tremendously, while preserving stability properties and accuracy. Only structural displacements and fluid pressure are iteratively coupled in each time step.
	\item Robin interface conditions might improve convergence even for high added-mass effects when suitable parameters are available and standard acceleration methods are directly applicable. 
	\item Divergence suppression avoids the accumulation of errors in mass conservation and neither spoils interface conditions nor requires a velocity projection step.
	\item All linear systems can be effectively tackled using off-the-shelf black-box preconditioning techniques available as \text{open-source} scientific software. 
\end{enumerate} 

\section{Computational results}
\label{sec:computational_results}
This section is devoted to the thorough testing of the presented schemes in terms of accuracy and robustness as well as critically comparing their individual performance. All of the showcased results were obtained with the finite element toolbox \texttt{deal.II} \cite{dealII92}, solving each of the arising linear systems involved in the FSI-algorithm~\eqref{eqn:fsi_leray}--\eqref{eqn:fsi_solid_mom} iteratively. We employ algebraic multigrid methods provided by Trilinos' ML package~\cite{Heroux2012} for preconditioning each linear solve. A preconditioned conjugate gradient method is used for the mass matrices in the viscosity projection step~\eqref{eqn:fsi_visc_proj}, the pressure Dirichlet data projection~\eqref{eqn:fsi_press_proj} and also for the Poisson problems, i.e., the PPE~\eqref{eqn:fsi_ppe} and the mesh motion equation~\eqref{eqn:fsi_mesh}. The linear systems corresponding to fluid and solid momentum balance equations are solved adopting a flexible generalised minimal residual method.

The studied test cases are two analytical solutions taken from \cite{Serino2019,Serino2019b} to demonstrate convergence rates numerically, a classical benchmark of a pressure pulse travelling a straight pipe in three spatial dimensions (see, e.g., \cite{Formaggia2001,Janela2010,Langer2018}) and, finally, we study the flow through an idealised abdominal aortic aneurysm to showcase performance in a practically relevant setting.

\subsection{Analytical solution: rectangular piston}

An analytical solution is taken from \cite{Serino2019,Serino2019b}, which describes the periodic motion of a linear elastic piston in vertical direction. \revised{The Newtonian fluid simply follows the motion of the solid and freely exits/enters the computational domain over all boundaries in order to fulfil the incompressibility constraint.} The computational domain $\hat{\Omega} := \hOf\cup\hOs = [0,L]\times[-H,H]$, where $L=1$ and $H=0.5$, is depicted in Figure~\ref{fig:domain_rect_piston}, with the fluid initially occupying the region $\hat{x}_2\geq0$ and the undeformed solid in $\hat{x}_2\leq0$.
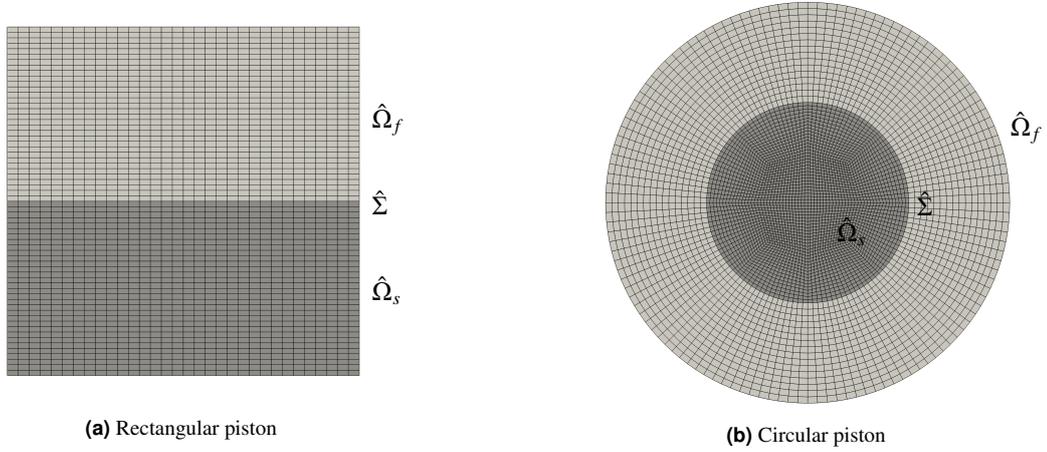
\begin{figure}[!htbp]
	\begin{minipage}{.49\linewidth}
		\centering
		\subfloat[Rectangular piston]{
			\label{fig:domain_rect_piston}
			\includegraphics[scale=0.02]{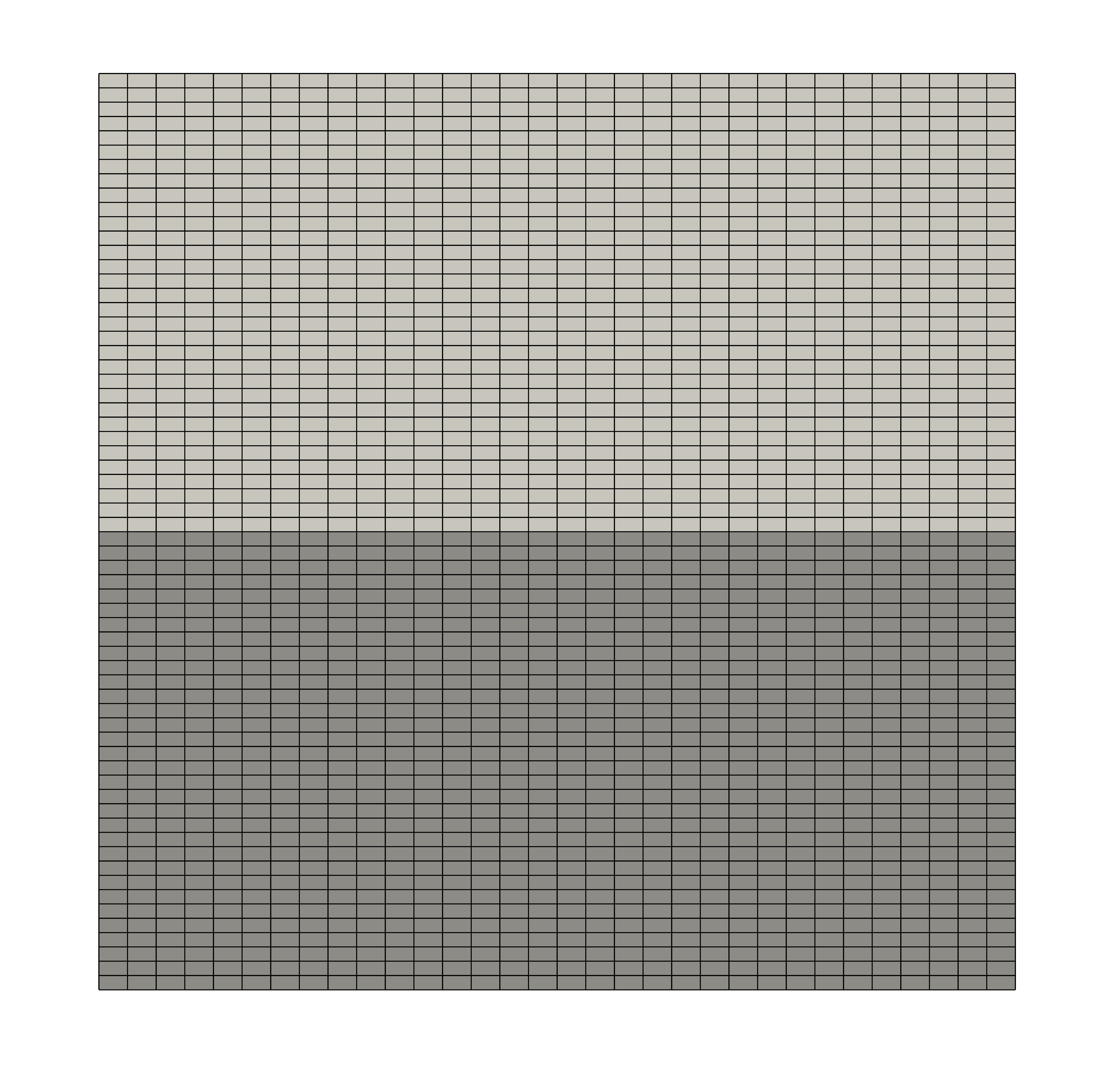}
			\put(-10,40){{$\hOs$}}
			\put(-10,105){{$\hOf$}}
			\put(-10,72.5){$\hInt$}
		}
	\end{minipage}%
	\hfil
	\begin{minipage}{.49\linewidth}
		\centering
		\subfloat[Circular piston]{
			\label{fig:domain_circ_piston}
			\includegraphics[scale=0.02]{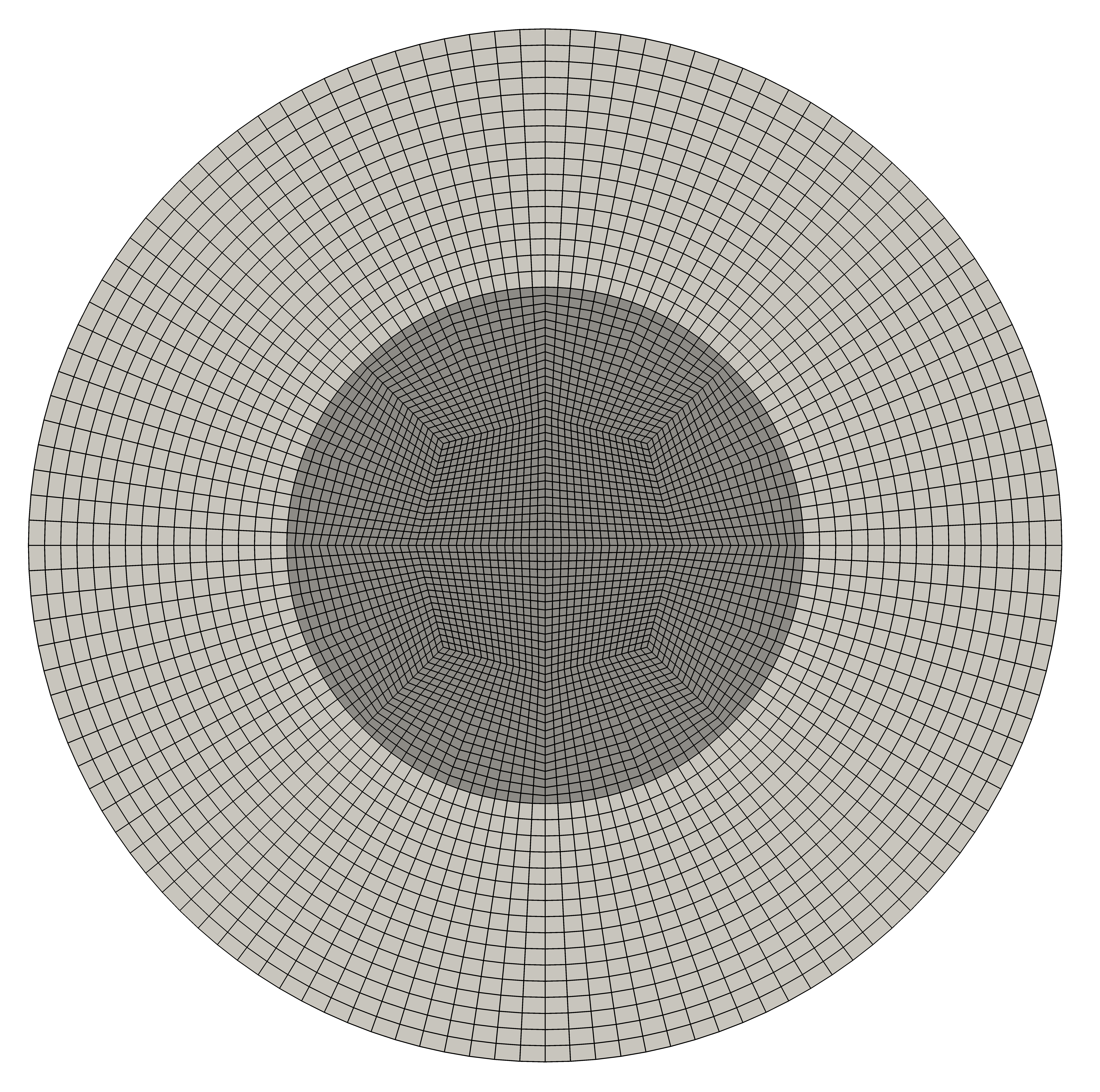}
			\put(-70,65){{$\hOs$}}
			\put(-5,105){{$\hOf$}}
			\put(-40,75){$\hInt$}
		}
	\end{minipage}
	\caption{Finite element meshes at refinement level 4 for analytical solutions \cite{Serino2019,Serino2019b}.}
	\label{fig:domain_rect_circ_piston}
\end{figure}
The exact solution is defined assuming zero displacements and velocities in horizontal direction and prescribing the vertical component of the interface displacement $\hat{d}_{\Sigma,2}$ as 
\begin{align}
	\hat{d}_{\Sigma,2} := a \sin(\omega t)
	\qquad
	\text{with}
	\quad
	a := 2 \alpha \sin\left(\frac{\omega H}{c_p}\right)
	,
	\quad
	c_p = \sqrt{\frac{\lambda_s+2\mu_s}{\rho_s}}	
	,
	\label{eqn:analyt_cp_def}
\end{align}
which oscillates vertically with amplitude $a$ and frequency $\omega$. The structural displacements are thus given by
\begin{align*}
	\hat{d}_{s,2} (\hat{x}_2, t)
	= 
	{f}\left(t-\frac{\hat{x}_2 + H}{c_p}\right)
	-
	{f}\left(t+\frac{\hat{x}_2 + H}{c_p}\right)
	\qquad
	\text{with }
	\quad
	{f}(\tau) = \alpha \cos(\omega \tau)
	\,.
\end{align*}
The fluid's vertical velocity resulting from the continuity equation is solely dependent on time, $u_{f,2}(t) = \frac{\partial}{\partial t} \hat{d}_{\Sigma,2}(t)$, and the pressure is given by
\begin{gather*}
	p_f (x_2,t) = 
	\frac{(H-x_2) p_{\Sigma} + (x_2 - \hat{d}_{\Sigma,2}) p_H}{H-d_{\Sigma,2}} 
	%\left[
	%	(H-x_2) p_{\Sigma} + (x_2 - \hat{d}_{\Sigma,2}) p_H
	%\right]
	\qquad
	\text{with}
	\quad
	p_{\Sigma} = -
	(\lambda_s+2\mu_s)\frac{\partial}{\partial \hat{x}_2}
	\hat{d}_{s,2}(0,t)
	\\
	\text{and}
	\quad
	p_H = - \rho_f \left[ H-\hat{d}_{s,2}(0,t) \right]
	\frac{\partial^2}{\partial t^2} \hat{d}_{s,2}(0,t)
	+ (\lambda_s+2\mu_s)
	\frac{\partial}{\partial \hat{x}_2} \hat{d}_{s,2}(0,t)	
	.
\end{gather*}
We refer to the original publications \cite{Serino2019,Serino2019b} for details on the derivation and proceed in defining problem parameters. The fluid density and dynamic viscosity are set to \mbox{$\rho_f = 1 \text{~kg/m$^3$}$} and \mbox{$\mu_f \equiv \eta_\infty = 0.1 \text{~Pa~s}$}, respectively. For the linear elastic solid we assign the density $\rho_s = 100 \text{~kg/m$^3$}$, the Young's modulus $E_s = 5~\text{kPa}$ and a Poisson's ratio of $\nu_s=0.3$. Additionally, we choose $a=0.005$ and $\omega=\pi$ to prescribe the piston motion and enforce Dirichlet conditions on all exterior boundaries of the solid domain, $\hGDs = \partial \hOs \setminus \hInt$. In terms of boundary conditions for the fluid, we prescribe Neumann conditions at $x_1=0$ and $x_1=L$ and Dirichlet conditions at $x_2=H/2$. Convergence rates are measured in the maximum $L^2$-error over all time steps $n=1,...,N_t$ for $\ve{u}_f$, $p_f$ and $\ve{d}_s$ defined as
\begin{align*}
	e_{\ve{u}_f} := \max_{n=1,...,N_t} \left\{ || \ve{u}_f-\ve{u}_f^h ||_{L^2(\Oft)} \right\}
	\, , \quad
	e_{p_f} := \max_{n=1,...,N_t} \left\{ || p_f-p_f^h ||_{L^2(\Oft)} \right\}
	\, , \quad
	e_{\ve{d}_s} := \max_{n=1,...,N_t} \left\{ || \ve{d}_s-\ve{d}_s^h ||_{L^2(\hat{\Omega}_s)} \right\}
	\, ,
\end{align*}
and compared to the estimated order of convergence ($eoc$) indicated by triangles in the plots. For the spatial discretisation, equal-order $Q_1/Q_1$ elements are employed, meaning that $d$-linear shape functions are used for velocities and displacements in both fluid and solid together with $d$-linear elements for the fluid viscosity $\mu_f$, fluid pressure $p_f$, the variable $\psi$ used for divergence suppression and corresponding traces for $\zeta$. For $t\in (0,0.5]$ we choose uniform time steps and the second-order scheme, i.e., using BDF-$2$, linear extrapolation and the generalised-$\alpha$ schemes with parameters set according to Table~\ref{tab:gen_alpha_alphas}. Regarding the coupling scheme, we focus first on the classical DN approach with Aitken's relaxation, implicitly coupling fluid and solid phase. To disentangle the various variants, we introduce them layer by layer and investigate thoroughly the consequences of each change, aiming for the most efficient overall scheme. Starting off, we compare different settings in the generalised-$\alpha$ time integrators. 

When refining the time step, expected convergence rates are observed in all primary variables when using the Newmark-$\beta$ scheme, as can be seen in Figure~\ref{fig:temp_Nb}. 
\begin{figure}[!htbp]
	\centering
	\includegraphics[width=0.7\textwidth]{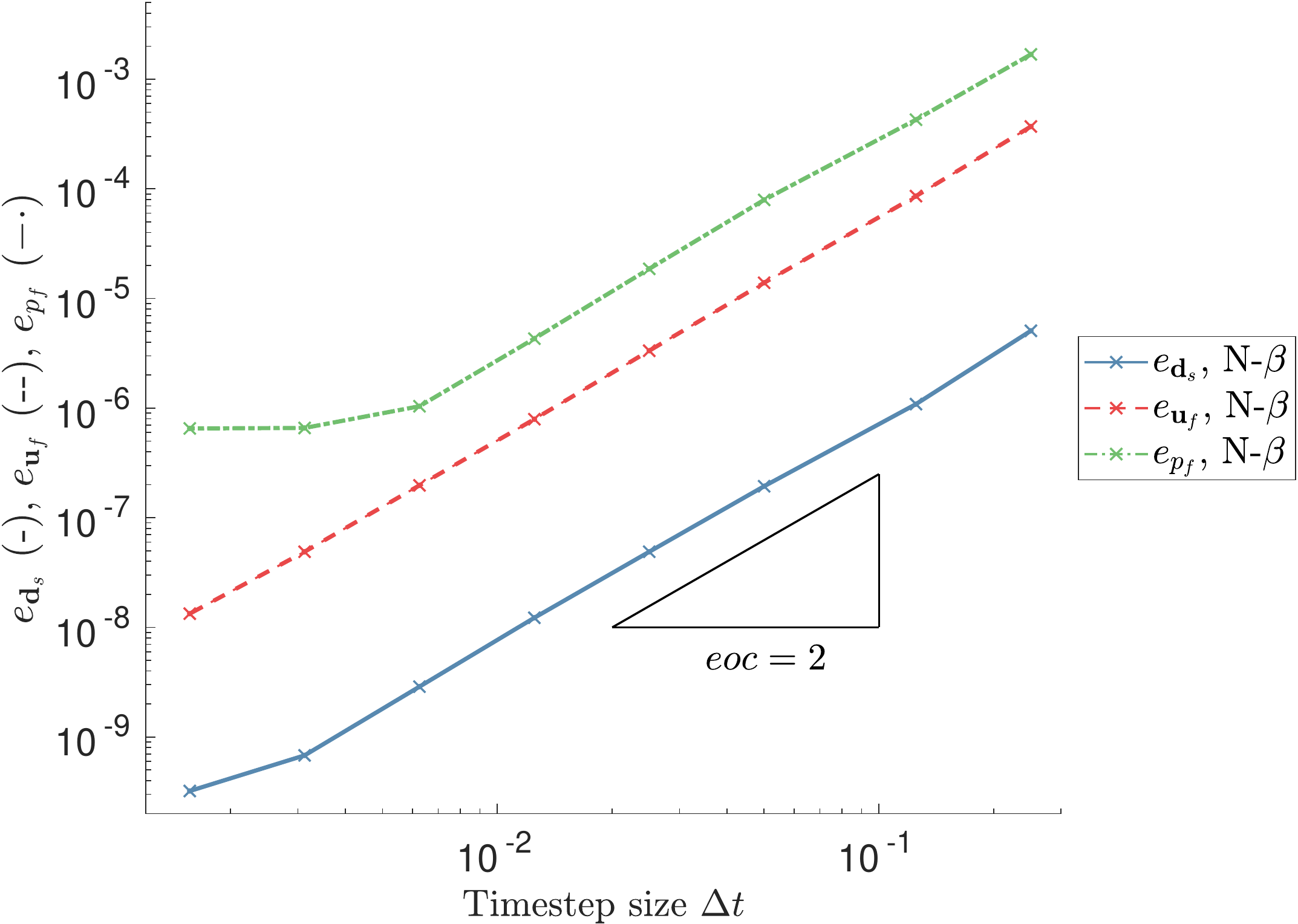}
	\caption{Newmark-$\beta$ time integration with implicit DN coupling yields the expected second-order convergence in time in solid displacements, fluid velocities and pressure.}
	\label{fig:temp_Nb}
\end{figure}
However, when introducing numerical high-frequency dissipation via the generalised-$\alpha$ time integration scheme with a spectral radius in the high frequency limit $\rho_\infty\neq 1$, an increase in the saturation error as exemplarily shown in Figure~\ref{fig:temp_GenA_0c98} for $\rho_\infty=0.98$ is observed. Choosing a practically relevant (user-specified) high-frequency dissipation/spectral radius in the high-frequency limit $\rho_\infty$, second-order convergence in velocities, displacements and pressure are maintained.
\begin{figure}[!htbp]
	\centering
	\includegraphics[width=0.7\textwidth]{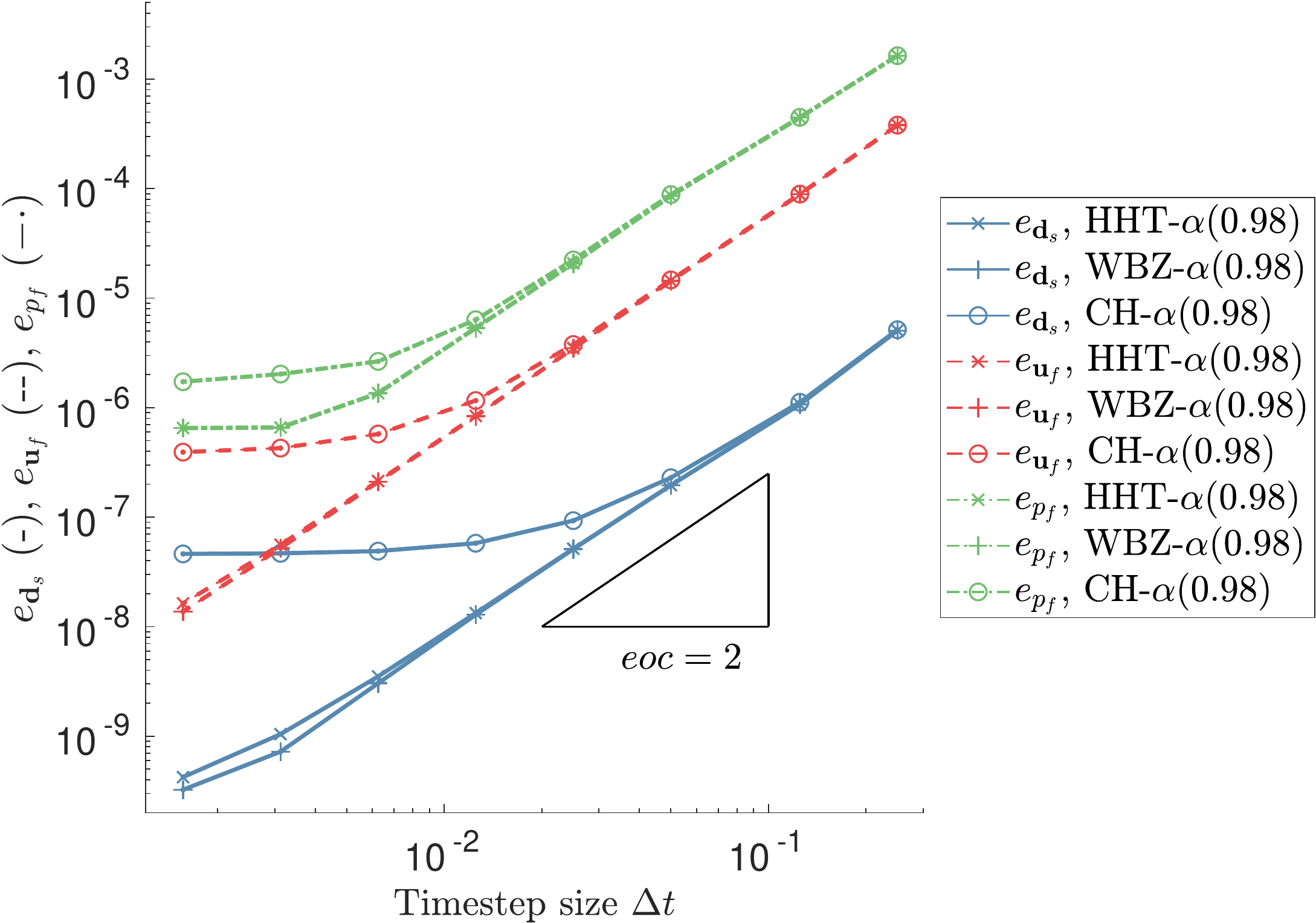}
	\caption{Generalised-$\alpha$ time integration with implicit DN coupling: Choosing a spectral radius in the high frequency limit $\rho_\infty=0.98$ results in an increased saturation error compared to $\rho_\infty=1.0$ or Newmark-$\beta$ time integration.}
	\label{fig:temp_GenA_0c98}
\end{figure}
A value of $\rho_\infty=1$ yields identical results to the N-$\beta$ scheme for the HHT-$\alpha$ and WBZ-$\alpha$, whereas the so-called asymptotic annihilation case, i.e., $\rho_\infty=0$, leads to a clearly linear convergence rate in $p_f$ for WBZ-$\alpha$ and CH-$\alpha$ methods as can be seen in Figure~\ref{fig:temp_GenA_special} (and is not admissible for the HHT-$\alpha$ scheme). Again, the earlier error saturation in $e_{\ve{d}_s}$ and $e_{\ve{u}_f}$ is most clearly observed for the CH-$\alpha$ scheme. The decreased convergence rate in the fluid pressure $p_f$ is linked to the term $\langle \varphi \ve{n}_f, \rho_f \ddot{\ve{d}}_s^{n+1}(\ve{d}_s^k)\rangle_{\Intt}$ appearing in the PPE~\eqref{eqn:fsi_ppe}, which is for the generalised-$\alpha$ scheme only first order accurate in time if $\alpha_f \neq \alpha_m$~\cite{Erlicher2002}.
\begin{figure}[!htbp]
	\centering
	\includegraphics[width=0.7\textwidth]{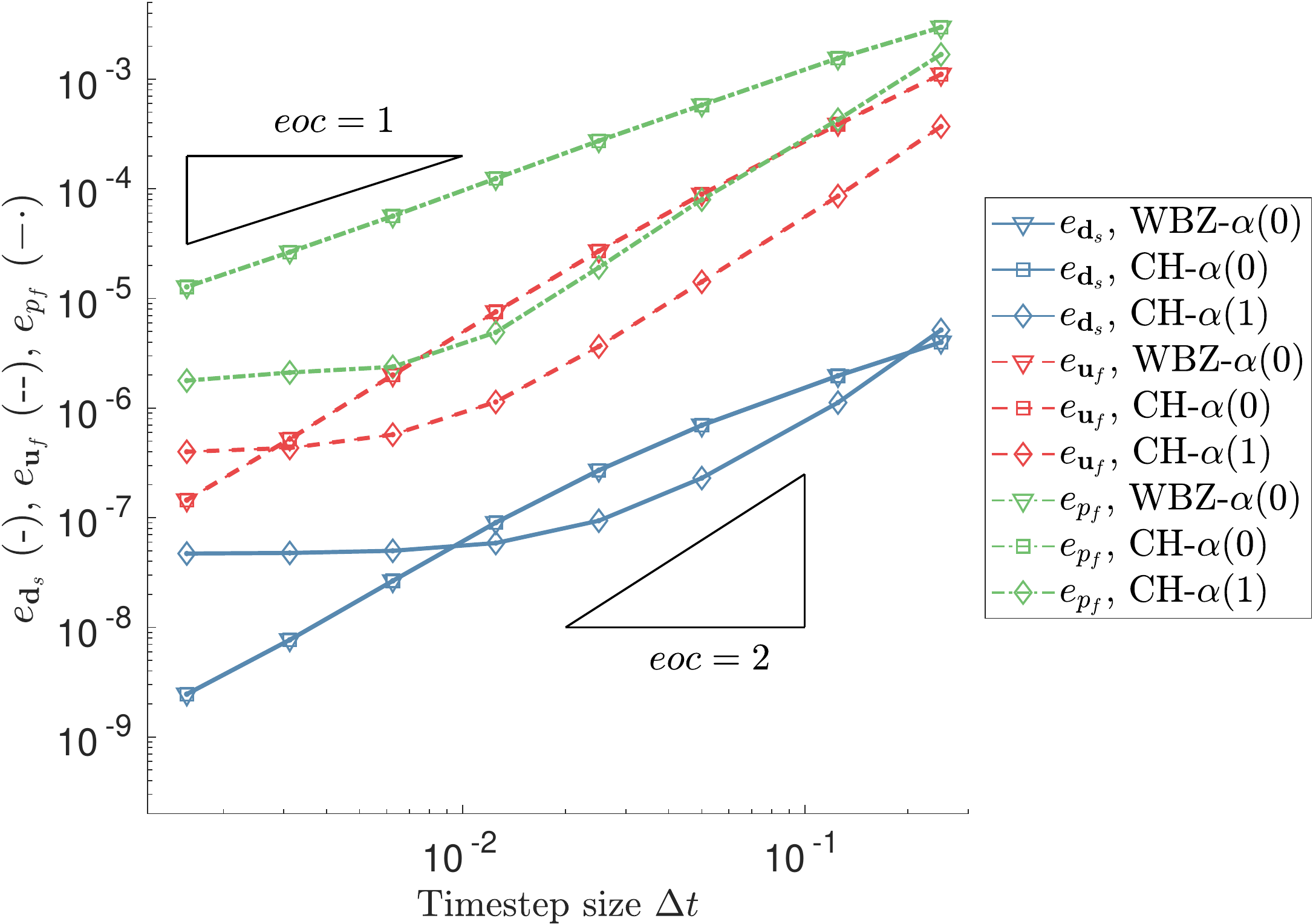}
	\caption{Generalised-$\alpha$ time integration with implicit DN coupling and various $\rho_\infty$: Increased algorithmic damping as $\rho_\infty \rightarrow 0$ results in an increased saturation error and reduces the temporal order of convergence of the fluid pressure from the optimal 2 down to 1.}
	\label{fig:temp_GenA_special}
\end{figure}
The fully implicit variant of the RR scheme yields similar results and is thus not further discussed at this point.

Regarding the semi-implicit Dirichlet--Neumann (SIDN) and Robin--Robin (SIRR) coupling schemes, we report the expected convergence rates in Figure~\ref{fig:temp_SI}, treating the mesh motion equation, fluid momentum equation, viscosity projection and divergence suppression explicitly. Additionally, taking the fully implicit N-$\beta$ scheme as a baseline, saturation errors in the fluid pressure increase for both semi-implicit variants. Using the RR coupling scheme, the saturation errors of the fluid pressure and velocity depend on the Robin parameters $\eta_s^R$ and $\eta_f^R$~\cite{Badia2008}
\begin{align}
	\eta_f^R = \frac{\rho_s H_s}{\Delta t} + \beta \Delta t
	\quad
	\text{and}
	\quad
	\eta_s^R = \frac{\rho_f}{\Delta t} \gamma \mu_{max}
	\,,
\label{eqn:robin_parameters}
\end{align}
where $H_s = 1$, $\beta = 0$ and $\gamma \mu_{max} = 0.01$ were chosen, the latter of which was found uncritical in this example. Considering the substantial increase in efficiency, the semi-implicit schemes are very attractive for practical applications as indicated at various places in the literature \cite{Fernandez2007,Astorino2010,Breuer2012,Quaini2007,Lozovskiy2015,He2015,Naseri2018}.
\begin{figure}[!htbp]
	\centering
	\includegraphics[width=0.7\textwidth]{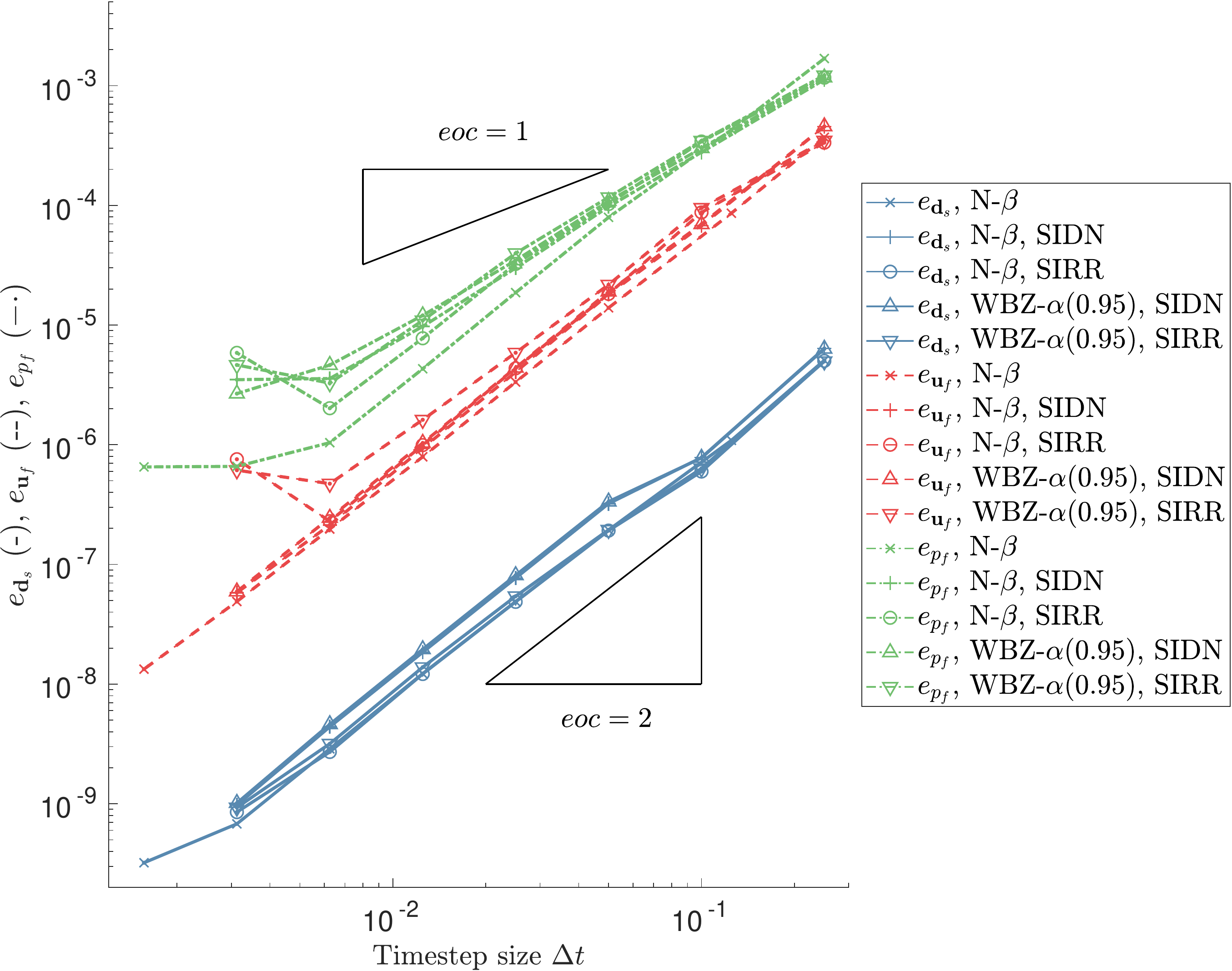}
	\caption{Newmark-$\beta$ and WBZ-$\alpha$ $(\rho_\infty=0.95)$ time integration: The semi-implicit SIDN and SIRR schemes yield almost identical errors compared to the fully implicit DN coupling independent of the time integration scheme applied.}
	\label{fig:temp_SI}
\end{figure}
\subsection{Analytical solution: circular piston}
Another analytical solution is employed to confirm the expected convergence rates in space. A circular piston as depicted in Figure~\ref{fig:domain_circ_piston}, where the structure occupies the region from $\hat{r} = 0$ to the interface at $\hat{r}_{\Sigma} = r_\Sigma (t=0)$ in the reference configuration, pulsates in radial direction, driving the fluid. \revised{The fluid adheres to the piston's periodic motion and since the time-dependent fluid domain ranging from $r\geq r_{\Sigma}(t)$ to $r = R$ changes in volume, the fluid freely exits and enters the computational domain due to incompressibility.} For a detailed derivation see \cite{Serino2019,Serino2019b}, whereas herein we conveniently express the solution in terms of the radial component of structure displacement
\begin{align*}
	\hat{\ve{d}}_{s,r} (\hat{r}, t)
	:= 
	\beta \mathcal{J}_1\left( \frac{\omega \hat{r}}{c_p}\right) \sin(\omega t)
	\,,
	\quad
	\text{such that }
	r_\Sigma (t) = r_\Sigma^0 + \hat{\ve{d}}_{s,r}(r_\Sigma^0,t)
\end{align*}
with frequency $\omega$, a parameter $\beta$ scaling the amplitude, initial piston radius $r_\Sigma^0$, $c_p$ as defined in Equation~\eqref{eqn:analyt_cp_def} and $\mathcal{J}_1$ denoting the Bessel function of the first kind and order one, which gives the fluid's radial velocity component
\begin{gather*}
	\ve{u}_{f,r}(r,t) = \frac{R}{r}V(t)
	\,,
	\quad
	\text{with}
	\quad
	V(t) = \frac{\omega \beta}{R} r_\Sigma (t) \mathcal{J}_1 \left(\frac{\omega r_\Sigma^0}{c_p}\right) \cos (\omega t)
	\,,
\end{gather*}
and the fluid pressure
\begin{align*}
	p_f(r,t) 
	=& 
	P(t) 
	+ 
	\frac{\rho_f}{2} \left[1-\left(\frac{R}{r}\right)^2\right] V(t)^2 
	+ 
	\rho_f R \log \left(\frac{R}{r}\right) \frac{\partial}{\partial t} V(t)
	\,,
	\\
	\text{with }
	P(t) 
	=& 
	-
	\frac{\rho_f}{2} \left[1-\left(\frac{R}{r_\Sigma (t)}\right)^2\right] V(t)^2
	-
	\rho_f R \log \left(\frac{R}{r_\Sigma(t)}\right) \frac{\partial}{\partial t} V(t)
	\\
	&
	- \beta \sin(\omega t) \left[
		(\lambda_s + 2\mu_s) \frac{\omega}{c_p} \mathcal{J}_1^\prime \left(\frac{\omega r_\Sigma^0}{c_p}\right) 
		+ 
		\frac{\lambda_s}{r_\Sigma^0} \mathcal{J}_1 \left(\frac{\omega r_\Sigma^0}{c_p}\right)
		\right]
	\,,
\end{align*}
where $\mathcal{J}_1^\prime$ denotes the first derivative of $\mathcal{J}_1$. The piston with initial radius $r_\Sigma^0=0.5$ pulsates with frequency $\omega = \pi$ and $\beta = 0.1$ is set. Fluid's density and dynamic viscosity are considered as $\rho_f=1\text{~g/m$^3$}$ and $\mu_f = \eta_\infty = 0.5\text{~mPa~s}$, respectively. The linear-elastic solid has a density of $\rho_s = 1\text{~kg/m$^3$}$, a Young's modulus of $E_s=100\text{~kPa}$ and a Poisson's ratio of $\nu_s=0.3$. To minimise the influence of time integration error, the time interval from $t=0$ to $t=0.05~\text{s}$ is divided into $80$ equal time steps of constant length $\Delta t = 0.625~\text{ms}$ and integrated using second-order time-stepping schemes and extrapolation.

The experimental orders of convergence are reported in Figure~\ref{fig:space_all}, allowing for a direct comparison of the errors obtained with $Q_2/Q_1$ and $Q_1/Q_1$ elements and various coupling schemes, where the Robin parameters are again set according to~\eqref{eqn:robin_parameters}. The deliberately chosen parameters result in directly recovering the exact solution of $\ve{d}_s$ up to the specified tolerance/time integration error, but allow easily measuring experimental convergence rates of fluid velocity and pressure. The fluid velocities converge at a rate of $eoc=2$ for the fully implicit and semi-implicit schemes with almost identical errors obtained. This is exactly as expected, even for the $Q_2/Q_1$ finite element pairing (cf. \cite{Liu2009,Johnston2004,Pacheco2021b}).
Pressure rates of order $1$ are observed for linear pressure interpolation and a slightly higher rate of $eoc\approx1.5$ when using $Q_2/Q_1$ interpolation. These results are optimal for the considered split-step scheme in the fluid pressure, where one might expect convergence rates of order $1$ using (bi-)linear elements in the observed norm. However, the measured rates of $\approx 1.5$ for the $Q_2/Q_1$ pair are lower than what one would hope for, judging from the basic split-step scheme, which gives rates of $2$ in the pressure norm observed here. These higher rates are found initially, but decrease under refinement. Additionally, the solid subproblem yielded optimal convergence rates when tested using different parameter settings, which is omitted for the sake of brevity. The sole reason for not choosing those different parameter settings altogether is the inherent difficulty in demonstrating all convergence rates at the same time, which was not found possible here due to the problem setup and relative sizes of physical quantities and tolerance choices.
\begin{figure}[!htbp]
	\centering
	\includegraphics[width=0.7\textwidth]{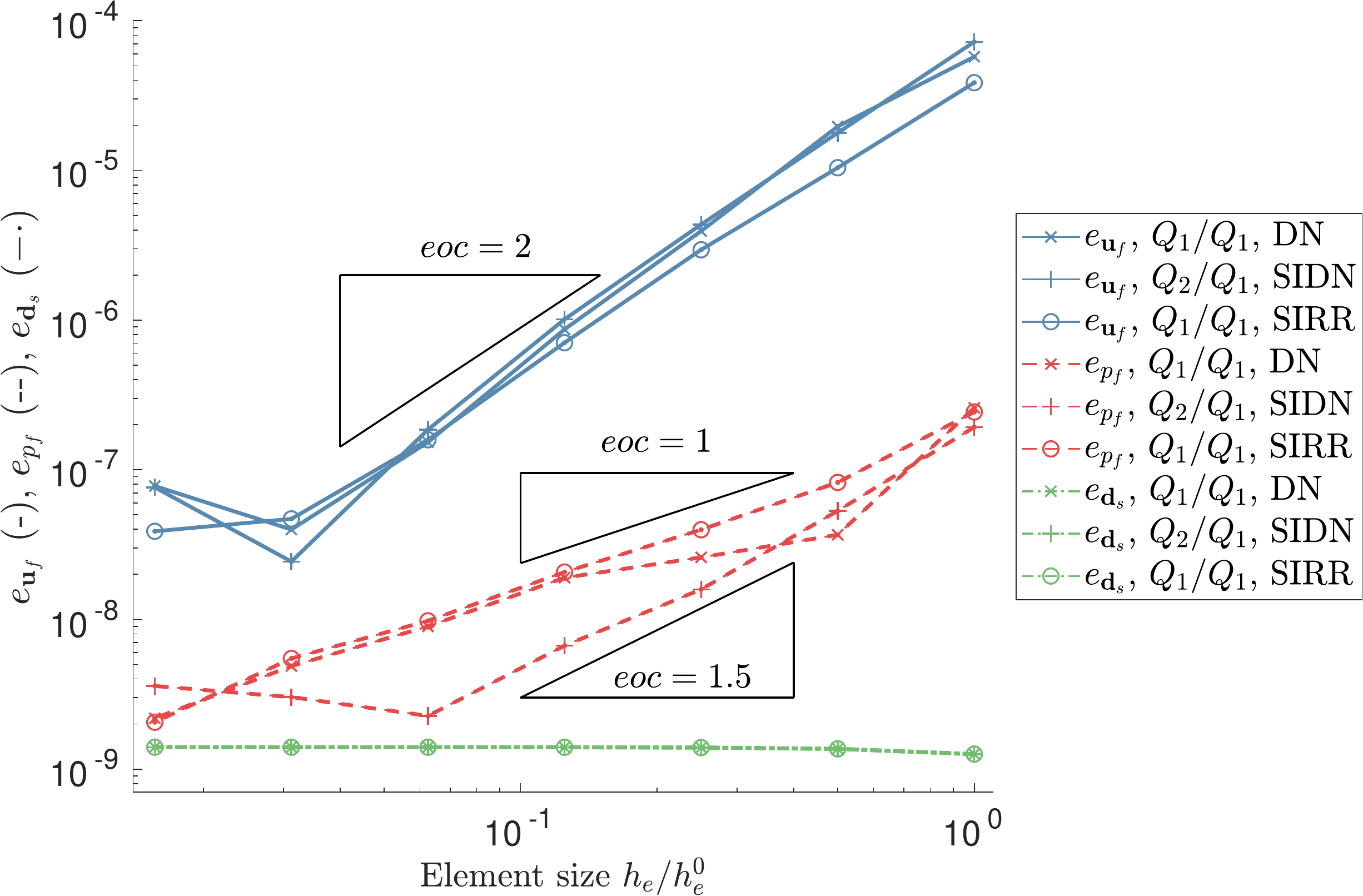}
	\caption{Spatial convergence rates obtained using $Q_1/Q_1$ or $Q_2/Q_1$ interpolation in combination with the DN coupling scheme or semi-implicit variants SIDN and SIRR.}
	\label{fig:space_all}
\end{figure}

\subsection{Pressure pulse benchmark}
\label{sec:examples_pressure_pulse}

To investigate the influence of different material models on the coupling algorithm, we consider a variant of the well-established benchmark example of a pressure pulse traveling through a straight flexible tube as used in \cite{Janela2010} based on \cite{Formaggia2001,Gerbeau2003}. The parameter choice is inspired by hemodynamic applications, fixing the tube length of $l=5~\text{cm}$ and inner radius of $r_\text{i}=0.5~\text{cm}$. The surrounding solid is divided into two layers $\hat{\Omega}_{s,1}$, $\hat{\Omega}_{s,2}$ of equal thickness $h_s = 0.05~\text{cm}$ indicated by different colors in Figure~\ref{fig:domain_tube_layers}, resulting in an outer radius of $r_\text{o}=0.6~\text{cm}$. 

\begin{figure}[!htbp]
	\begin{minipage}{.48\linewidth}
		\centering
		\subfloat[Fluid mesh with boundary layers and cut solid domains resolving medial and adventitial tissue layers.]{
			\label{fig:domain_tube_layers}
			\includegraphics[width=0.97\textwidth]{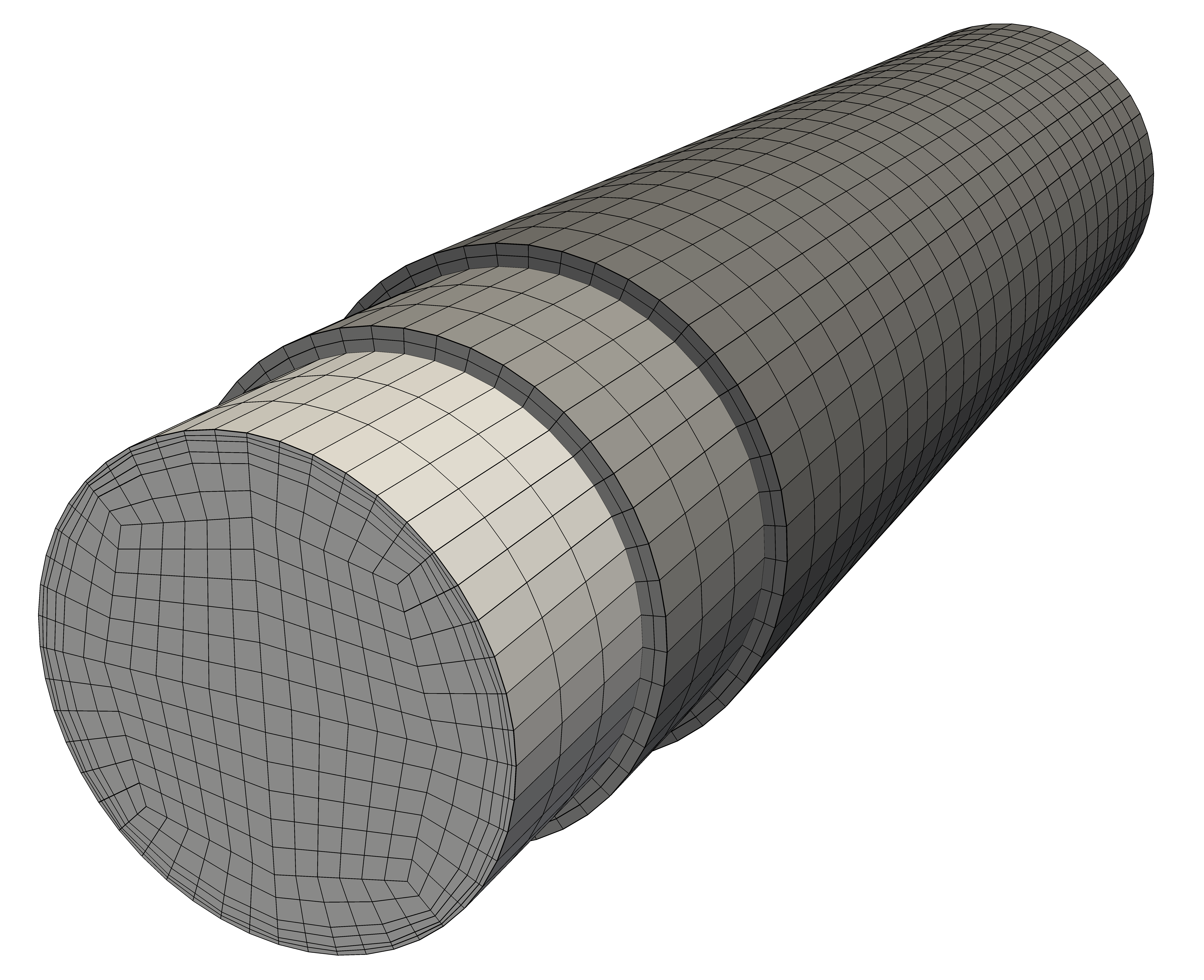}
			\put(-140,160){{$\hat{\Omega}_{s,2}$}}
			\put(-180,135){{$\hat{\Omega}_{s,1}$}}
			\put(-220,110){{$\hOf$}}
			%\put(-10,72.5){$\hInt$}
		}
	\end{minipage}%
	\hfil
	\begin{minipage}{.48\linewidth}
		\centering
		\subfloat[Radial (blue) and longitudinal (yellow) orientation vectors in cut solid domains used to construct $\hat{\ve{e}}_1$ and $\hat{\ve{e}}_2$.]{
			\label{fig:domain_tube_ori}
			\includegraphics[width=0.97\textwidth]{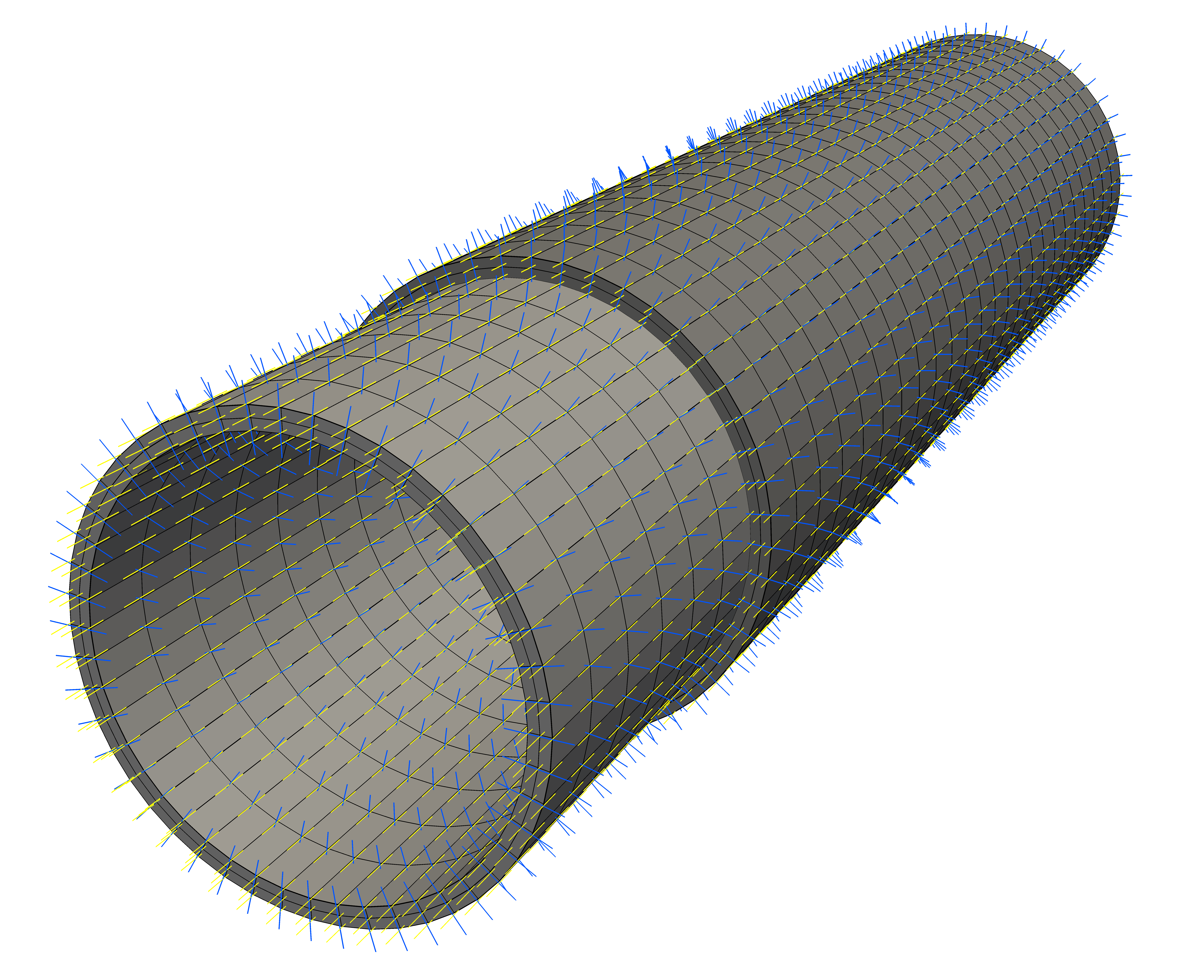}
			%\put(-70,65){{$\hOs$}}
			%\put(-5,105){{$\hOf$}}
			%\put(-40,75){$\hInt$}
		}
	\end{minipage}
	\caption{Finite element mesh considered for the pressure pulse benchmark. }
	\label{fig:domain_tube}
\end{figure}

A pressure pulse is generated setting $\te{\sigma}_f \ve{n}_f = -p_\mathrm{in} \ve{n}_f$ at the inlet, with \mbox{$p_\mathrm{in} = 10~\text{mmHg}\approx1333.22~\text{Pa}$} for the first $9~\text{ms}$ and zero otherwise. On the outflow boundary, zero Neumann conditions are enforced. The tube is fixed at both ends, and zero traction conditions are prescribed at the solid's external boundary at $r=r_o$. Effects of external tissue support and downstream vasculature are neglected at this point, but can be introduced easily using lumped parameter models \cite{Formaggia2001,ARBIA2016,Moireau2012,CROSETTO2011,REYMOND2013,EsmailyMoghadam2011,Bertoglio2013}. Concerning the material parameters, we set the fluid and solid densities to $\rho_s = 1200~\text{kg/m$^3$}$ and $\rho_f = 1060~\text{kg/m$^3$}$, which triggers strong added-mass effects. Further, we consider a Newtonian fluid with a viscosity of \mbox{$\mu_f = 3.5~\text{mPa~s}$}, which needs to be suitably chosen (see, e.g., \cite{Cho1991,Janela2010,Ranftl2021}) for comparison to the more general Carreau model with $\eta_0=56~\text{mPa~s}$, $\eta_\infty=3.45~\text{mPa~s}$, $\lambda_f=3.313~\text{s}$ and \mbox{$n=0.3568$} taken from \cite{Kim2000}.
For the solid phase, linear elasticity~\eqref{eqn:solid_const_rel_linelast} or a St.~Venant--Kirchhoff solid~\eqref{eqn:solid_const_rel_StVenant}, both with Young's modulus $E_s = 300~\text{kPa}$ and Poisson ratio $\nu_s = 0.3$ are compared to layered models of neo-Hookean material with and without additional fiber reinforcement (Equations~\eqref{eqn:solid_const_rel_qiNH} and~\eqref{eqn:solid_const_rel_qiHGO2006}). The latter two choices use $\nu_s=0.499$ and shear rates of $\mu_{s,1} = 62.1~\text{kPa}$ and $\mu_{s,2} = 21.6~\text{kPa}$ for the inner and outer layers, respectively. \revised{The fiber orientation is computed from radial and longitudinal orientation vectors as depicted in Figure~\ref{fig:domain_tube_ori}. These local systems are constructed by solving two auxiliary scalar Laplace problems with suitable boundary conditions in the solid domain as discussed in Section \ref{sec:subproblem_solid}. The two fiber families are oriented relative to the circumferential orientation vectors, rotated by $\pm \alpha_c$ into the longitudinal direction.} Fiber parameters are chosen as
$k_1=1.4~\text{kPa}$, $k_2=22.1$ and $\kappa_{c,1}=0.12$, $\alpha_{c,1}=27.47^\circ$ or $\kappa_{c,2}=0.25$, $\alpha_{c,2}=52.88^\circ$ in inner and outer layers according to \cite{Rolf-Pissarczyk2021,Weisbecker2012,Schussnig2021PAMMb}. 

We use a uniform time step of $\Delta t=0.5~\text{ms}$ in the second-order accurate scheme, i.e., BDF-2 and linear extrapolation as initial guess or possible linearisation. The WBZ$-\alpha$ and CH$-\alpha$ time integrators with $\rho_\infty=0$ are selected to counteract pressure oscillations in time caused by the jump in the Neumann condition and simultaneously give the lowest iteration counts. Aitken's relaxation is initiated with $\omega_0 = 0.01$, coupling the subproblems until reaching $\epsilon_{abs}=10^{-7}$ or $\epsilon_{rel} = 10^{-4}$, while a relative Newton tolerance of $\epsilon_N = 10^{-3}$ in~\eqref{eqn:newton_rel_tol} was found sufficient. 

Snapshots of the travelling pulse are shown in Figure~\ref{fig:pipe_pulse_snapshots}, where the HGO and Carreau models were employed. As opposed to the traditional benchmark, the hemodynamic-inspired setting features larger displacements given the lower material stiffness, as can be seen in Figure~\ref{fig:pipe_compare_const_d}. However, despite the maximum displacement being $\approx30$~\% of the tube's thickness (and also only $\approx0.6$\% of the tube's length), strains are small enough for linear elasticity~(E) and St.~Venant--Kirchhoff~(SVK) solid to yield almost identical values for both the observed quantities (see Figures~\ref{fig:pipe_compare_const_d}~and~\ref{fig:pipe_compare_const_p}). The fiber contribution present in the HGO model is already visible, but not dominant given the fiber parameters and small strains. Comparing the different rheological laws used, one can see large variations in viscosity at any point in time (see Figure~\ref{fig:pipe_pulse_snapshots}), but negligible effects on displacements and pressure.
Additionally, we observe ``backflow'' instabilities as a consequence of fluid entering over the inlet Neumann boundary due to mass conservation. Despite those oscillations close to the inlet, the solver still converges. Possible remedies to stabilise such effects are backflow stabilisation or similar techniques \cite{ARBIA2016,Vignon-Clementel2010,EsmailyMoghadam2011,Formaggia2001,Formaggia2003,Attaran2018}, which are not considered at this point.

\begin{figure}[!htbp]
	\begin{minipage}{.48\linewidth}
		\centering
		\subfloat[$\ve{d}_s$ and $\ve{u}_f$ at $t=15~\text{ms}$.]{
			\label{fig:pipe_d_u_step14}
			\includegraphics[width=0.97\textwidth]{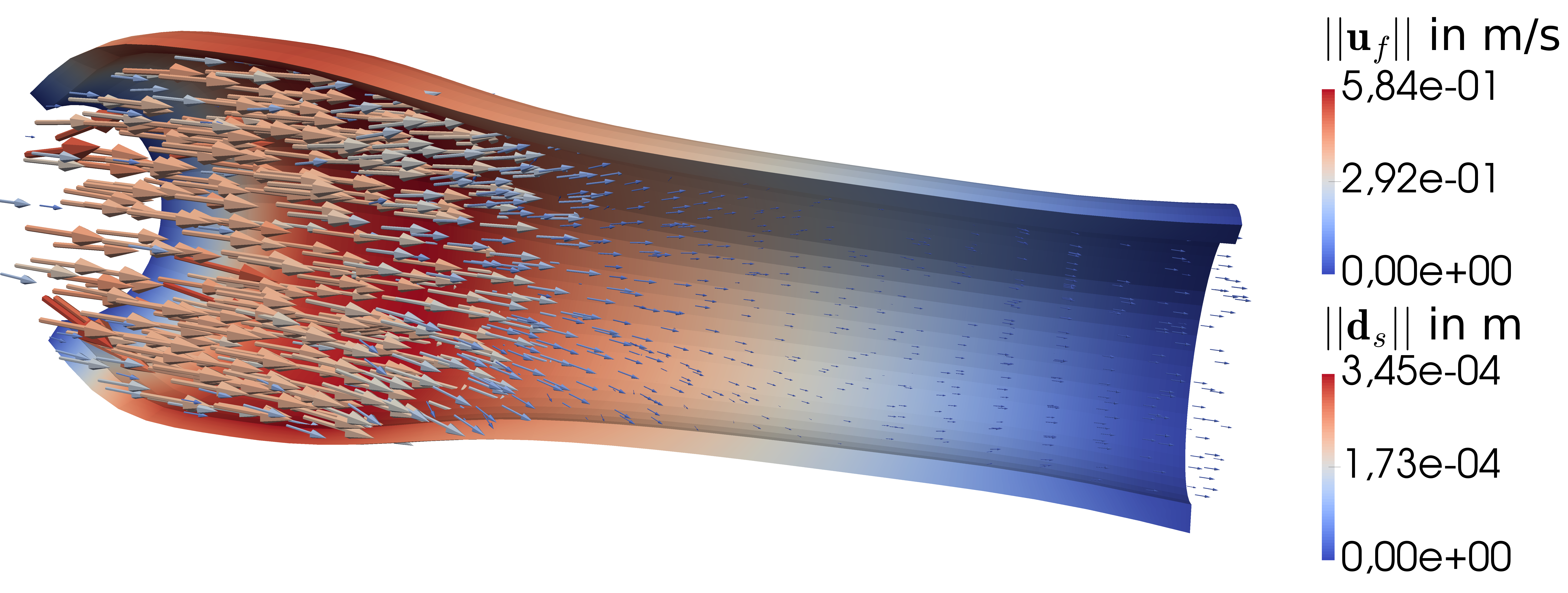}
			%\put(-10,72.5){$\hInt$}
		}
	\end{minipage}%
	\hfil
	\begin{minipage}{.48\linewidth}
		\centering
		\subfloat[$p_f$ (top) and $\mu_f$ (bottom) at $t=15~\text{ms}$.]{
			\label{fig:pipe_p_mu_step14}
			\includegraphics[width=0.97\textwidth]{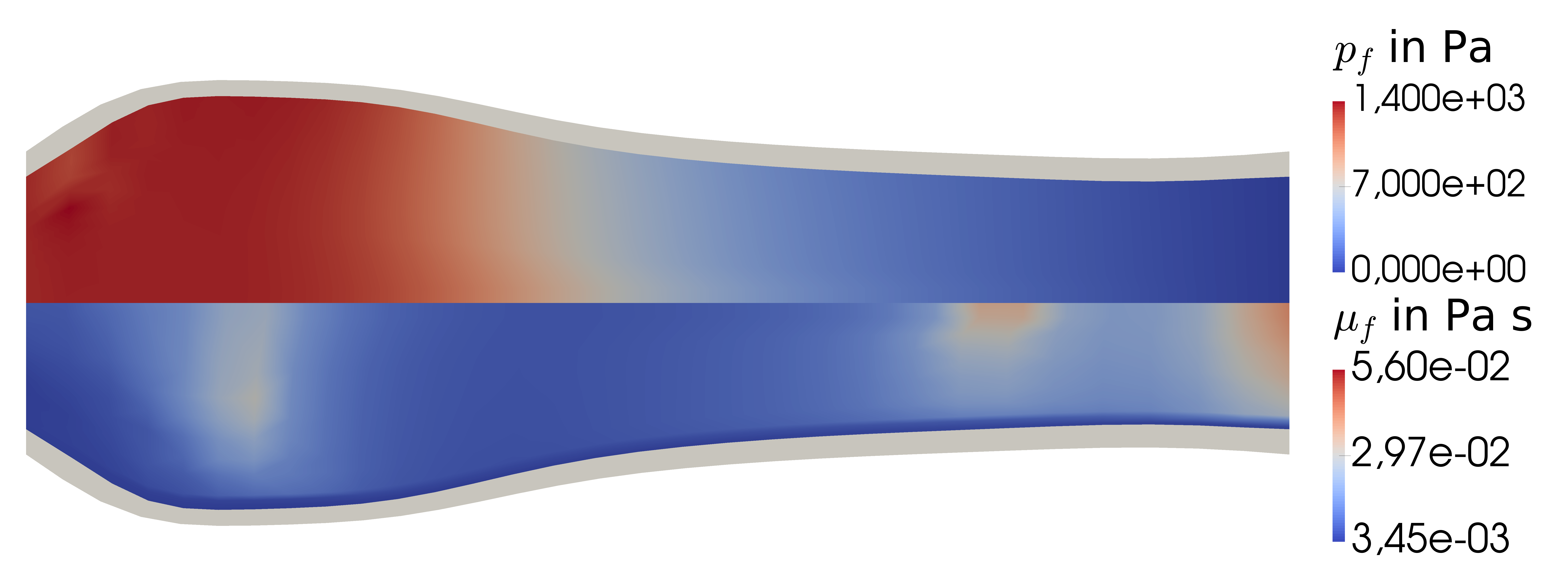}
			%\put(-40,75){$\hInt$}
		}
	\end{minipage}
	\\
	\begin{minipage}{.48\linewidth}
		\centering
		\subfloat[$\ve{d}_s$ and $\ve{u}_f$ at $t=22~\text{ms}$.]{
			\label{fig:pipe_d_u_step21}
			\includegraphics[width=0.97\textwidth]{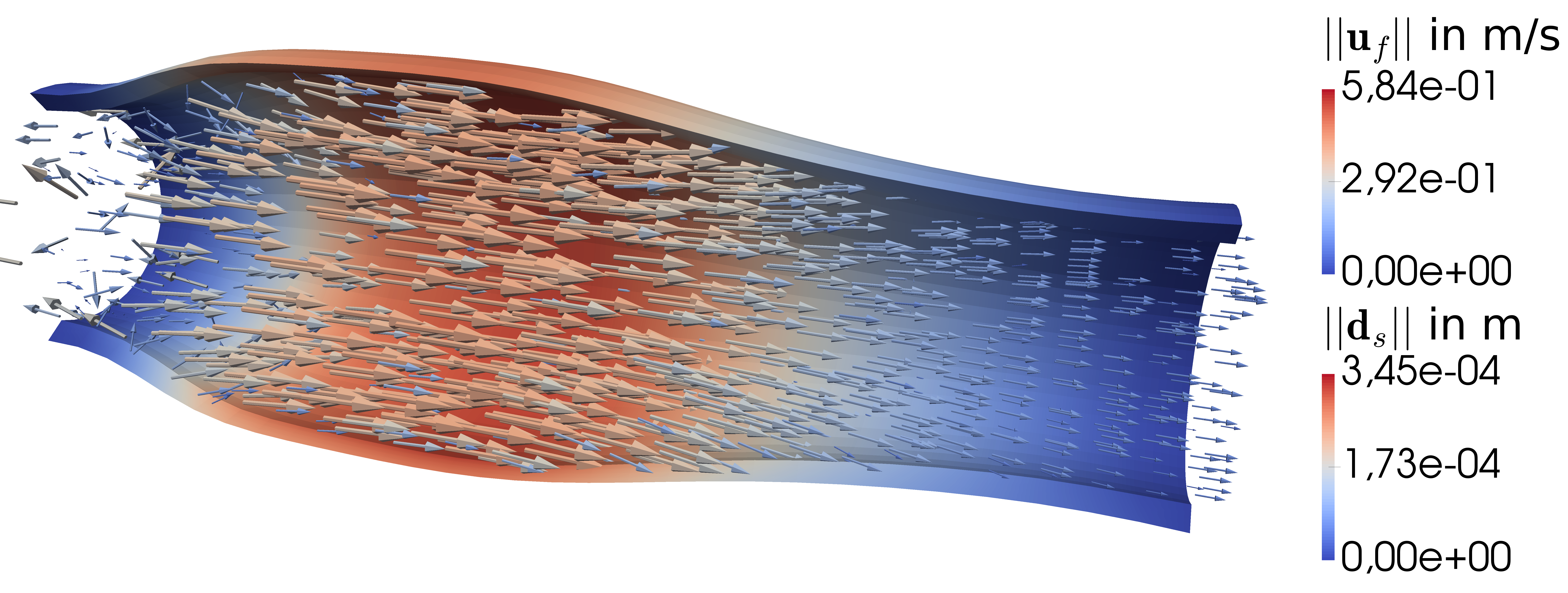}
			%\put(-10,72.5){$\hInt$}
		}
	\end{minipage}%
	\hfil
	\begin{minipage}{.48\linewidth}
		\centering
		\subfloat[$p_f$ (top) and $\mu_f$ (bottom) at $t=22~\text{ms}$.]{
			\label{fig:pipe_p_mu_step21}
			\includegraphics[width=0.97\textwidth]{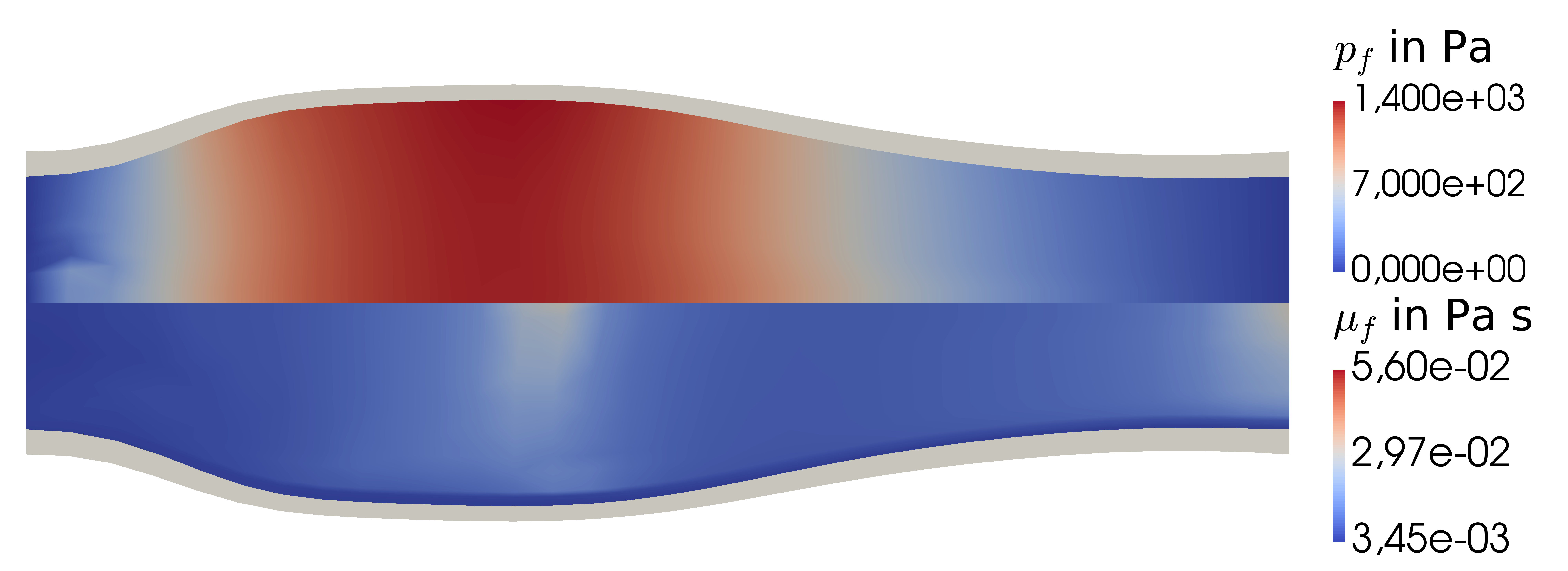}
			%\put(-40,75){$\hInt$}
		}
	\end{minipage}
	\\
	\begin{minipage}{.48\linewidth}
		\centering
		\subfloat[$\ve{d}_s$ and $\ve{u}_f$ at $t=29~\text{ms}$.]{
			\label{fig:pipe_d_u_step28}
			\includegraphics[width=0.97\textwidth]{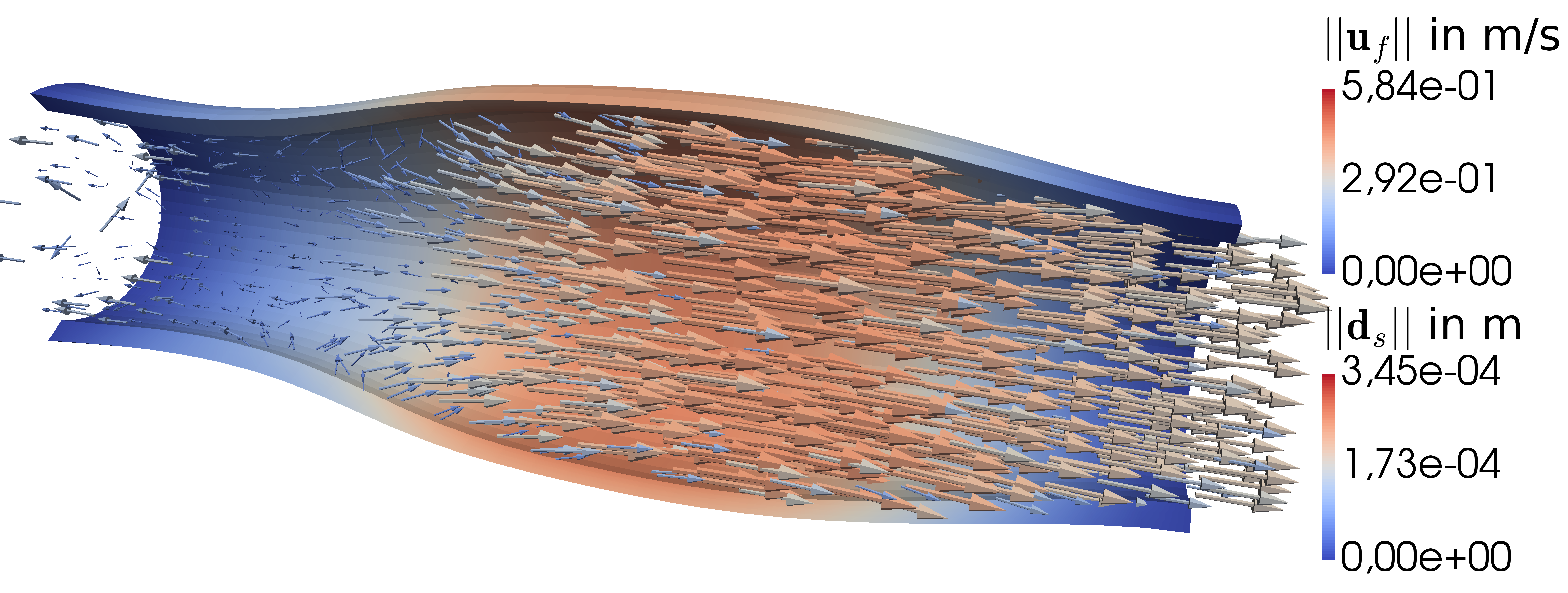}
			%\put(-10,72.5){$\hInt$}
		}
	\end{minipage}%
	\hfil
	\begin{minipage}{.48\linewidth}
		\centering
		\subfloat[$p_f$ (top) and $\mu_f$ (bottom) at $t=29~\text{ms}$.]{
			\label{fig:pipe_p_mu_step28}
			\includegraphics[width=0.97\textwidth]{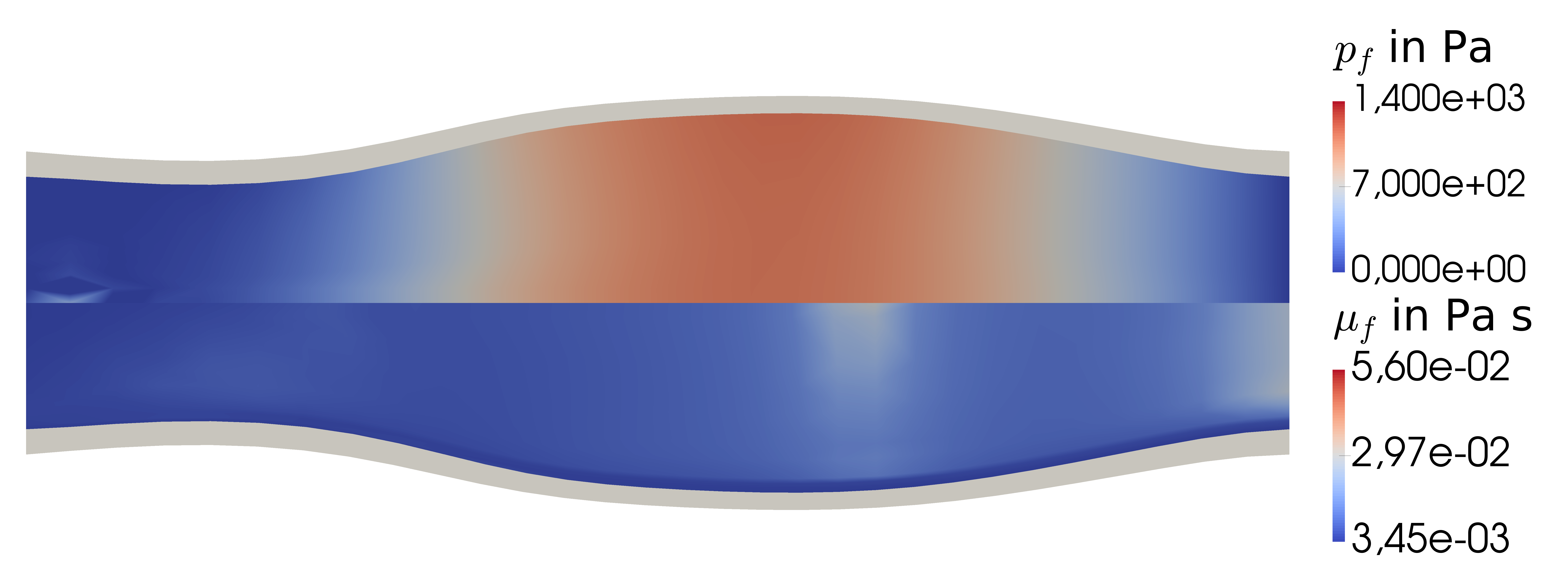}
			%\put(-40,75){$\hInt$}
		}
	\end{minipage}
	\caption{Snapshots at time $t=15,22,29~\text{ms}$ of the pressure pulse in the Carreau fluid traveling the tube of HGO material (deformation scaled by $10$): Solid displacement $\ve{d}_s$ and fluid velocity $\ve{u}_f$ in cut tube (left), pressure $p_f$ and viscosity $\mu_f$ in slice at $x_2=0$ (right).}
	\label{fig:pipe_pulse_snapshots}
\end{figure}
\begin{figure}[!htbp]
	\begin{minipage}{.48\linewidth}
		\centering
		\subfloat[Displacement comparison.]{
			\label{fig:pipe_compare_const_d}
			\includegraphics[width=0.97\textwidth]{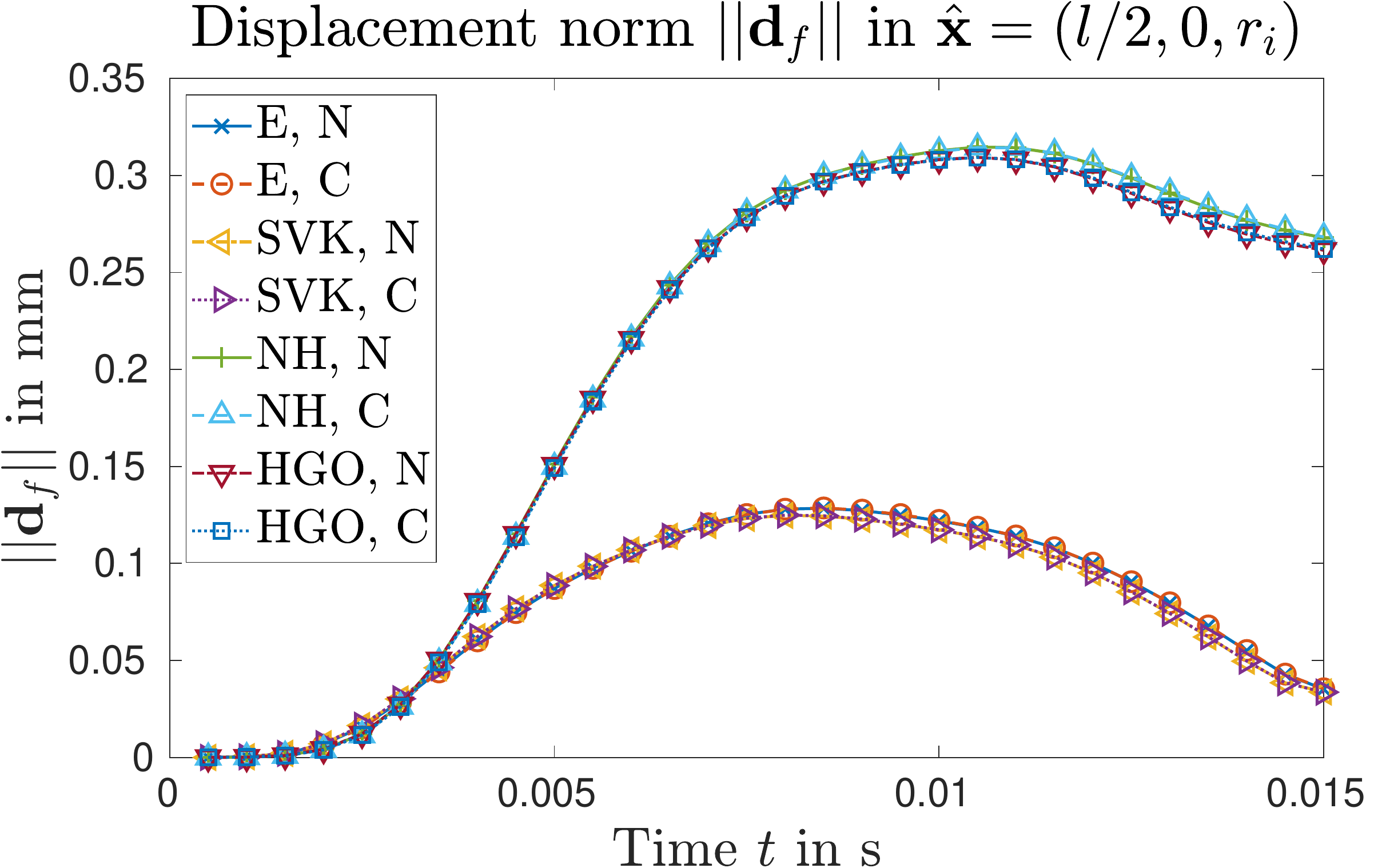}
			%\put(-10,72.5){$\hInt$}
		}
	\end{minipage}%
	\hfil
	\begin{minipage}{.48\linewidth}
		\centering
		\subfloat[Pressure comparison.]{
			\label{fig:pipe_compare_const_p}
			\includegraphics[width=0.97\textwidth]{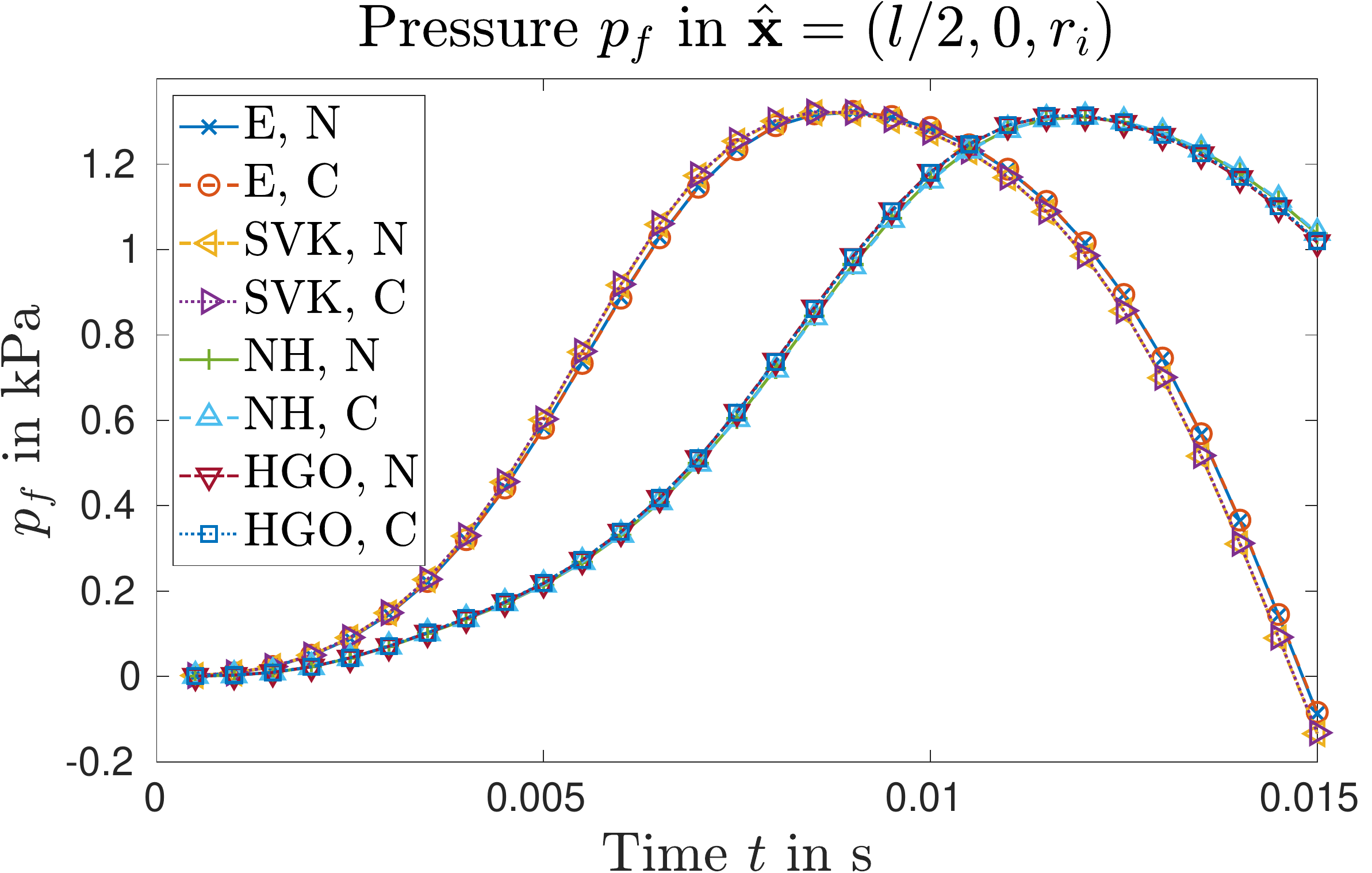}
			%\put(-40,75){$\hInt$}
		}
	\end{minipage}
	\caption{Pressure-pulse benchmark with linear elasticity~(E), St.~Venant--Kirchhoff~(SVK), neo-Hookean~(NH) or Holzapfel--Gasser--Ogden~(HGO) material models for the solid phase and Newtonian~(N) or Carreau~(C) fluids.}
	\label{fig:pipe_compare_const}
\end{figure}

Let us note here that the presented solutions can merely indicate that suitable parameters and material laws themselves are central aspects in the description of the system behavior. This, however, does not lie within the scope of this contribution, since the aim here is simply to demonstrate versatility by easily switching between constitutive equations of particular interest depending on the available data or application. Within this work, we present this problem setup to assess computational performance of the various coupling schemes presented. From now on, we only consider linear elasticity with a Newtonian fluid and the HGO model together with the Carreau law. 

Inspecting the accumulated FSI iterations depicted in Figure~\ref{fig:pipe_compare_iters}, we see that the total iteration counts are decreasing with the semi-implicit schemes. Interestingly, treating the fluid mesh motion and momentum balance equations explicitly does not only reduce the computational cost of one individual coupling step, but also reduces the number of total iterations. This effect is independent of the material laws applied, potentially boosting computational performance depending on the fluid and solid subproblem sizes. But on the flip side of the coin, we rather surprisingly see that the Robin--Robin variants of the scheme do not improve convergence in contrast to the results presented in~\cite{Astorino2010,Badia2008}. In those papers, the Robin condition on the fluid phase was found to substantially decrease the number of needed coupling steps, which does not seem to transfer to the present split-step approach. To diminish a possible influence of the Robin parameter, tests are conducted in the linear elasticity/Newtonian fluid case varying the Robin parameters in the fluid momentum balance and PPE. The latter option modifies Equations~\eqref{eqn:fsi_press_proj}~and~\eqref{eqn:fsi_ppe} to include Robin interface conditions as well, but was found diverging for any $\eta_f^R$. Scaling the Robin parameter $\eta_f^R$ computed via Equation~\eqref{eqn:robin_parameters} in the fluid momentum balance by some $\alpha_R$, iteration counts stay almost identical for $\alpha_R$ large enough, as can be seen in Figure~\ref{fig:pressure_pulse_IDN_linE_Newtonian_Robin_vary}. Choosing $\alpha_R\rightarrow0$ and $\alpha_R\rightarrow\infty$ correspond to Neumann--Neumann and the standard Dirichlet--Neumann coupling schemes, first of which needs more than the preset maximum of 150 steps to converge leading to divergence in the fifth time step.

The Robin parameter $\eta_s^R$ in the solid's balance of linear momentum has two distinct interpretations depending on the coupling scheme applied: in implicit schemes, a standard Robin--Robin coupling is recovered, possibly mildly increasing efficiency as reported by~\citet{Badia2008}. In our semi-implicit scheme, however, it is recalled that the interface Robin term is responsible for information transfer from fluid to solid phase (with $n$ indicating the time step, $k$ the iterate in the FSI coupling scheme and $l$ the iterate in Newton's method) as
\begin{align*}
	\alpha_f' \langle \ve{\varphi} , \te{P}(\ve{d}_s^{n+1}) \hns \rangle_{\hInt}
	= &
	\alpha_f'
	\langle \ve{\varphi},
	\ve{h}_s^{n+1} - \eta_s^R \dot{\ve{d}}_s^{n+1}(\ve{d}_s^l)
	\rangle_{\hInt}
	\\
	= &
	\alpha_f'
	\langle \ve{\varphi},
	\eta_s^R \ve{u}_f^k + J_f^k \te{\sigma}(\ve{u}_f^k,p_f^{k+1},\mu_f^{k+1}) \te{F}_{\! f}^{-T}\hns - \eta_s^R \dot{\ve{d}}_s^{n+1}(\ve{d}_s^l)
	\rangle_{\hInt}
	,
\end{align*} 
which as $\eta_s^R\rightarrow0$ reduces to a Neumann condition on the interface, but as $\eta_s^R\rightarrow\infty$ smoothly transitions to an explicit coupling scheme based on the extrapolated velocity $\ve{u}_f^k$. Despite the latter greatly decreasing the number of coupling iterations depending on the specific choice for $\eta_s^R$, one also observes decreased temporal stability. Therefore, $\eta_f^R = 10^2 \frac{\rho_s}{\Delta t}$ and $\eta_s^R = 10^{-4} \frac{\rho_f}{\Delta t}$ are chosen for all of the RR results presented in Figure~\ref{fig:pipe_compare_iters}.

\begin{figure}%[!htbp]
	\begin{minipage}{.48\linewidth}
		\centering
		\subfloat[Linear Elasticity, Newtonian fluid.]{
			\label{fig:pipe_compare_iters_FSI_linENewton}
			\includegraphics[width=0.97\textwidth]{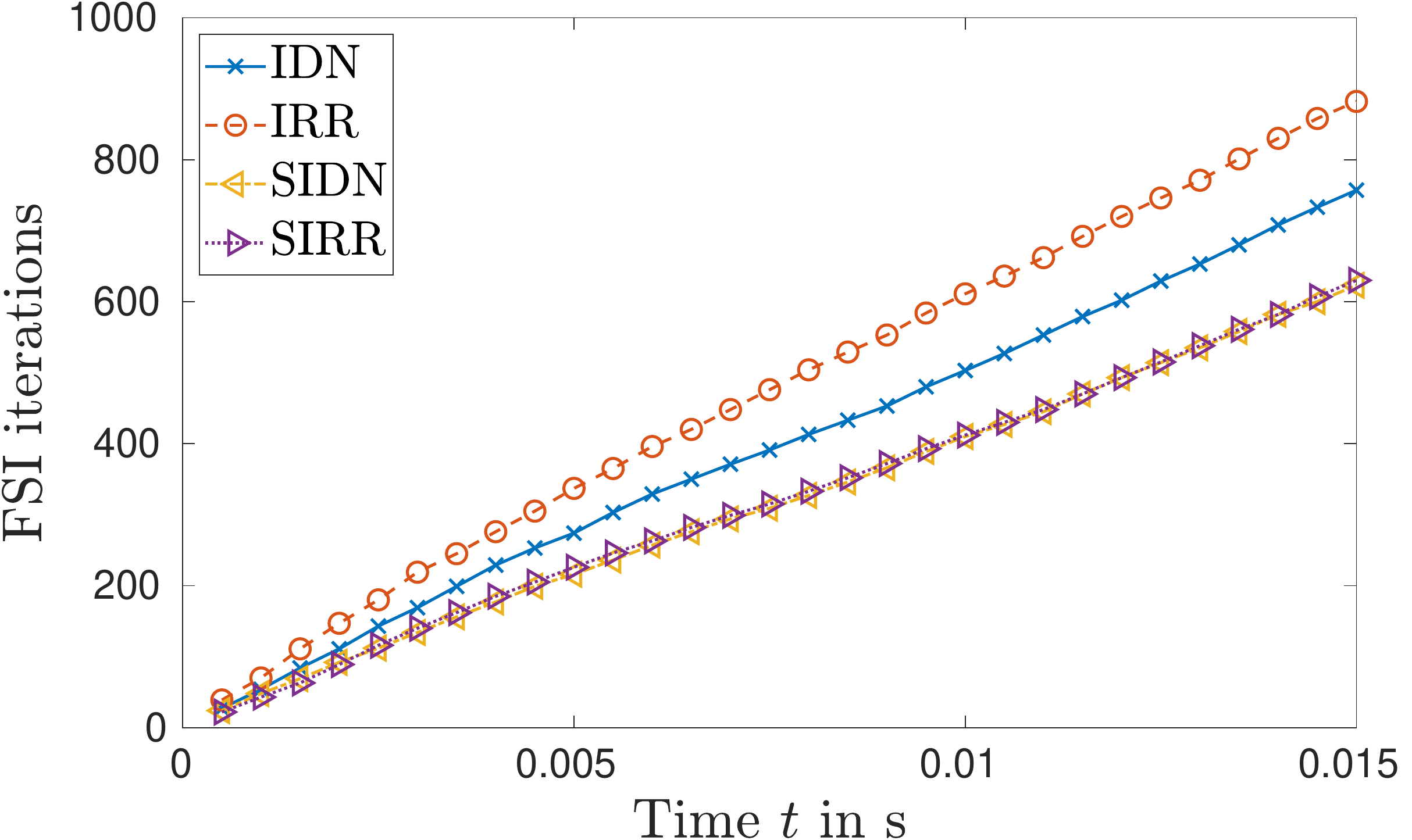}
			%\put(-10,72.5){$\hInt$}
		}
	\end{minipage}%
	\hfil
	\begin{minipage}{.48\linewidth}
		\centering
		\subfloat[HGO model, Carreau fluid.]{
			\label{fig:pipe_compare_iters_FSI_qiHGOCarreau}
			\includegraphics[width=0.97\textwidth]{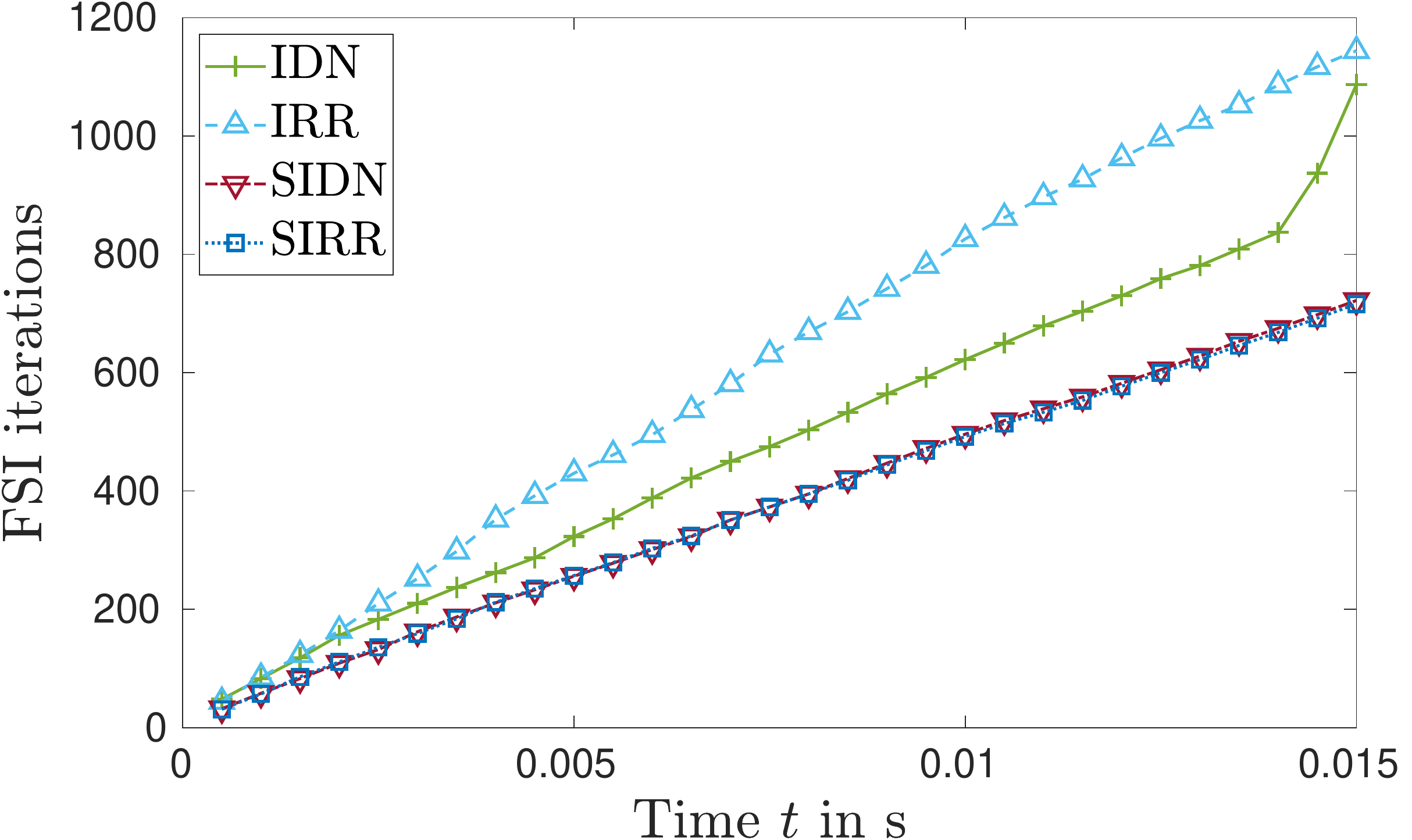}
			%\put(-40,75){$\hInt$}
		}
	\end{minipage}
	\caption{Accumulated FSI iterations using implicit Dirichlet--Neumann~(IDN) or Robin--Robin~(IRR) variants and semi-implicit counterparts~(SIDN,~SIRR) in the pressure pulse benchmark.}
	\label{fig:pipe_compare_iters}
\end{figure}

Before moving on to a final, more challenging example, let us summarise the insights gained from this simple pressure pulse benchmark: As expected, semi-implicit treatment of both the fluid mesh motion and fluid balance of linear momentum is found sufficiently accurate, while the constitutive equations considered for fluid and solid phases can be exchanged effortlessly similar to other partitioned approaches. In this simple setup, however, neither the generalised Newtonian rheology nor the nonlinear contributions, which differentiate linear elasticity from the St.~Venant--Kirchhoff model or the fiber reinforcement in the HGO model compared to the neo-Hookean one, change the observed quantities much. Moreover, introducing Robin conditions in the framework did not improve convergence properties, only mild improvements are seen in some cases, but instabilities arise from unsuitable choices. The semi-implicit Robin--Robin or Dirichlet--Robin schemes transition to fully explicit schemes as $\eta_s^R\rightarrow\infty$, which are found rather unstable in first numerical tests. Consequently, the method of choice is the SIDN algorithm, which delivered low iteration counts, stable solutions and requires less data transfer than the Robin variants.

\begin{figure}%[!htbp]
	\centering
	\includegraphics[width=0.49\linewidth]{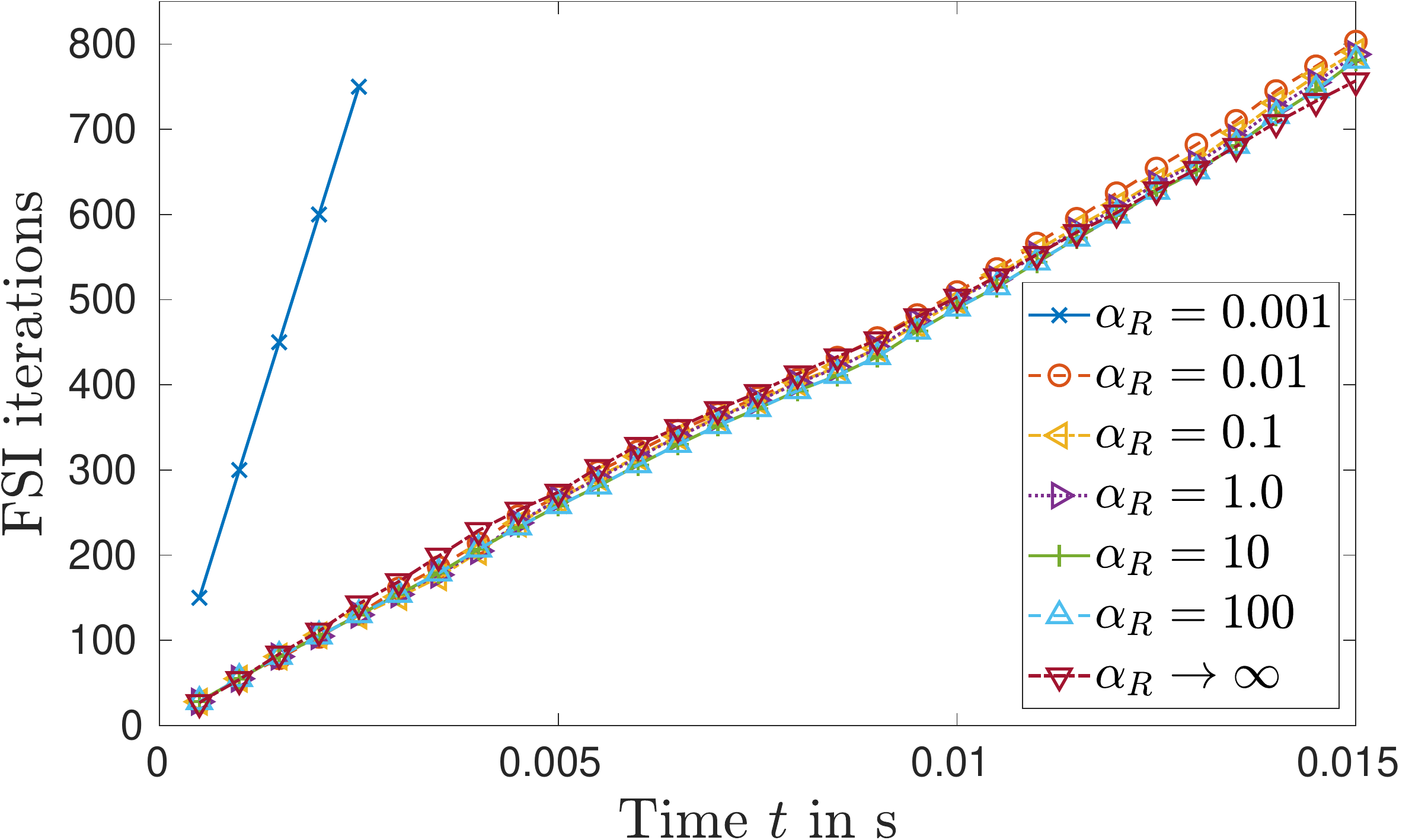}
	\caption{Accumulated FSI iterations using implicit Robin--Neumann coupling with Robin parameter $\eta_f^R = \alpha_R \frac{\rho_s}{\Delta t}$, linear elasticity and a Newtonian fluid in the pressure pulse benchmark.}
	\label{fig:pressure_pulse_IDN_linE_Newtonian_Robin_vary}
\end{figure}

\subsection{Flow through an idealised aneurysm}

In this final numerical example, we consider the flow through an idealised abdominal aortic aneurysm (AAA) using physiological parameter sets to challenge the devised framework. The setup chosen is similar to \cite{Lin2017}, combining a prototypical geometry as recommended by surgeons \cite{Meyer2011,Deplano2007} with flow data from \cite{Mills1970}. Comparable settings in a biomedical context with benchmark character were presented by \citet{Turek2011} and \citet{Balzani2016}, while patient-specific geometries were considered, e.g., in \cite{Gee2011,Kuttler2010,Bazilevs2010b}.

The AAA geometry with a length of $l=20~\text{cm}$ and with inner/lumen radius of $r_i=1~\text{cm}$ at the in- and outlet is discretised using $\approx1.3 \times 10^{5}$ trilinear elements as depicted in Figure~\ref{fig:AAA_mesh}, resulting in $\approx1.3 \times 10^{5}$ nodes. These radii are expanded to a maximum of $\approx 6.5~\text{cm}$ in lateral direction and $\approx5.5~\text{cm}$ in anterior-posterior direction in the middle of the aneurysm. \revised{The medial and adventitial layers of the aorta}, $\hat{\Omega}_{s,1}$ and $\hat{\Omega}_{s,2}$, are considered with a uniform and equal thickness of $h_s=0.75~\text{mm}$, where a fiber orientation is constructed by solving two auxiliary Laplace equations as in the pressure pulse benchmark and \revised{further processing the resulting radial and longitudinal orientation vectors shown in Figure~\ref{fig:AAA_ori} to obtain mean fiber directions rotated by $\pm \alpha_c$ from circumferential into longitudinal direction}. The computational domain is distributed to 8 processors as indicated by different colors in Figure~\ref{fig:AAA_processors}, ignoring the fluid structure interface.

\begin{figure}%[!htbp]
	\begin{minipage}{.48\linewidth}
		\centering
		\subfloat[Fluid mesh with boundary layers and cut solid domains resolving medial and adventitial tissue layers.]{
			\label{fig:AAA_layers}
			\includegraphics[width=0.97\textwidth]{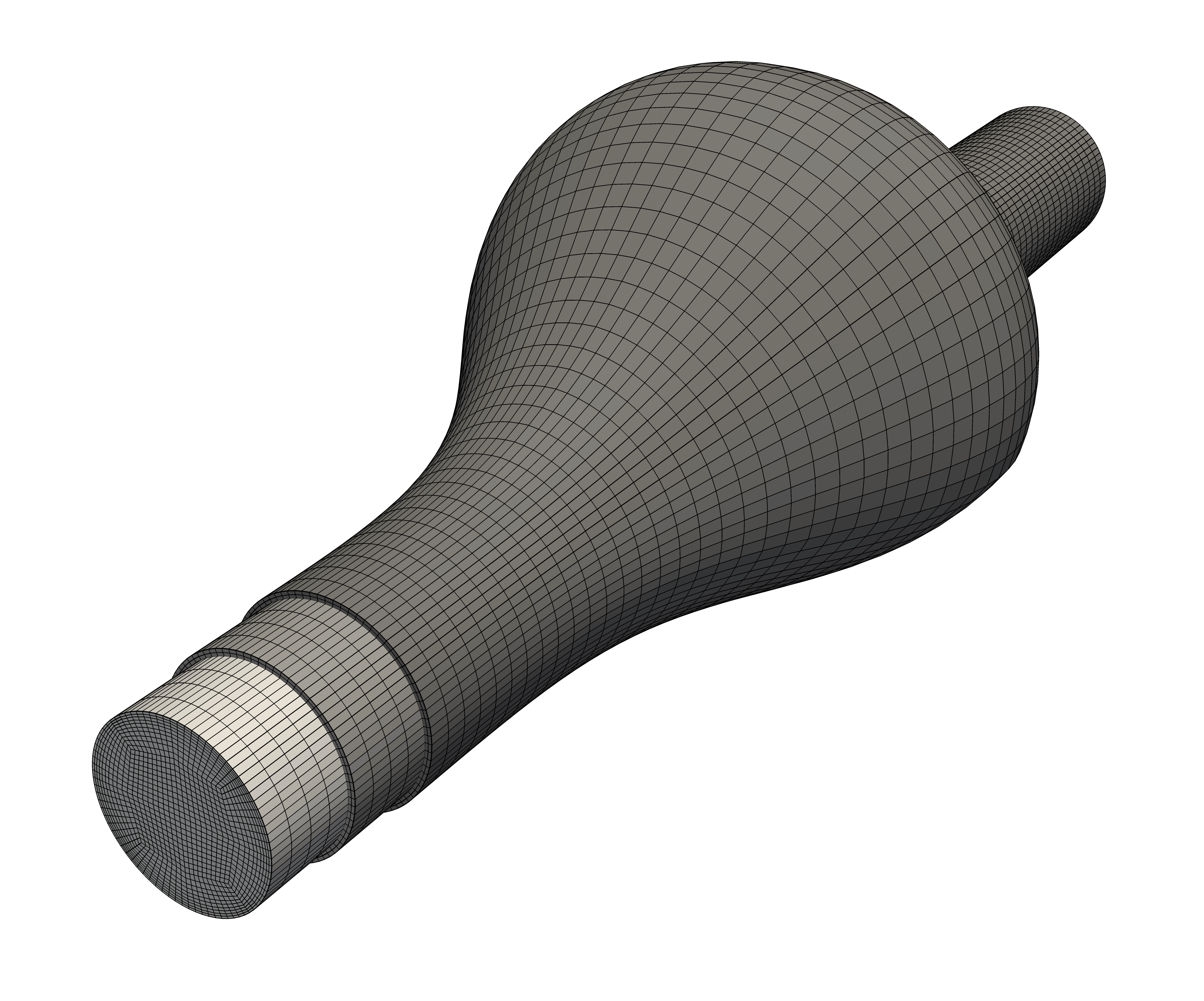}
			\put(-190,90){{$\hat{\Omega}_{s,2}$}}
			\put(-205,75){{$\hat{\Omega}_{s,1}$}}
			\put(-220,60){{$\hOf$}}
			%\put(-10,72.5){$\hInt$}
		}
	\end{minipage}%
	\hfil
	\begin{minipage}{.48\linewidth}
		\centering
		\subfloat[Radial (blue) and longitudinal (yellow) orientation vectors in cut solid domains used to construct $\hat{\ve{e}}_1$ and $\hat{\ve{e}}_2$.]{
			\label{fig:AAA_ori}
			\includegraphics[width=0.97\textwidth]{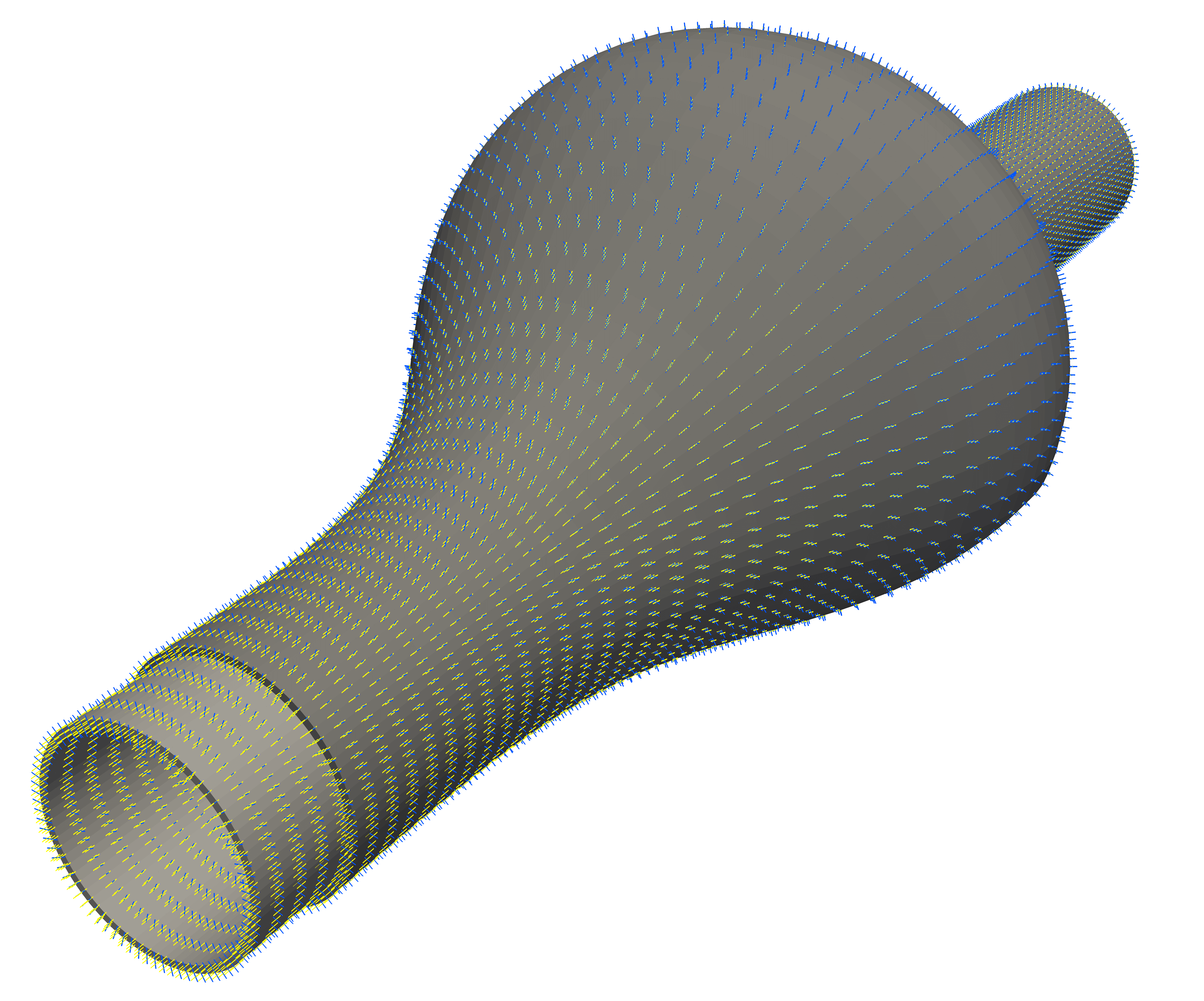}
			%\put(-70,65){{$\hOs$}}
			%\put(-5,105){{$\hOf$}}
			%\put(-40,75){$\hInt$}
		}
	\end{minipage}
	\caption{Finite element mesh used in the idealised abdomial aortic aneurysm example.}
	\label{fig:AAA_mesh}
\end{figure}
\begin{figure}%[!htbp]
	\begin{minipage}{.48\linewidth}
		\centering
		\subfloat[Discretised AAA geometry.]{
			\label{fig:AAA_processors}
			\includegraphics[height=5cm]{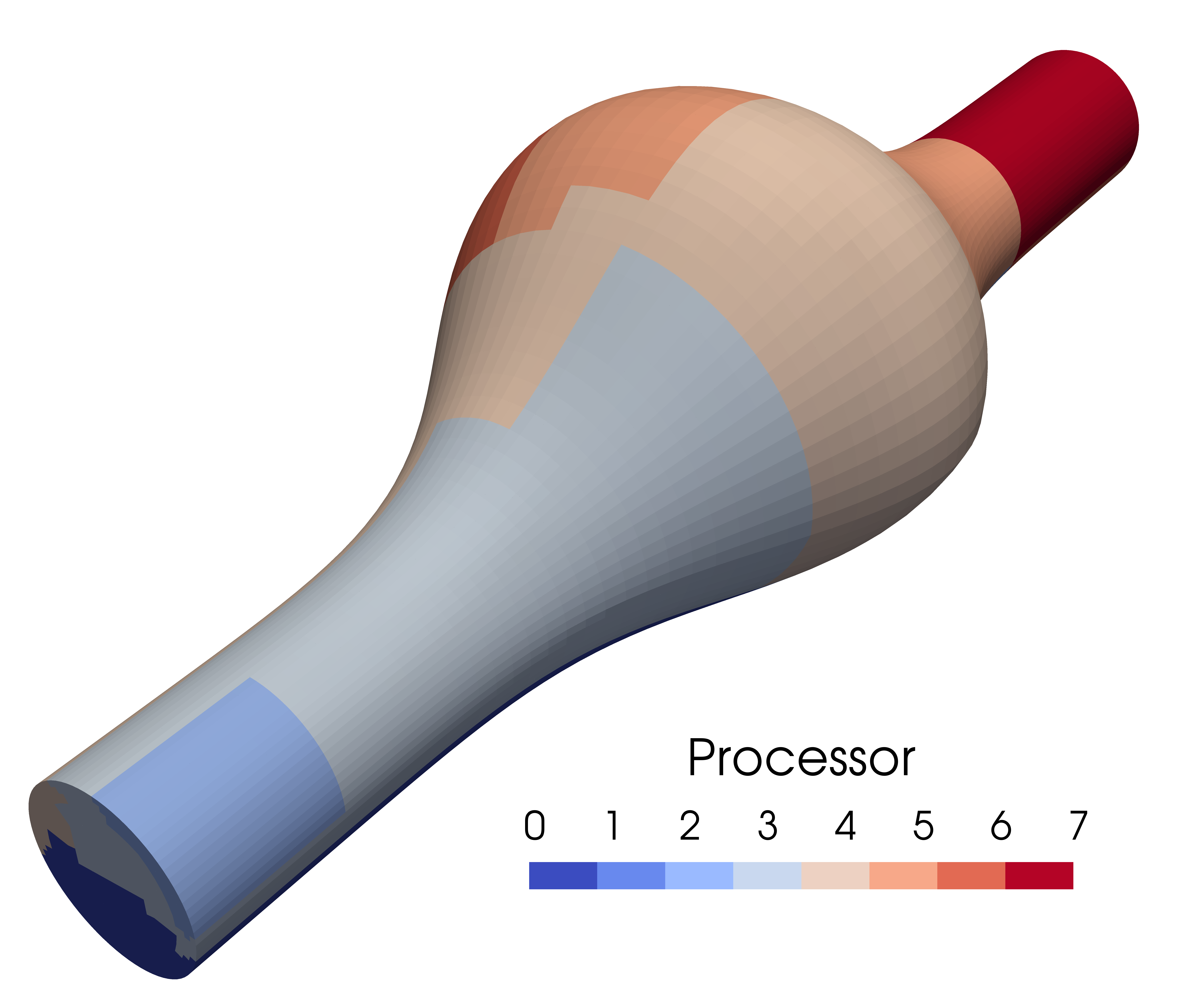}
			%\put(-190,90){{$\hat{\Omega}_{s,2}$}}
			%\put(-205,75){{$\hat{\Omega}_{s,1}$}}
			%\put(-220,60){{$\hOf$}}
			%\put(-10,72.5){$\hInt$}
		}
	\end{minipage}%
	\hfil
	\begin{minipage}{.48\linewidth}
		\centering
		\subfloat[In- and outlet data.]{
			\label{fig:AAA_vin_pout}
			\includegraphics[height=5cm]{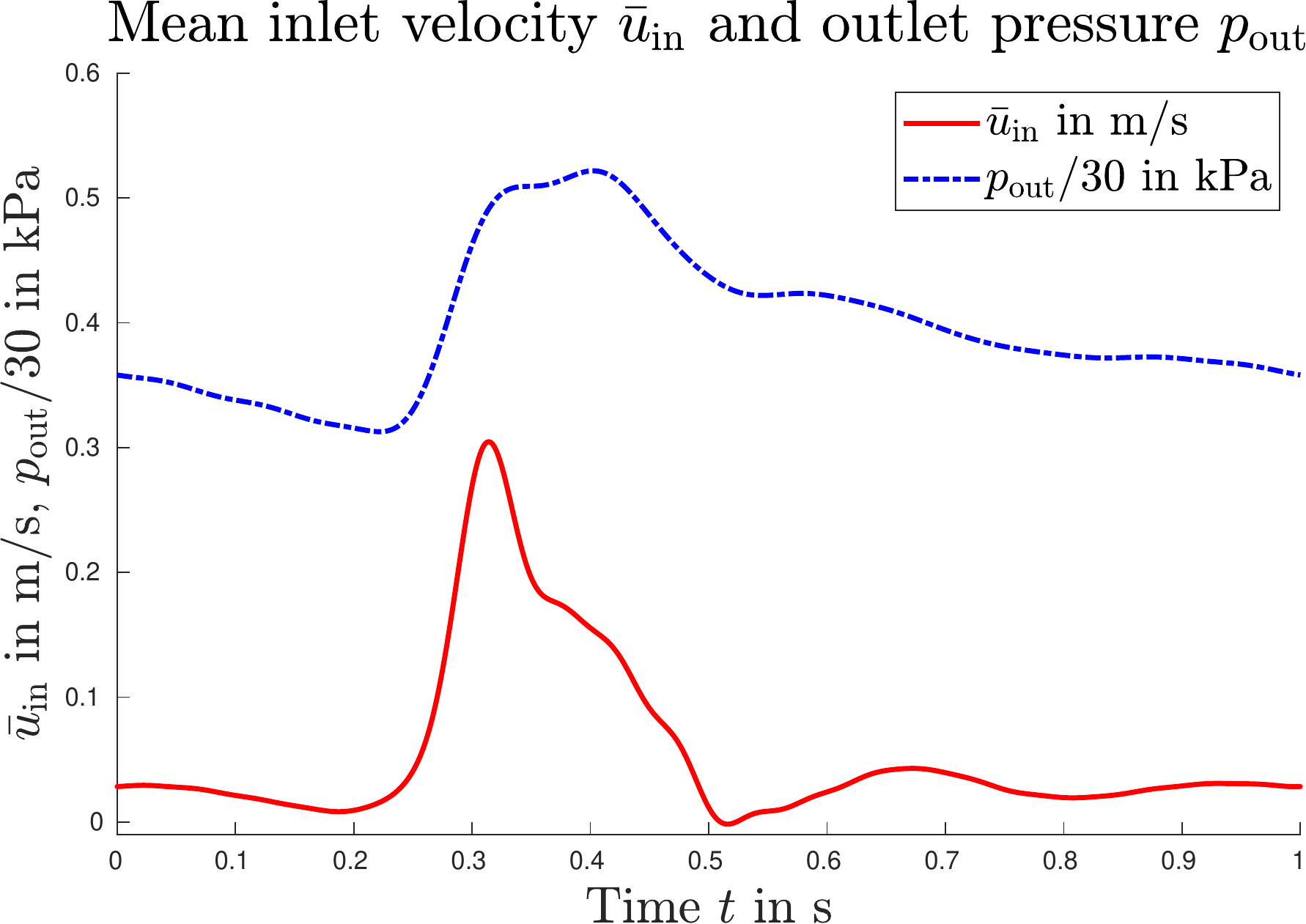}
			%\put(-70,65){{$\hOs$}}
			%\put(-5,105){{$\hOf$}}
			%\put(-40,75){$\hInt$}
		}
	\end{minipage}
	\caption{Realistic geometry distributed to 8 processors (left) and flow data (right) \cite{Lin2017,Meyer2011,Deplano2007}.}
	\label{fig:AAA_processors_and_vinpout}
\end{figure}

Regarding boundary conditions, we fix the nodes at the in- and outlet and consider a viscoelastic external tissue supporting the solid by enforcing
\begin{align*}
    \te{P}\hns =\ve{t}_s = - k_{e} \ve{d}_s - c_{e} \frac{\partial}{\partial t} \ve{d}_s - p_e \hns
\end{align*}
via the boundary term on $\hGNs$ with $k_e=10^7~\text{N/m}^3$, $c_e=10^5~\text{Ns/m}^3$ and $p_e = 0~\text{N/m}^2$ chosen similar to \cite{Moireau2012,REYMOND2013,Baumler2020}. The fluid inlet condition is computed from periodic mean inlet velocity $\bar{u}_\mathrm{in}$ as shown in Figure~\ref{fig:AAA_vin_pout}, prescribing the normal inlet velocity component $u_1$ in terms of the distance from the circular inlet center $r$:
\begin{align*}
    u_1 =  2 \bar{u}_\mathrm{in} \left( 1 - \frac{r}{r_\text{i}}\right) \xi_t(t)
    \, ,
    \text{ with} 
    \quad 
	\xi_t(t)
	=
	\begin{cases}
	\sin^2\left(\frac{\pi t}{0.4}\right) & \text{for} \ t \leq 0.2, \\
	1 & \text{otherwise},
	\end{cases}
\end{align*}
and a factor of 2 to match the volumetric flow rate computed from the mean velocity $\bar{u}_\mathrm{in}$ with a parabolic velocity profile. The outlet pressure is approximated via a three-element Windkessel model (see, e.g., \cite{EsmailyMoghadam2011,Bertoglio2013,ARBIA2016, Baumler2020})
\begin{align*}
    C \frac{\partial}{\partial t} p_p + \frac{p_p}{R_d} = Q(\ve{u}_f)
    \, ,
    \quad
    p_c - p_p = R_p Q(\ve{u}_f)
    \, ,
\end{align*}
which represents the characteristics of the excluded downstream vasculature. Therein, the flow over the outlet is denoted as $Q(\ve{u}_f)$ and parameters for the capacitance $C=1.25\times 10^{-9}~\text{m$^4$s$^2$/kg}$ and the proximal and distal resistances are specified as $R_p=266.66\times 10^{5}~\text{kg/m$^4$s}$ and $R_d=6.8\times 10^{8}~\text{kg/m$^4$s}$, respectively.
The outlet pressure is then weakly enforced through the standard traction boundary integral term, setting 
\begin{align*}
    \te{\sigma}_f \ve{n}_f = -p_c \ve{n}_f - \rho_f \revised{\ve{u}_f} (\ve{u}_f \cdot \ve{n}_f)_{-}
    \, ,
\end{align*}
also including backflow stabilisation according to \cite{Bazilevs2009}, where $(\ve{u}_f \cdot \ve{n}_f)_{-} = - \ve{u}_f \cdot \ve{n}_f$ for $\ve{u}_f \cdot\ve{n}_f \leq 0$ and $0$ otherwise. Moreover, Galerkin least-squares stabilisation \cite{Hughes1987} is added to the fluid momentum balance equation to counteract dominant convective terms.

Concerning material properties, we keep the physiological parameters as set in the pressure pulse benchmark, i.e., densities of $\rho_f=1060~\text{kg/m$^3$}$ and $\rho_s=1200~\text{kg/m$^3$}$, either a Newtonian fluid with viscosity $\mu_f = 3.5~\text{mPa~s}$ or a Carreau fluid with $\eta_0 = 56~\text{mPa~s}$, $\eta_\infty = 3.45~\text{mPa~s}$, $\lambda_f = 3.313~\text{s}$ and $n = 0.3568$, again, taken from \cite{Kim2000}. For the solid, we consider a St.~Venant--Kirchhoff solid or linear elasticity, both using a Young's modulus $E_s = 300 ~\text{kPa}$ and a Poisson's ratio $\nu_s=0.3$ and compare with layered neo-Hookean models with $\nu_s=0.499$ and a shear rate of $\mu_{s,1}=62.1~\text{kPa}$ and $\mu_{s,2}=21.6~\text{kPa}$ for inner and outer layers, respectively. Fibers are included for the HGO model, setting $k_1=1.4~\text{kPa}$, $k_2=22.1$ and $\kappa_{c,1} = 0.12$, $\alpha_{c,1} = 27.47^\circ$ and $\kappa_{c,2} = 0.25$, $\alpha_{c,2} = 52.88^\circ$ for the media and adventitia layers of the aorta \cite{Rolf-Pissarczyk2021,Weisbecker2012}. The prestress present in the tissue in a geometry reconstructed from medical image data is not considered in this setup. However, we can interpret the present mesh as the zero-stress geometry being the result of the prestress strategy ``deflating'' the model as in \cite{Tezduyar2008,Takizawa2018}. This slight discrepancy is simply ignored, since we do not investigate prestress effects.

Altogether, three pulses are considered, i.e., $t\in(0,3]$, using a uniform time step of $\Delta t=1~\text{ms}$ in the second-order accurate scheme, i.e., BDF-2 and linear extrapolation for linearisation. The WBZ$-\alpha$ and CH$-\alpha$ time integrators with $\rho_\infty=0$ are selected, since they were again found to be more robust when using the Windkessel model. Aitken's relaxation is initiated with $\omega_0 = 0.001$, coupling the pressure Poisson (PPE) and solid momentum subproblems in the semi-implicit Dirichlet--Neumann (SIDN) approach until reaching $\epsilon_{abs}=10^{-6}$ or $\epsilon_{rel} = 5\times10^{-4}$, while a relative tolerance in the Newton solver~\eqref{eqn:newton_rel_tol} of $\epsilon_N = 10^{-5}$ was selected conservatively.

Looking at solution snapshots at three distinct time instances at $t=2.2, 2.4 \text{ and } 2.6~\text{s}$, we observe strong recirculations in the AAA, especially after the rapid drop of the inflow velocity. In Figure~\ref{fig:AAA_snapshots}, selected streamlines indicate areas of recirculatory flow, leading to large gradients in the velocity fields and hence to a large shear rate, which ultimately results in strong gradients in the viscosity field. Further, the maximal and minimal viscosities observed at any point in time span almost the entire range from $\eta_\infty$ to $\eta_0$. The pressure in the lumen is spatially rather uniform, but rapidly changes in time due to the Windkessel model setting the pressure level, which dominates the deformation rather than the velocity acting on the vessel. Inspecting the time evolution of the displacement norm and fluid pressure at the apex point at $\hat{\ve{x}}=(0.1,0.039,0)^T$ depicted in Figure~\ref{fig:AAA_compare}, we see that the periodic state is not yet reached. This is due to the Windkessel model applied on the outflow boundary, which only gradually increases the pressure level in the AAA. Observing the displacement norm and fluid pressure only, effects of the more advanced solid constitutive models cannot be investigated, especially since the displacement remains rather small. The parameter choice for the tissue, however, leads to a difference easily observable with the naked eye.

Regarding the difference when applying Newtonian or Carreau rheological models, Figure~\ref{fig:AAA_compare} is not sufficient, since the relevant quantities are neither the displacement nor the pressure, but rather the time averaged shear rate and shear stress. After defining the time average $\bar{f}(\ve{x})$ of a function $f(\ve{x},t)$ over a period $T_p$ by 
\begin{align*}
    \bar{f} = \frac{1}{T_p} \int_{i T_p}^{(i+1)T_p} f(\mathcal{A}_t(\hat{\ve{x}},t),t) 
    \text{d} t
    \, ,
\end{align*}
calculate $\bar{\dot{\gamma}}$  setting $i=T_p=1$, while the shear stress is computed following \citet{John2017}:
\begin{align*}
    \te{\tau} 
    &= 
    \te{\sigma}_f\ve{n}_f - \left[ (\te{\sigma}_f\ve{n}_f) \cdot \ve{n}_f \right] \ve{n}_f
    =
    \left(
        2 \mu_f \nabla^S \ve{u}_f
    \right)
    \ve{n}_f
    -
    \left\{
        \left[
            \left( 
            2 \mu_f
                \nabla^S\ve{u}_f
            \right)
            \ve{n}_f
        \right]
        \cdot
        \ve{n}_f
    \right\}
    \ve{n}_f
    \, .
\end{align*}
Inspecting these time averaged quantities in the third cycle, as shown in Figure~\ref{fig:AAA_compare_t_avg}, a striking difference is observed as expected. Nonetheless, within this work we focus on the FSI solver rather than the phenomenological influence of different model decisions, but still want to demonstrate the versatility of the framework.

\begin{figure}%[!htbp]
	\begin{minipage}{.48\linewidth}
		\centering
		\subfloat[$\ve{d}_s$ and $\ve{u}_f$ at $t=2.2~\text{s}$.]{
			\label{fig:streamlines_219}
			\includegraphics[width=0.97\textwidth]{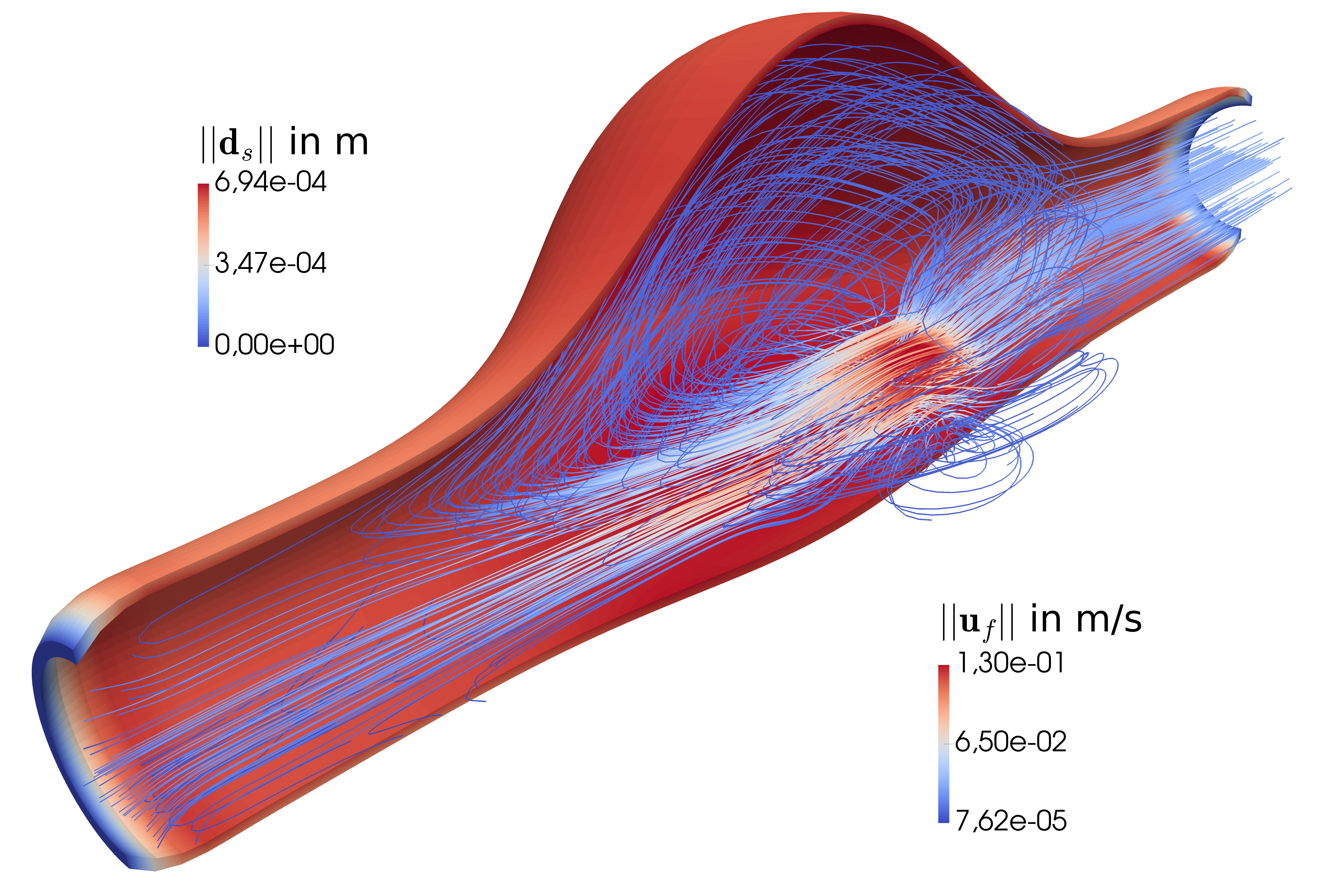}
			\put(-110,10){
			    \includegraphics[width=0.15\textwidth]{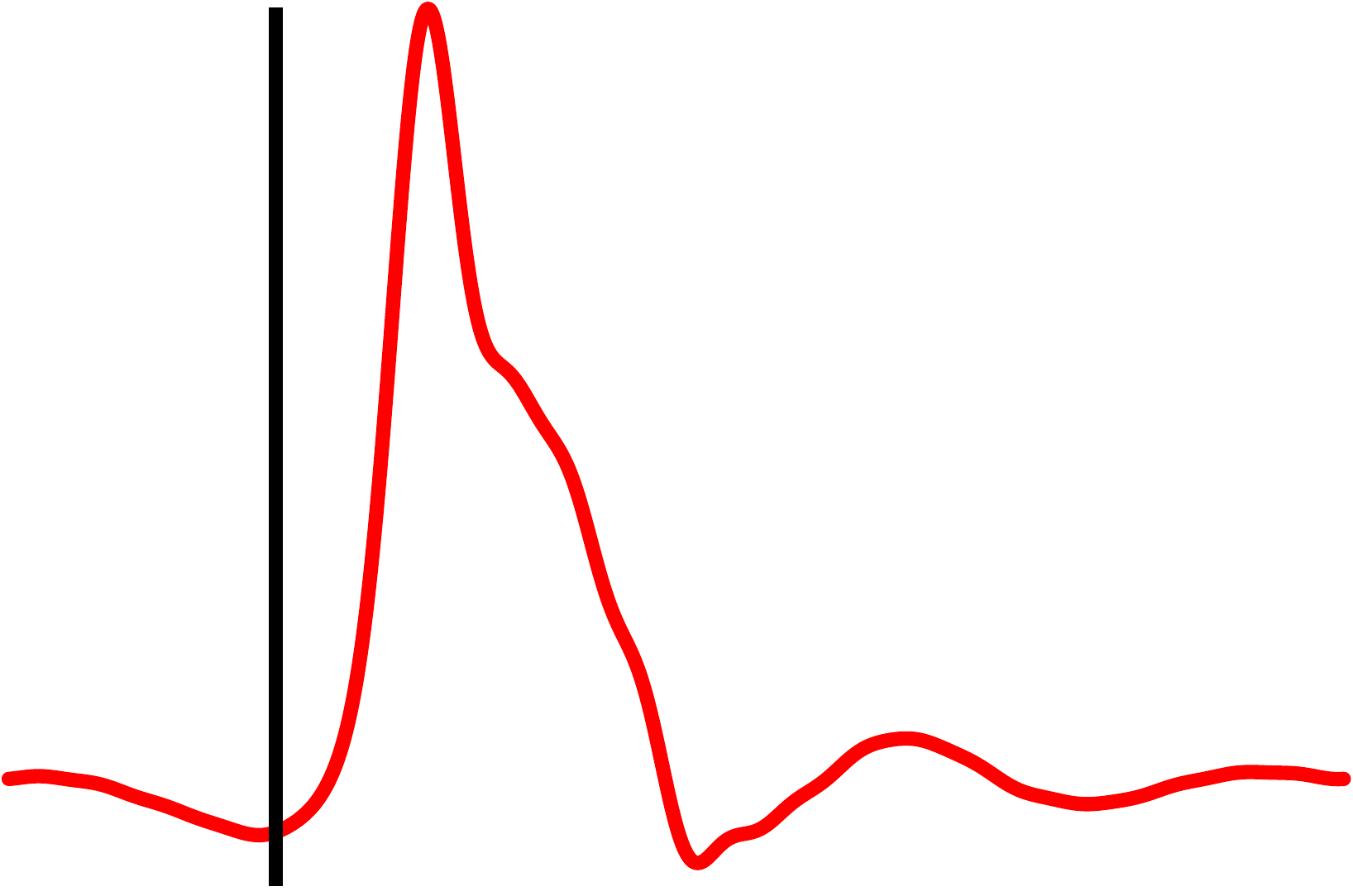}
			}
		}
	\end{minipage}%
	\hfil
	\begin{minipage}{.48\linewidth}
		\centering
		\subfloat[$\ve{d}_s$ and $\mu_f$ at $t=2.2~\text{s}$.]{
			\label{fig:viscosity_219}
			\includegraphics[width=0.97\textwidth]{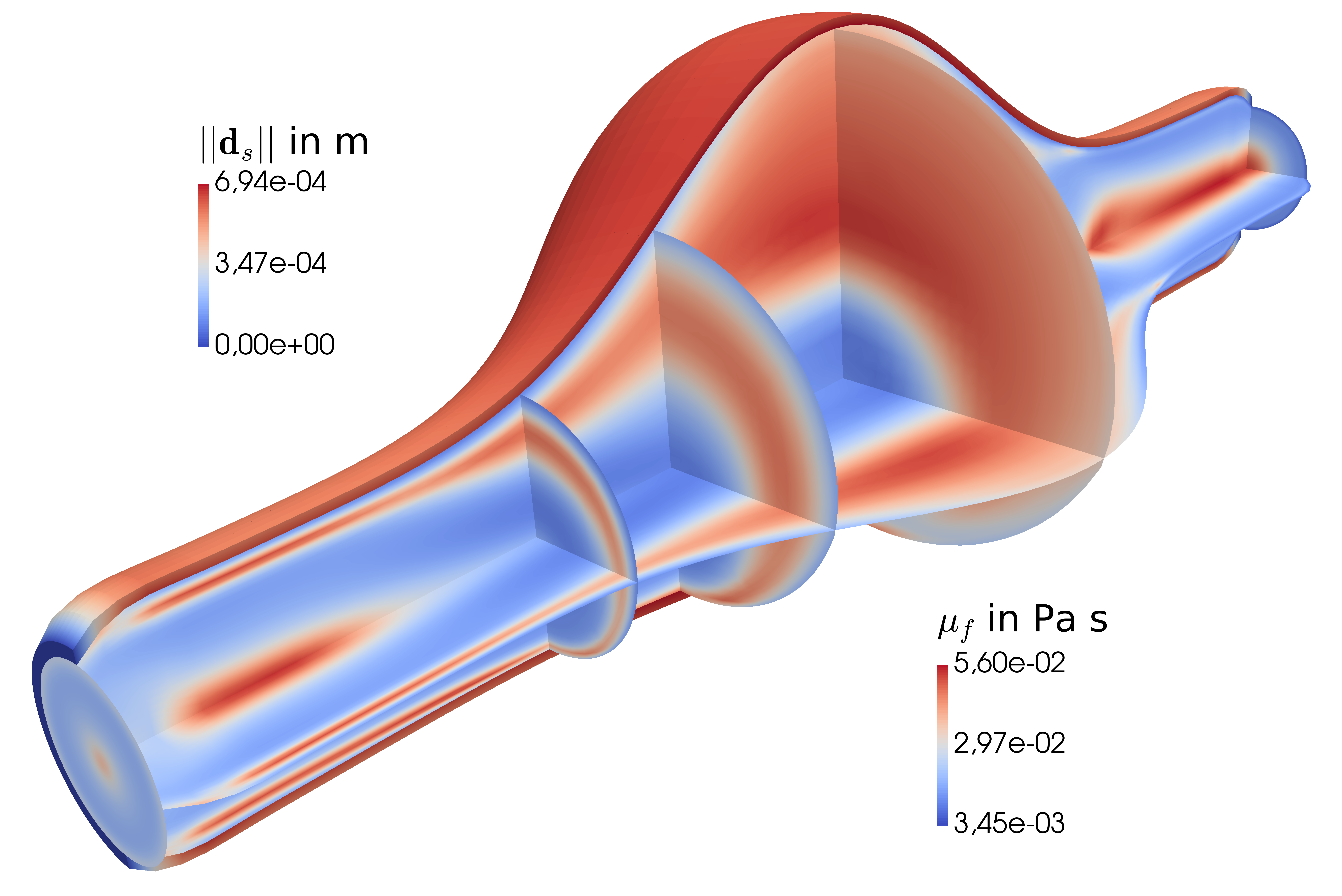}
			%\put(-40,75){$\hInt$}
		}
	\end{minipage}
	\\
	\begin{minipage}{.48\linewidth}
		\centering
		\subfloat[$\ve{d}_s$ and $\ve{u}_f$ at $t=2.4~\text{s}$.]{
			\label{fig:streamlines_239}
			\includegraphics[width=0.97\textwidth]{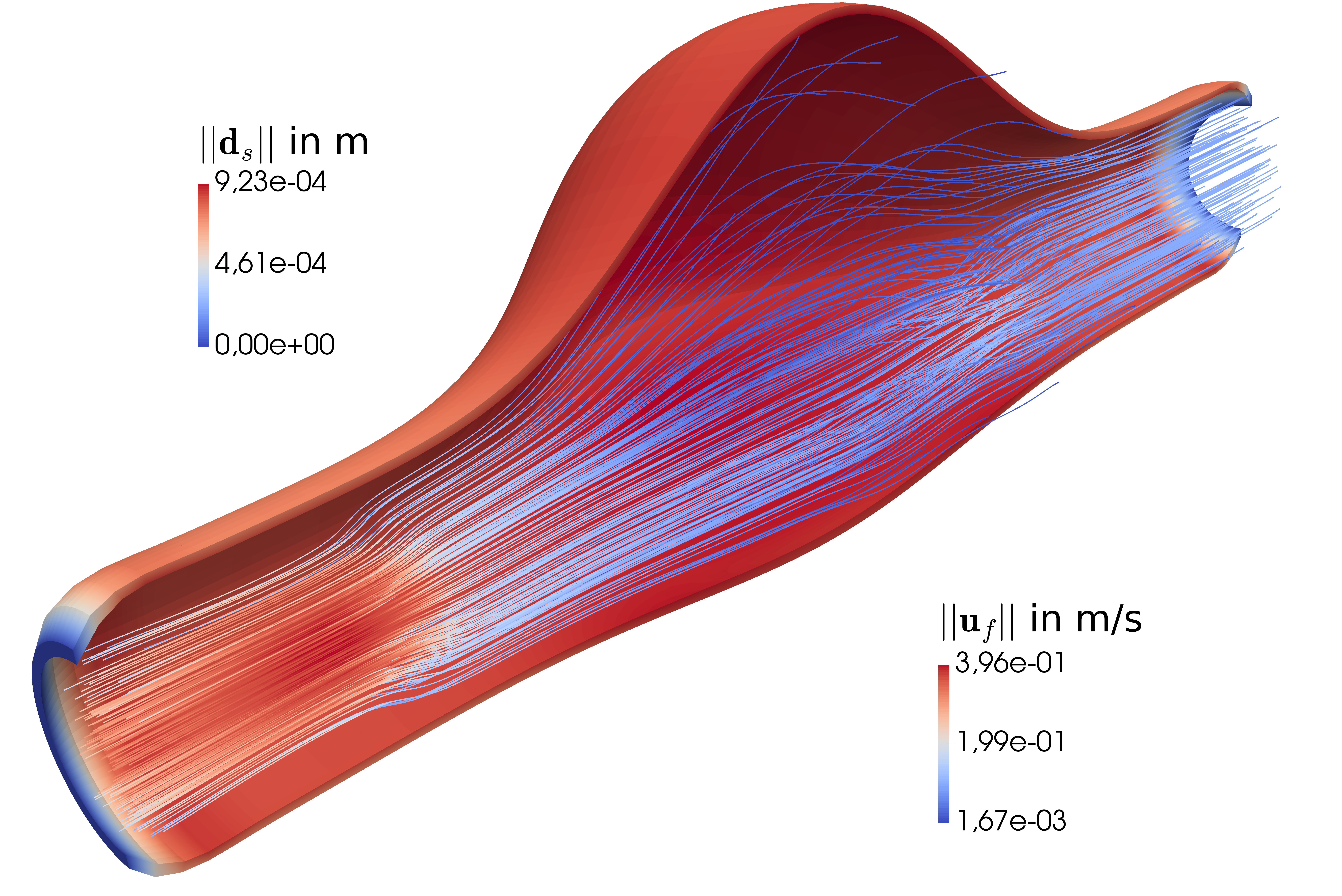}
			\put(-110,10){
			    \includegraphics[width=0.15\textwidth]{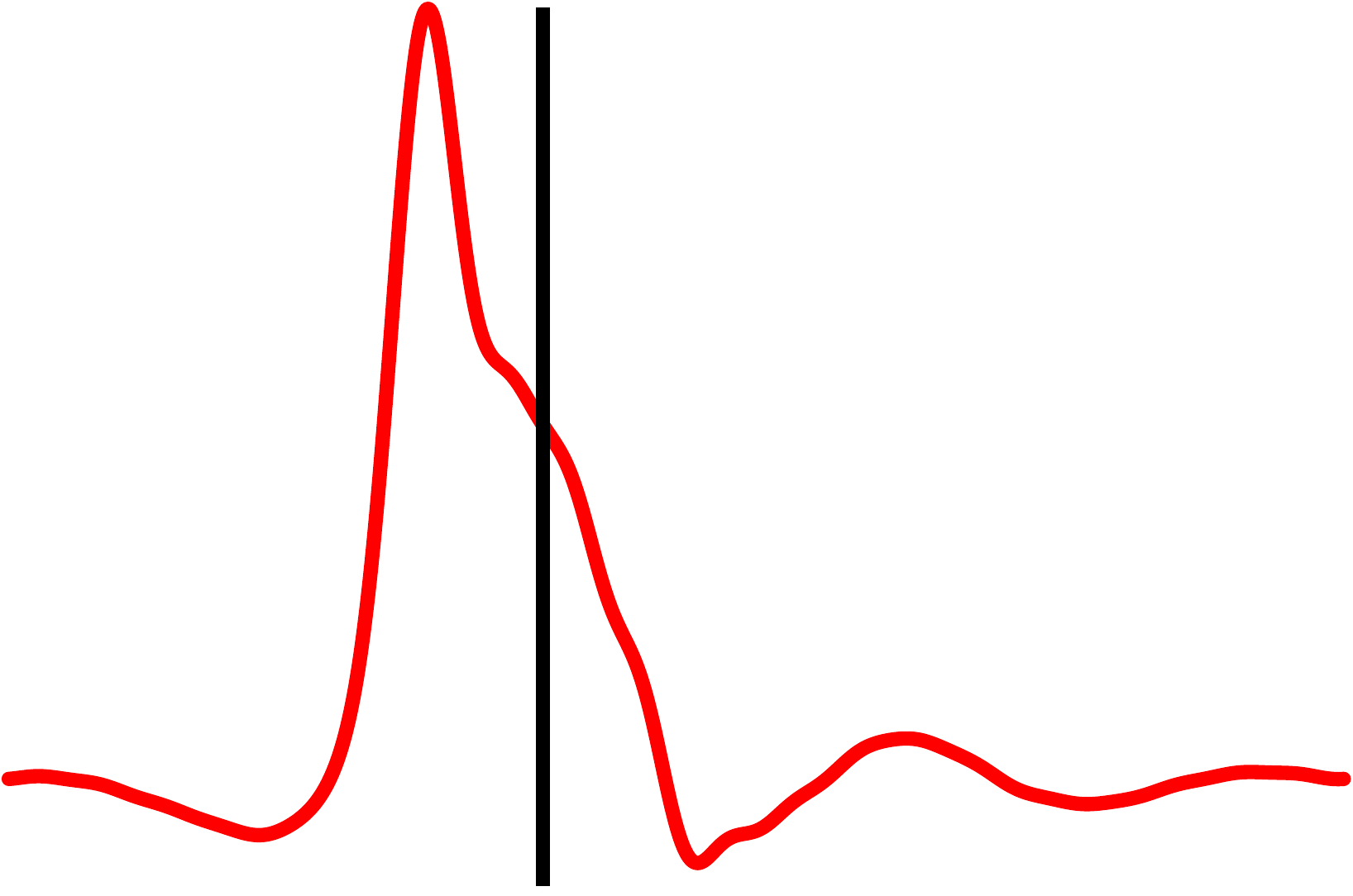}
			}
		}
	\end{minipage}%
	\hfil
	\begin{minipage}{.48\linewidth}
		\centering
		\subfloat[$\ve{d}_s$ and $\mu_f$ at $t=2.4~\text{s}$.]{
			\label{fig:viscosity_239}
			\includegraphics[width=0.97\textwidth]{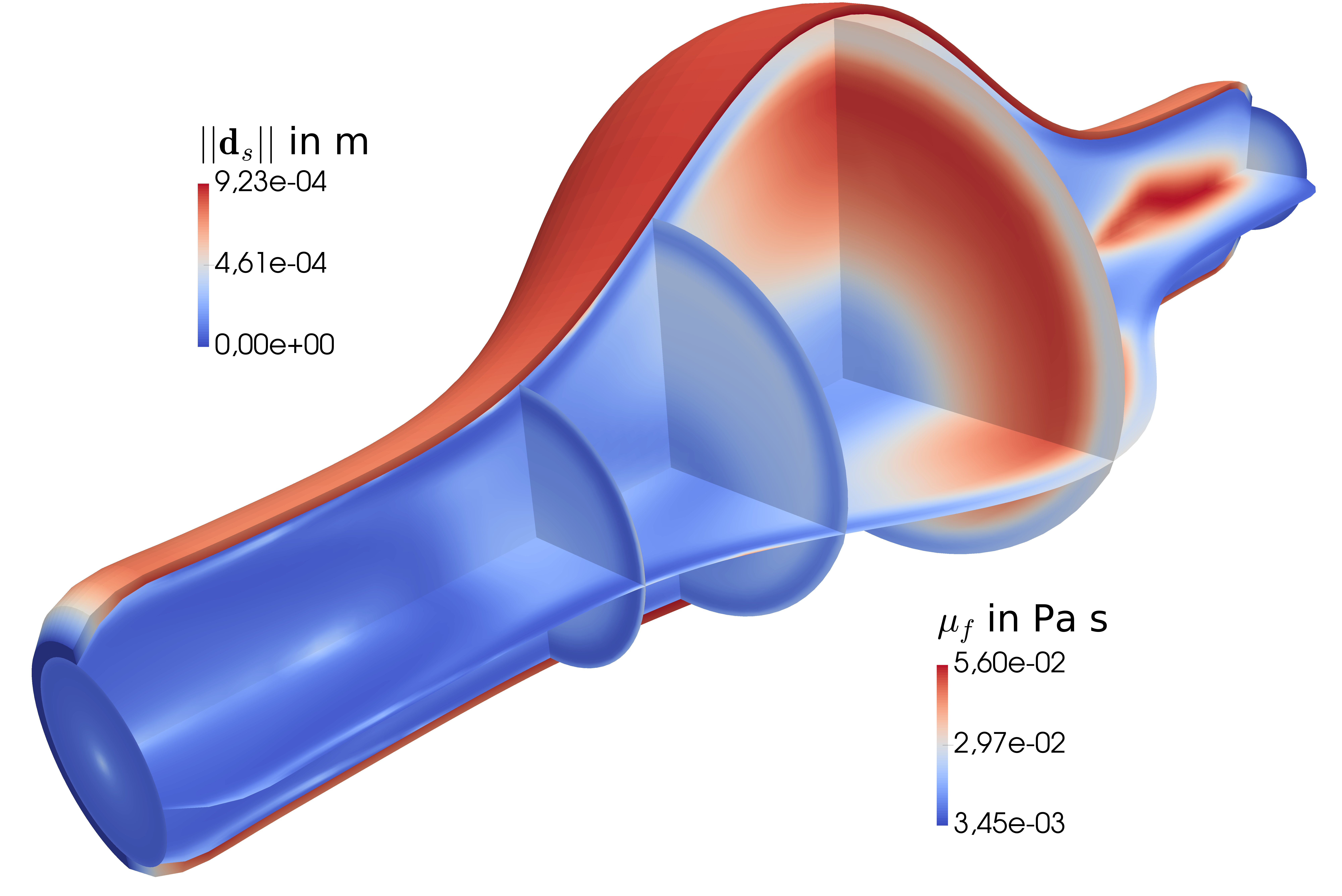}
			%\put(-40,75){$\hInt$}
		}
	\end{minipage}
	\\
	\begin{minipage}{.48\linewidth}
		\centering
		\subfloat[$\ve{d}_s$ and $\ve{u}_f$ at $t=2.6~\text{s}$.]{
			\label{fig:streamlines_259}
			\includegraphics[width=0.97\textwidth]{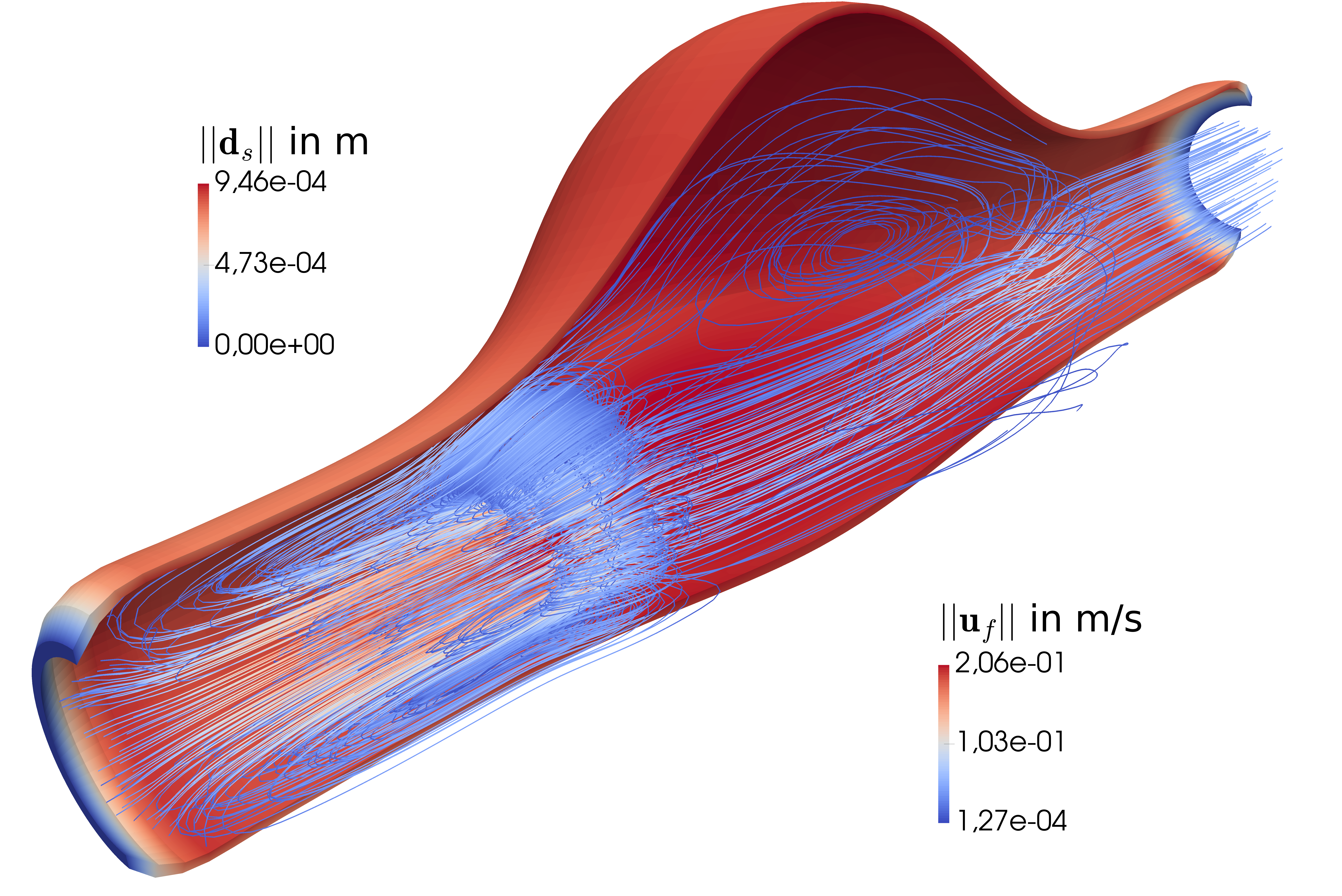}
			\put(-110,10){
			    \includegraphics[width=0.15\textwidth]{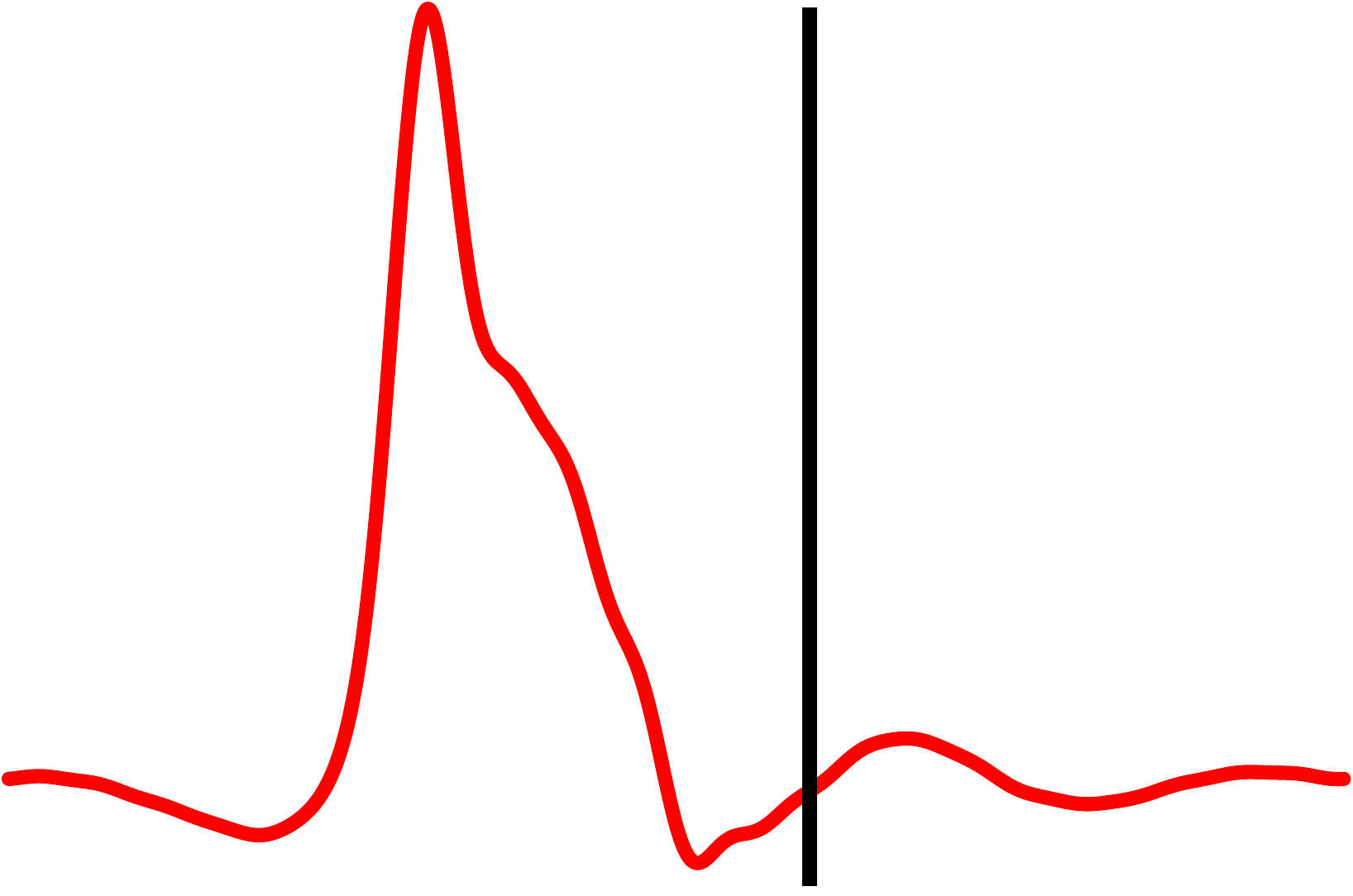}
			}
		}
	\end{minipage}%
	\hfil
	\begin{minipage}{.48\linewidth}
		\centering
		\subfloat[$\ve{d}_s$ and $\mu_f$ at $t=2.6~\text{s}$.]{
			\label{fig:viscosity_259}
			\includegraphics[width=0.97\textwidth]{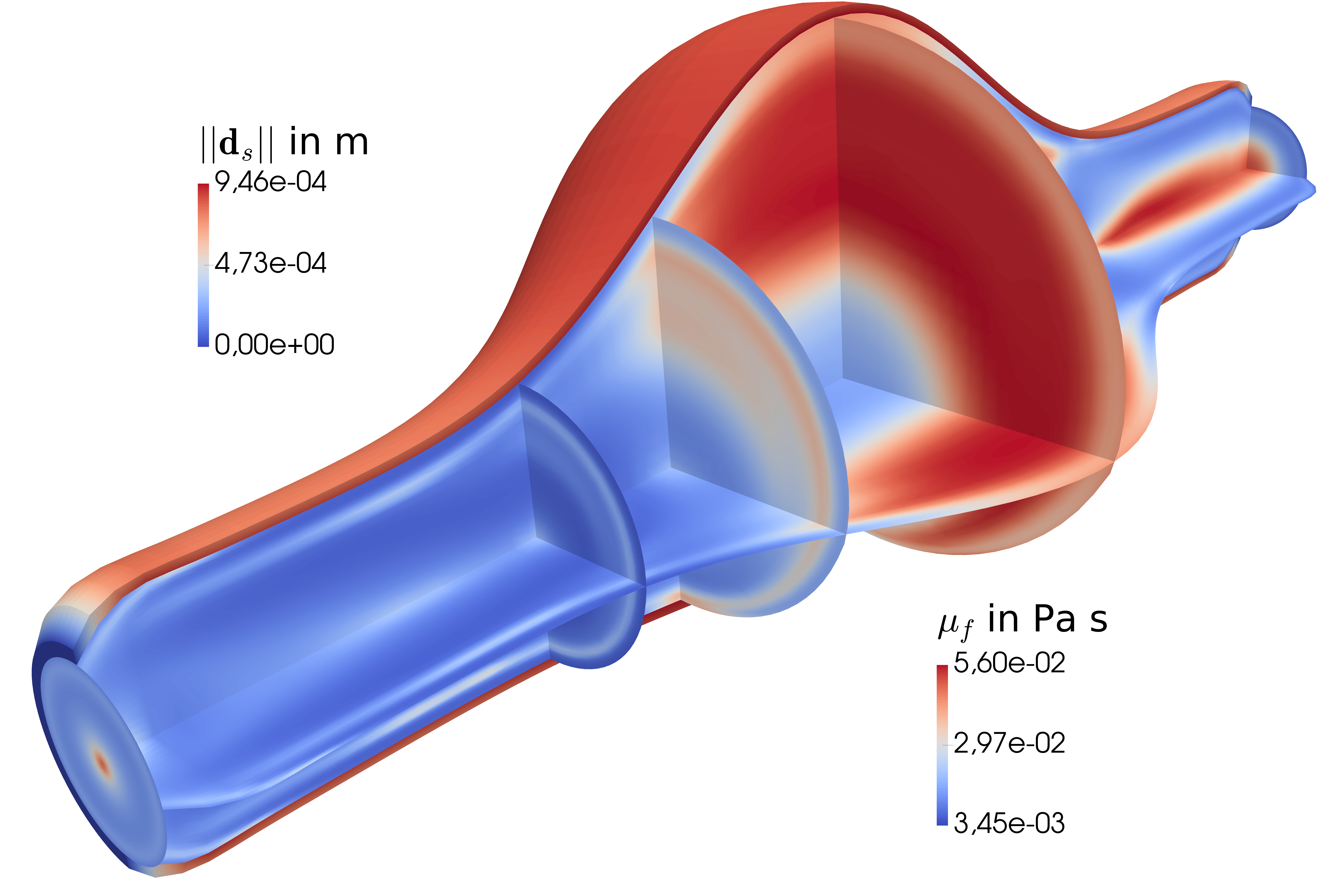}
			%\put(-40,75){$\hInt$}
		}
	\end{minipage}
	\caption{Snapshots at time $t=2.2,2.4,2.6~\text{s}$ \revised{indicated by a vertical line in the periodic mean inflow scaling:} Carreau fluid flowing through the abdominal aortic aneurysm of HGO material (deformation scaled by $5$): Solid displacement $\ve{d}_s$ and selected streamlines of the fluid velocity $\ve{u}_f$ (left) or viscosity $\mu_f$ in selected slices (right).}
	\label{fig:AAA_snapshots}
\end{figure}
\begin{figure}%[!htbp]
	\begin{minipage}{.48\linewidth}
		\centering
		\subfloat[$||\ve{d}_s||$ in the apex point.]{
			\label{fig:AAA_apex_d}
			\includegraphics[width=0.97\textwidth]{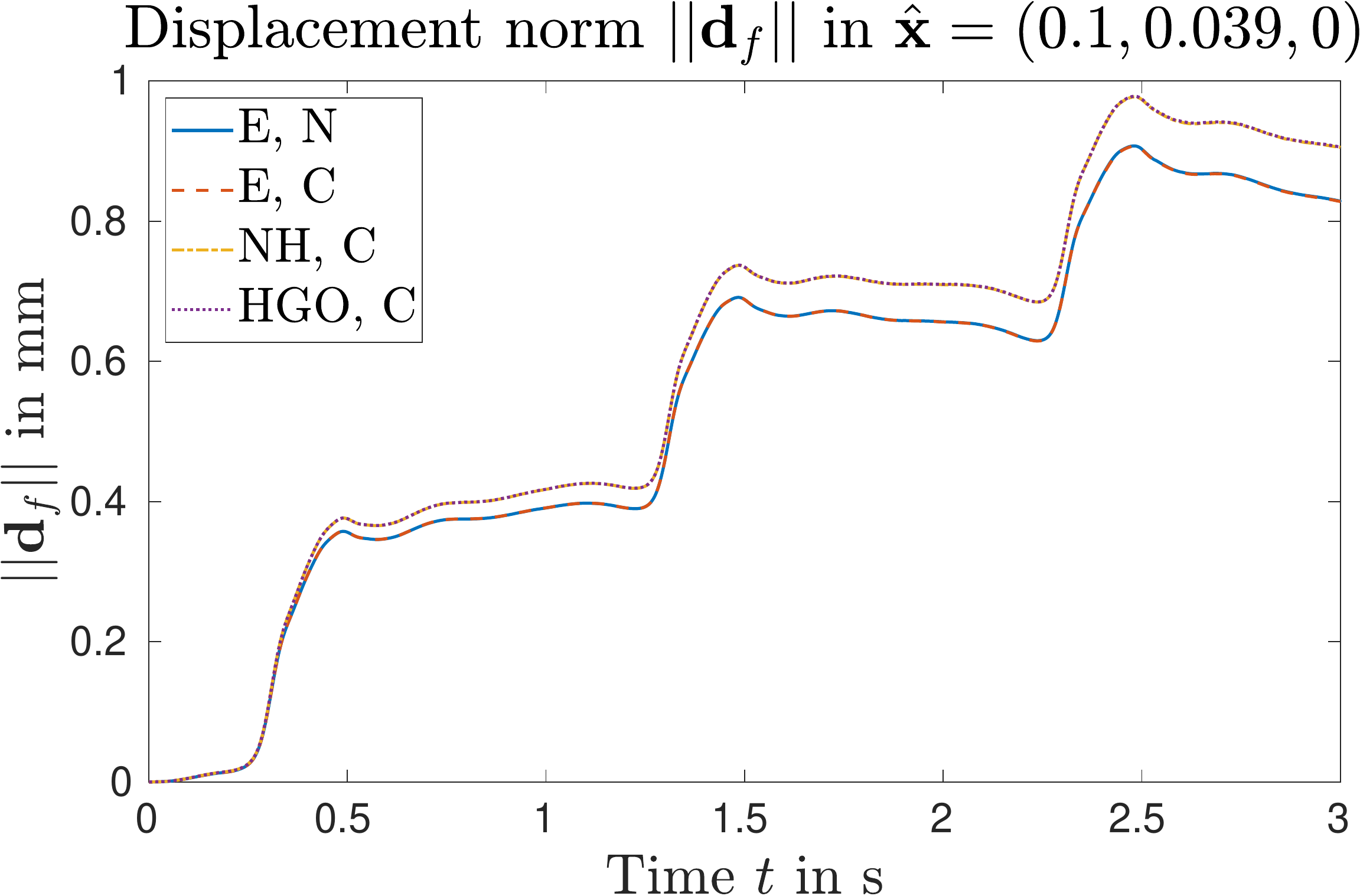}
			%\put(-10,72.5){$\hInt$}
		}
	\end{minipage}%
	\hfil
	\begin{minipage}{.48\linewidth}
		\centering
		\subfloat[$||p_f||$ in the apex point.]{
			\label{fig:AAA_apex_p}
			\includegraphics[width=0.97\textwidth]{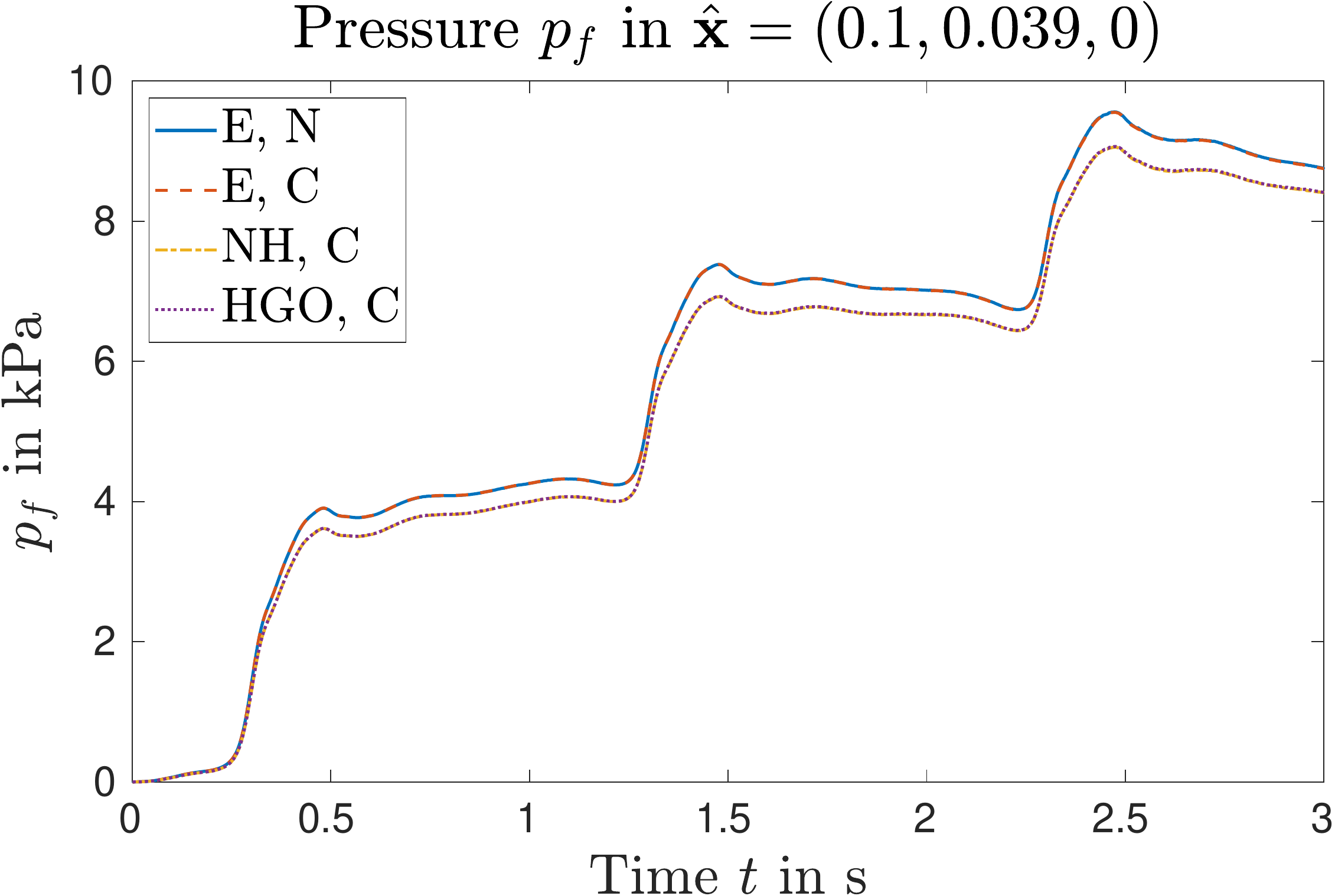}
			%\put(-40,75){$\hInt$}
		}
	\end{minipage}
	\caption{Solid displacement and fluid pressure in the apex point using linear elasticity~(E), neo-Hookean~(NH) or Holzapfel--Gasser--Ogden~(HGO) material models for the solid phase and Newtonian~(N) or Carreau~(C) fluids. The periodic state is not yet reached due to the Windkessel model. Effects of nonlinearities are small in the observed quantities.}
	\label{fig:AAA_compare}
\end{figure}

Having thoroughly inspected the solution itself, let us now turn our attention to the numerical aspects of this test case. Motivated by the previous example, the semi-implicit Dirichlet--Neumann (SIDN) coupling scheme is considered. Throughout the entire simulation time, less than 30 steps coupling solid momentum balance and PPE are needed, as shown in Figure~\ref{fig:AAA_fsi_iters}. Only a slight dependence on the inflow profile and pressure level in the geometry are observed. Interestingly, accumulated FSI iteration counts show that NH and HGO models lead to fewer FSI iterations, which is most likely triggered by higher pressure levels due to a stiffer material response caused by the selected (higher) Young's modulus. For comparison, the implicit Dirichlet--Neumann (IDN) scheme and a scheme treating the mesh motion explicitly, leading to a geometry explicit (GEDN) approach are also included in Figure~\ref{fig:AAA_fsi_iters_accum1}, showing a decrease of up to 31\% in FSI coupling steps needed when using the semi-implicit variant.

\begin{figure}%[!htbp]
	\begin{minipage}{.48\linewidth}
		\centering
		\subfloat[$\bar{\te{\tau}}$ on the fluid--structure interface $\Intt$.]{
			\label{fig:AAA_t_avg_tau}
			\includegraphics[width=0.97\textwidth]{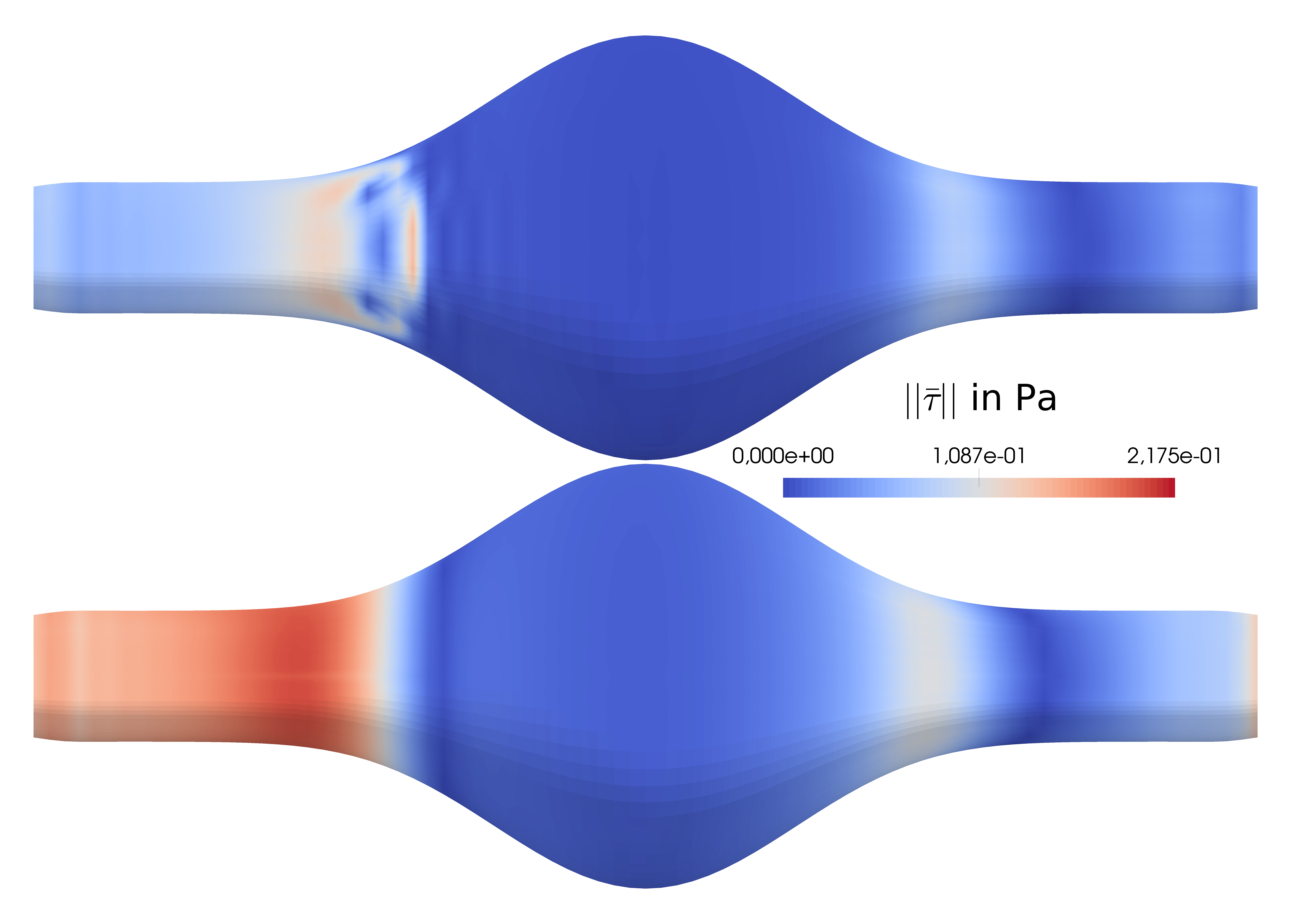}
			%\put(-40,75){$\hInt$}
		}
	\end{minipage}
	\hfil
	\begin{minipage}{.48\linewidth}
		\centering
		\subfloat[$\bar{\dot{\gamma}}$ in cut fluid domain $\Oft$.]{
			\label{fig:AAA_t_avg_gamma}
			\includegraphics[width=0.97\textwidth]{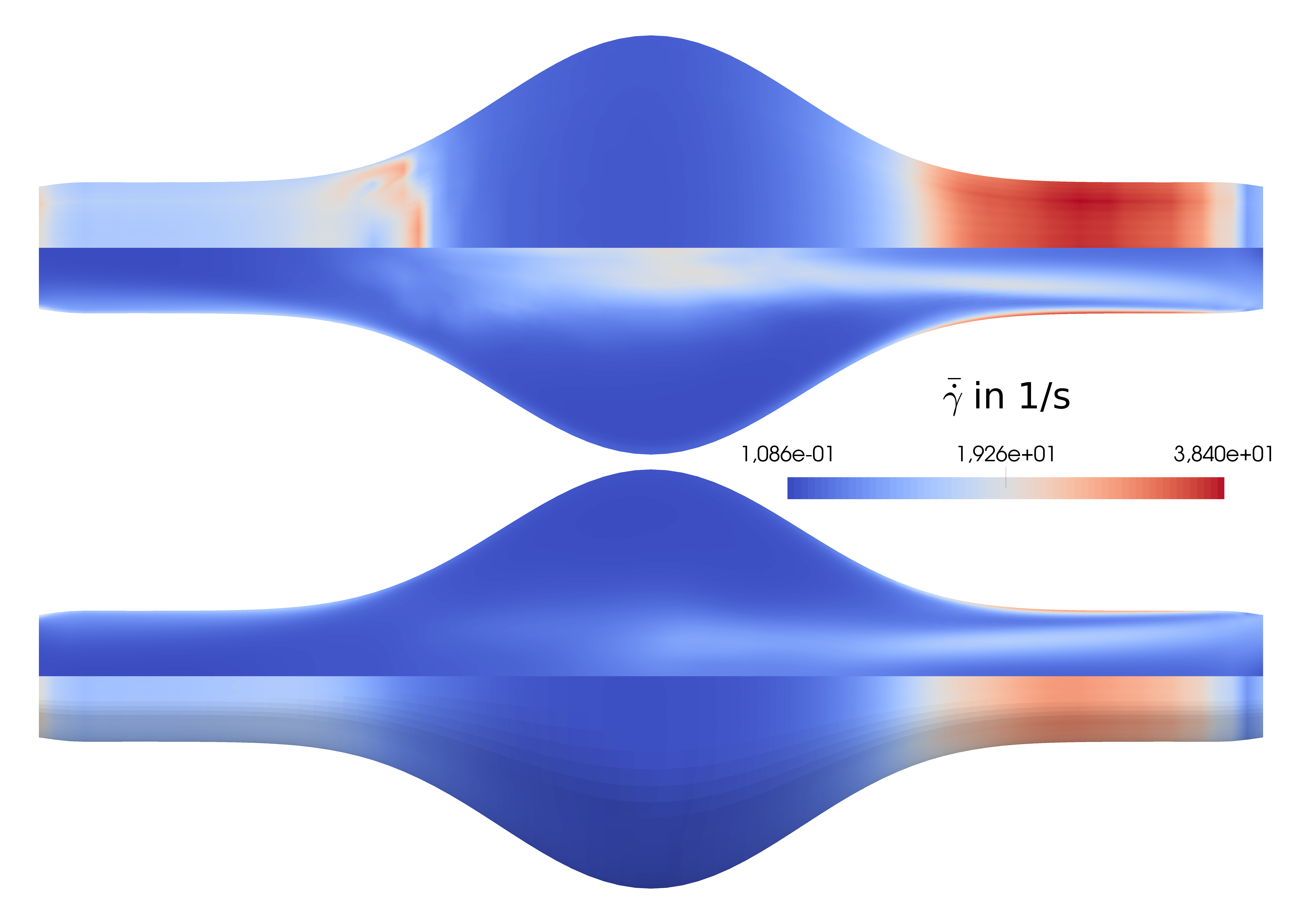}
			%\put(-10,72.5){$\hInt$}
		}
	\end{minipage}%
	\caption{Time averaged quantities in anterior-posterior view, inlet on the right: the symmetric solutions show vast differences using Newtonian (top row) and Carreau (bottom row) models.}
	\label{fig:AAA_compare_t_avg}
\end{figure}
\begin{figure}%[!htbp]
	\begin{minipage}{.48\linewidth}
		\centering
		\subfloat[Coupling steps per time step.]{
			\label{fig:AAA_fsi_iters_accum0}
			\includegraphics[width=0.97\textwidth]{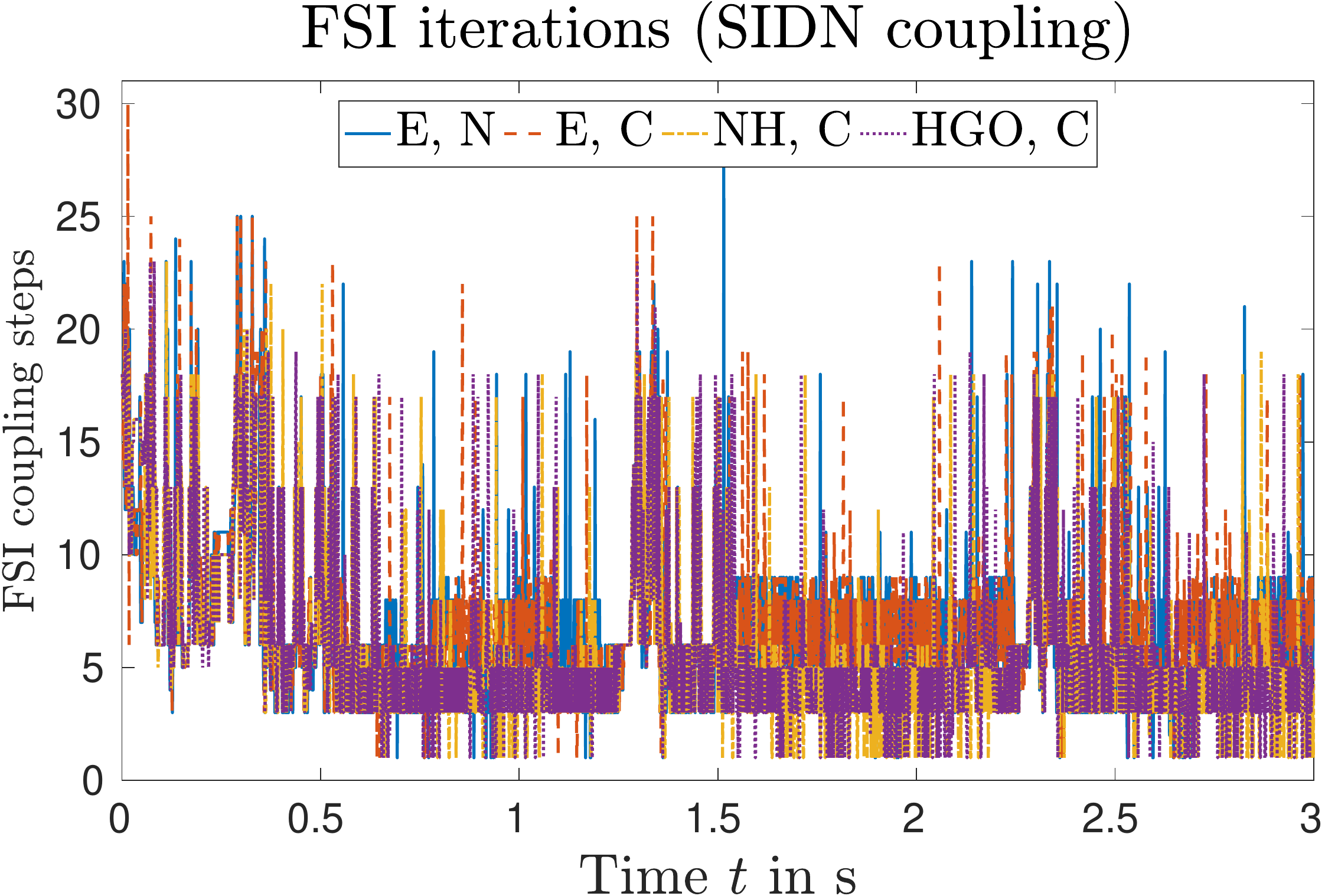}
			%\put(-40,75){$\hInt$}
		}
	\end{minipage}
	\hfil
	\begin{minipage}{.48\linewidth}
		\centering
		\subfloat[Accumulated coupling steps.]{
			\label{fig:AAA_fsi_iters_accum1}
			\includegraphics[width=0.97\textwidth]{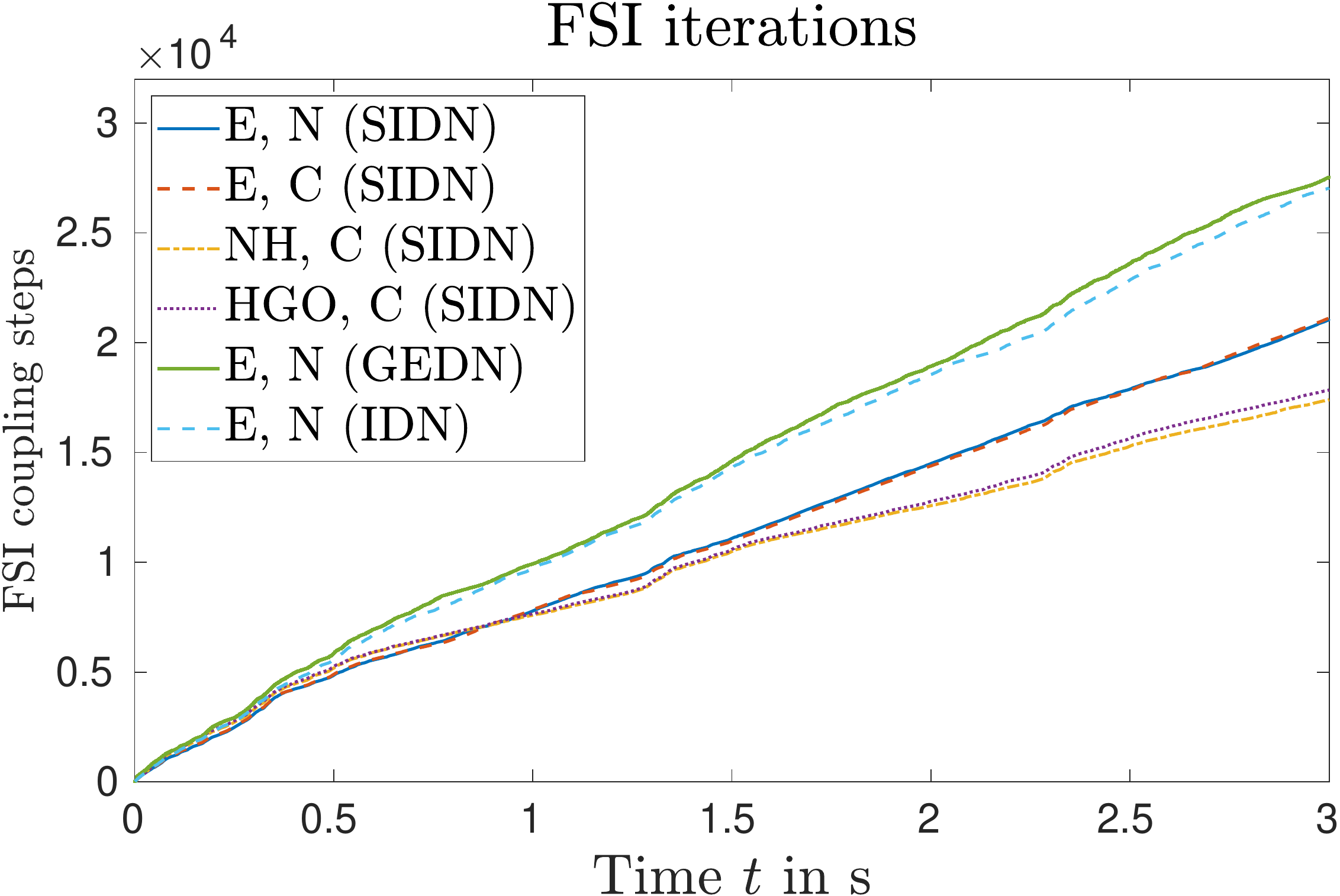}
			%\put(-10,72.5){$\hInt$}
		}
	\end{minipage}%
	\caption{Coupling steps needed using the semi-implicit (SIDN), implicit (IDN) or geometry explicit (GEDN) Dirichlet--Neumann schemes: iterations per time step (left) and accumulated (right) comparing linear elasticity~(E), neo-Hookean~(NH) or Holzapfel--Gasser--Ogden~(HGO) material models for the solid phase and Newtonian~(N) or Carreau~(C) fluids.}
	\label{fig:AAA_fsi_iters}
\end{figure}

Regarding the iteration counts in the subproblems, consider first systems solved only once per time step. Solving the linear system corresponding to the fluid momentum balance equations needs an almost constant number of 3 iterations to reach a relative error of $\leq 10^{-8}$ applying the AMG-preconditioned FGMRES method using a Chebyshev smoother due to small enough time steps. The iteration counts over time are depicted in Figure~\ref{fig:AAA_iters_uf_df} together with the iteration counts needed so solve the mesh motion equation with an AMG-preconditioned CG solver. The mass matrix corresponding to the viscosity projection step is lumped and therefore easily inverted, simply scaling the assembled right-hand side by a vector representing the diagonal. The iteration counts for the Leray projection step stay nicely bounded around $\approx 35$. For brevity, we do not show a corresponding plot, since it behaves similar to the PPE and is only computed once per time step -- in fact, one can combine PPE and Leray projection steps of the past velocities when using divergence damping, solving only one Poisson problem for a combined variable. This detail does barely pay off using the SIDN scheme, which is why we refer to \cite{Pacheco2021b} for further details.  

\begin{figure}%[!htbp]
	\begin{minipage}{.48\linewidth}
		\centering
		\subfloat[FGMRES steps needed for linear solver convergence in the fluid momentum balance.]{
			\includegraphics[width=0.97\textwidth]{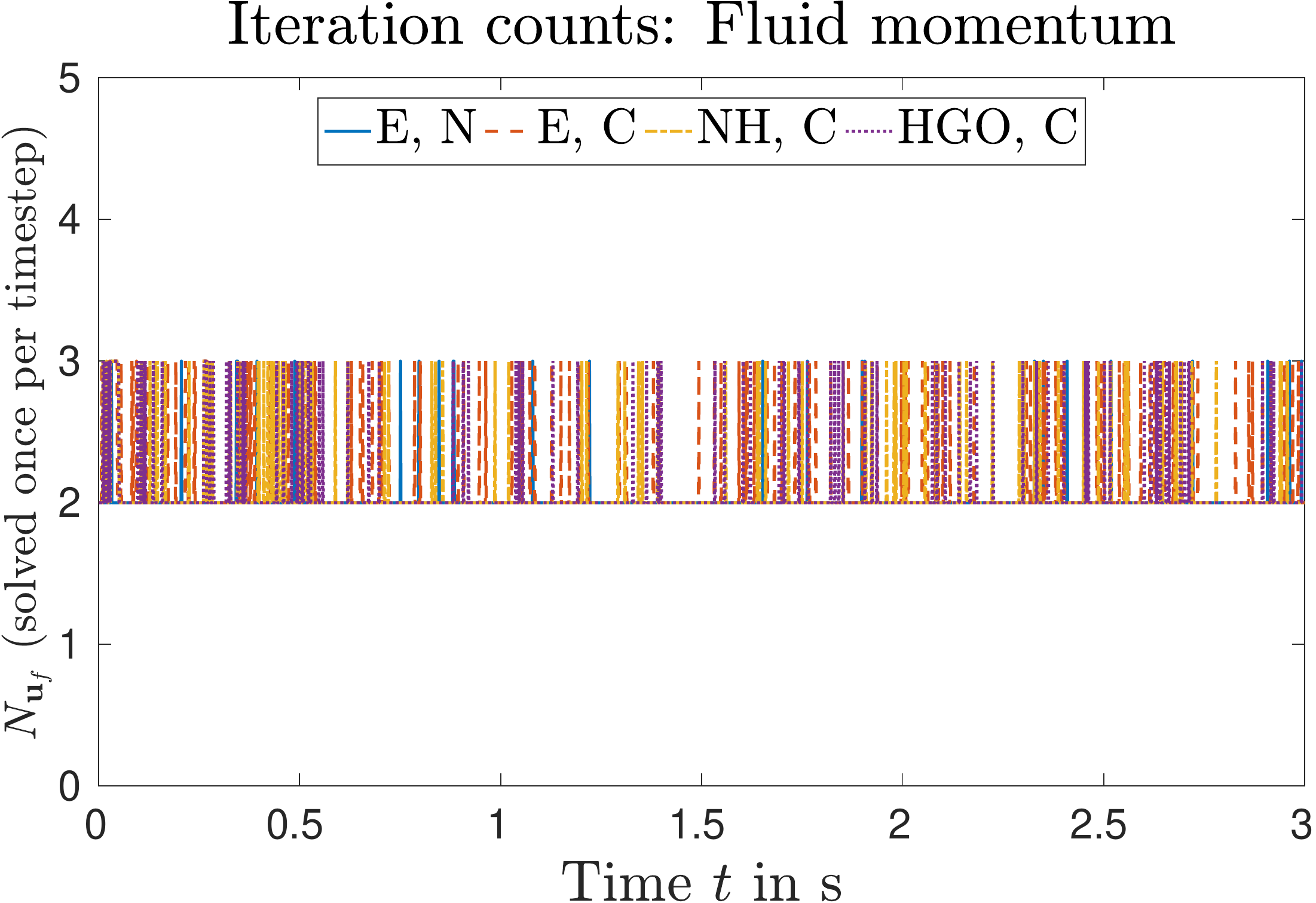}
			%\put(-40,75){$\hInt$}
		}
	\end{minipage}
	\hfil
	\begin{minipage}{.48\linewidth}
		\centering
		\subfloat[CG steps needed for linear solver convergence in the mesh motion equation.]{
			\includegraphics[width=0.97\textwidth]{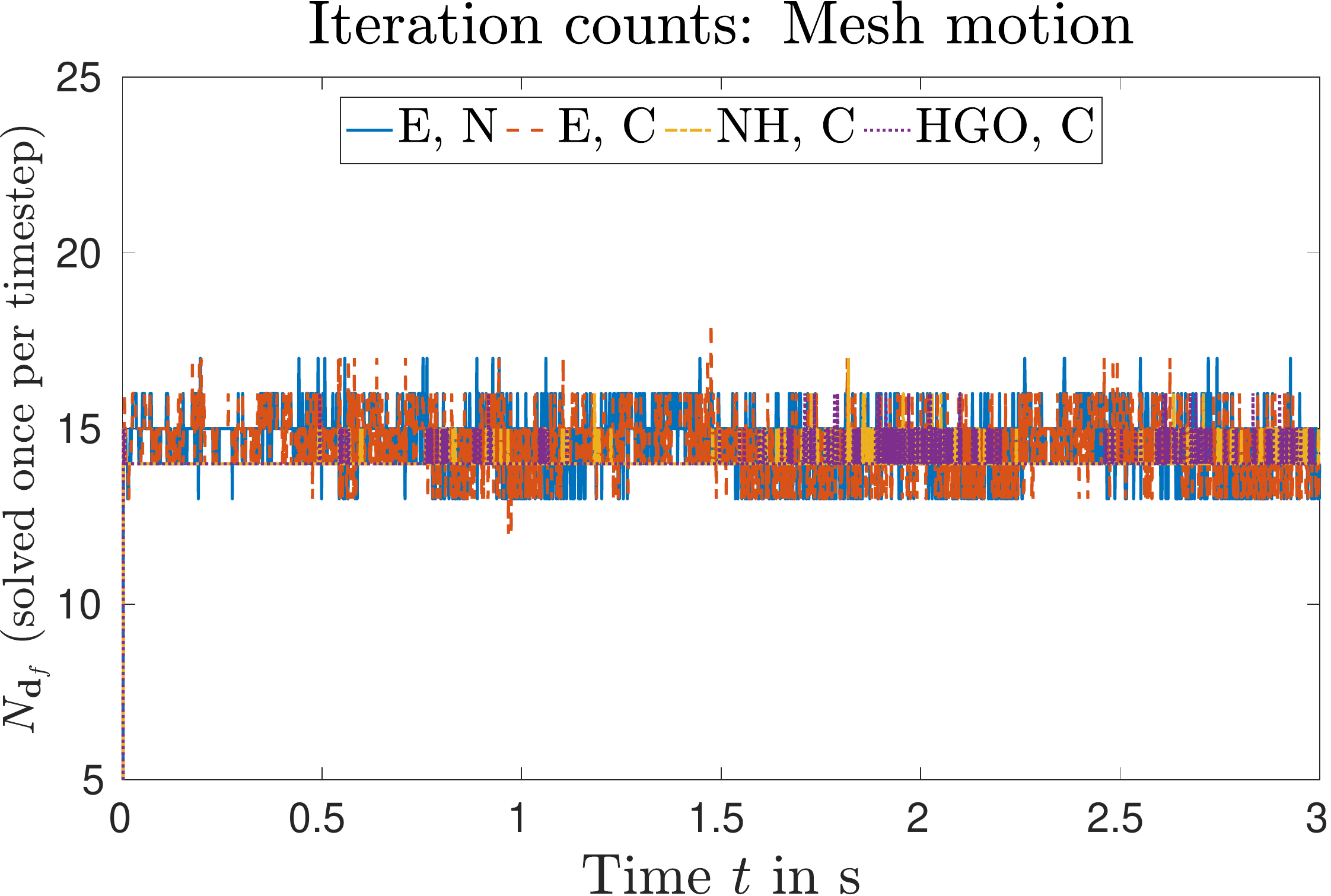}
			%\put(-10,72.5){$\hInt$}
		}
	\end{minipage}%
	\caption{Iteration counts in the fluid momentum balance (left) and mesh motion (right) solvers using linear elasticity~(E), neo-Hookean~(NH) or Holzapfel--Gasser--Ogden~(HGO) material models for the solid phase, Newtonian~(N) or Carreau~(C) fluids and the SIDN scheme.}
	\label{fig:AAA_iters_uf_df}
\end{figure}

For the overall performance, the steps within the semi-implicit coupling loop, namely the PPE and solid momentum balance solves are dominant in terms of computational effort. Thus, the good performance of the linear solvers as depicted in Figure~\ref{fig:AAA_pf_ds} is crucial. We observe little dependence on the fluid velocity and pressure level, given the relatively small time step size. Additionally, nonlinear convergence in the Newton solver is reached in any time step within 3 iterations, which is also a result of the small time step size and the quadratic extrapolation used as an initial guess. Comparing the constitutive models, the linear solvers show equal behaviour independent of the fluid model employed. The solid model, however, heavily influences the iteration count in the linear system solve, resulting from the more complex terms and matrices when considering NH or HGO models.  

\begin{figure}%[!htbp]
	\begin{minipage}{.48\linewidth}
		\centering
		\subfloat[CG steps needed for linear solver convergence in the PPE.]{
			\includegraphics[width=0.97\textwidth]{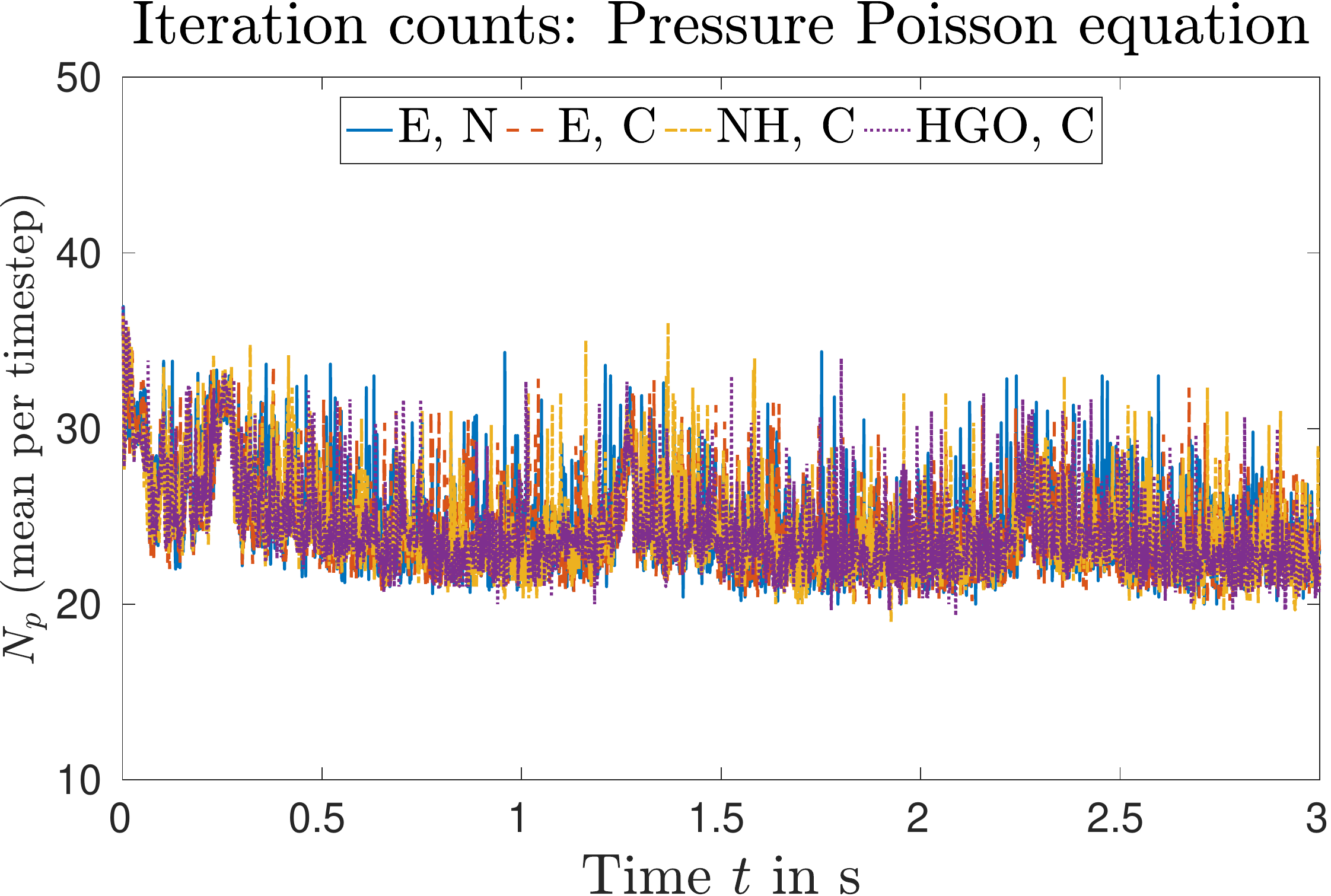}
			%\put(-40,75){$\hInt$}
		}
	\end{minipage}
	\hfil
	\begin{minipage}{.48\linewidth}
		\centering
		\subfloat[FGMRES steps needed for linear solver convergence in the solid momentum balance.]{
			\includegraphics[width=0.97\textwidth]{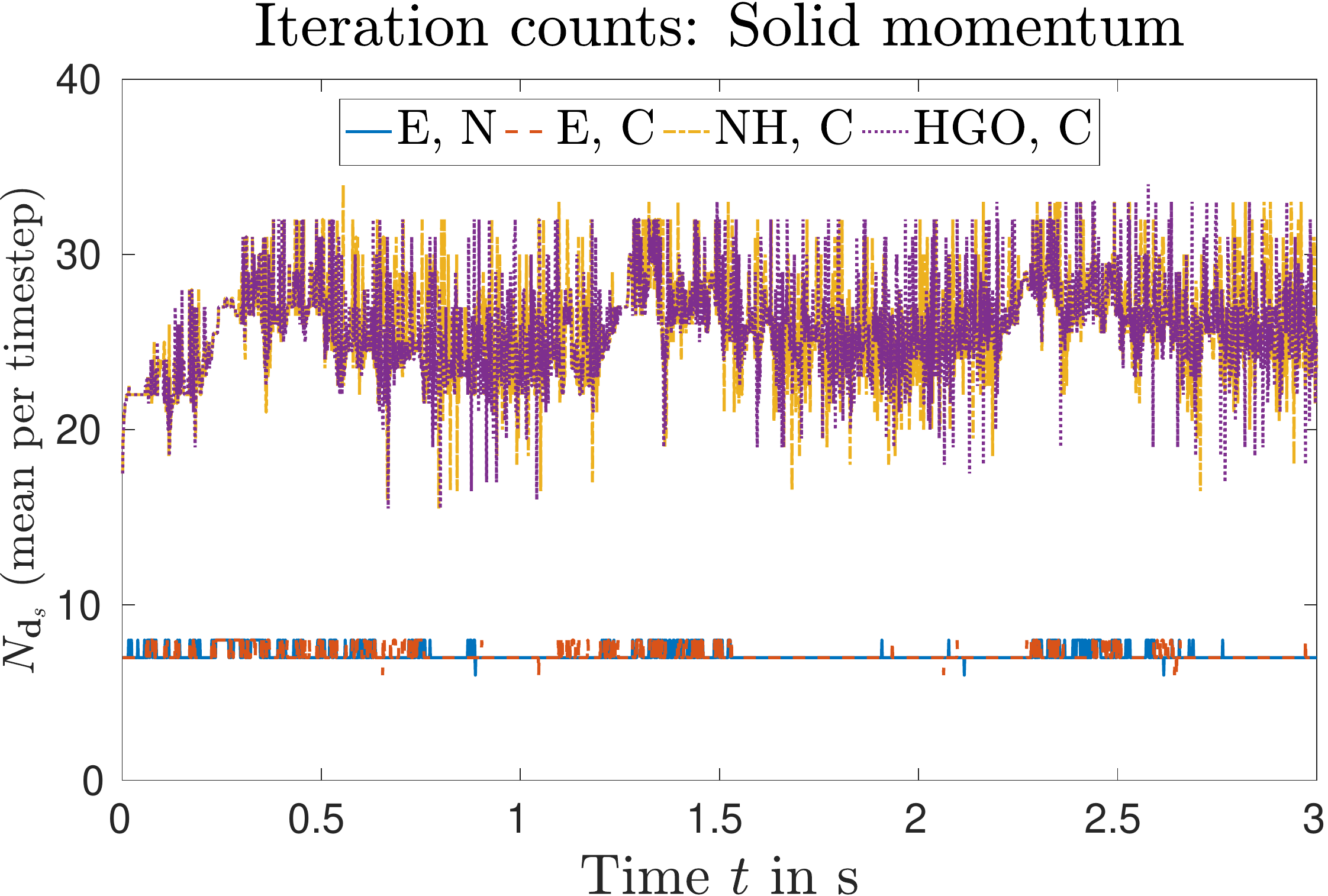}
			%\put(-10,72.5){$\hInt$}
		}
	\end{minipage}%
	\caption{Iteration counts in the PPE (left) and solid momentum balance (right) solvers using linear elasticity~(E), neo-Hookean~(NH) or Holzapfel--Gasser--Ogden~(HGO) material models for the solid phase, Newtonian~(N) or Carreau~(C) fluids and the SIDN scheme.}
	\label{fig:AAA_pf_ds}
\end{figure}

Summing up, the solutions of the linear systems corresponding to each of the subproblems are obtained within a satisfactory number of iterations and show low iteration counts given the relevant parameter ranges including realistic Reynold's numbers and tissue stiffness. As a last note on the computational performance of the implementation, compare the absolute and relative time spent per time step using different constitutive models. As can be seen from Table~\ref{tab:AAA_compare_time_table}, the number of FSI coupling steps needed for convergence are slightly higher for stiffer material parameters, while the viscosity projection step is negligible in terms of computational effort. The solid constitutive models, however, do have a strong influence on the time spent per time step. This additional effort is caused by both the increased complexity of element integration and $\approx2.5\times$ more iterations in the linear system solve (see also Figure~\ref{fig:AAA_pf_ds}). The SIDN variant outperforms the GEDN and IDN schemes being at least twice as fast overall, needing less FSI iterations and less time per coupling step. Interestingly, the number of coupling steps increases when using the GEDN scheme, but individual steps and the overall time interval are completed faster compared to the IDN scheme, since the mesh motion equation is only solved once per time step.

\begin{table}%[ht!]
	\centering
	\begin{tabular}{||l|c|c|c||} 
		\hline
		&&&\\[-2ex]
		& computing time & FSI steps & time/FSI step \\
		\hline\hline
		&&&\\[-1ex]
		E, N (SIDN)   & $5.53\times10^5~\text{s}$ (100.0\%) & 21046 (121.0\%) & $26.3~\text{s}$ (100.0\%)\\[0.5ex]
		E, C (SIDN)   & $5.57\times10^5~\text{s}$ (100.1\%) & 21118 (121.4\%) & $26.4~\text{s}$ (100.4\%)\\[0.5ex]
		NH, C (SIDN)  & $9.22\times10^5~\text{s}$ (166.7\%) & 17395 (100.0\%) & $53.0~\text{s}$ (201.5\%)\\[0.5ex]
		HGO, C (SIDN) & $1.09\times10^6~\text{s}$ (196.4\%) & 17846 (102.6\%) & $56.5~\text{s}$ (214.8\%)\\[0.5ex]
		E, N (GEDN)   & $1.07\times10^6~\text{s}$ (193.5\%) & 27532 (158.3\%) & $38.9~\text{s}$ (147.9\%)\\[0.5ex]
		E, N (IDN)    & $1.41\times10^6~\text{s}$ (255.0\%) & 27045 (155.5\%) & $52.3~\text{s}$ (198.9\%)
		\\[1.4ex]
		\hline
	\end{tabular}
	\caption{Absolute (relative) computing times using linear elasticity~(E), neo-Hookean~(NH) or Holzapfel--Gasser--Ogden~(HGO) material models for the solid phase, Newtonian~(N) or Carreau~(C) fluids and the semi-implicit~(SIDN), implicit~(IDN) or geometry explicit~(GEDN) Dirichlet--Neumann coupling schemes.}
	\label{tab:AAA_compare_time_table}
\end{table}
\begin{table}%[ht!]
	\centering
	\begin{tabular}{||c|c|c|c|c|c|c||} 
		\hline
		&&&&&&\\[-2ex]
		& $\ve{d}_f$ & $\ve{u}_f$ & $\psi$ & $\zeta$ & $p_f$ & $\ve{d}_s$ \\
		\hline\hline
		&&&&&&\\[-1ex]
		E, N   & 6.0 & 8.0 & 4.8 & 2.6 & (6.4+4.2+19.7) & (20.3+4.9+16.2) \\[0.5ex]
		E, C   & 5.9 & 8.0 & 4.9 & 2.6 & (6.3+4.2+19.5) & (20.2+4.9+16.2) \\[0.5ex]
		NH, C  & 2.3 & 2.2 & 3.3 & 0.8 & (1.6+2.6+14.5) & (29.5+5.1+35.7) \\[0.5ex]
		HGO, C & 1.6 & 1.7 & 2.2 & 0.6 & (1.3+1.9+10.4) & (47.6+3.7+26.9) 
		\\[1.4ex]
		\hline
	\end{tabular}
	\caption{Relative computing times of relevant steps, indicated by the unknown solved for in the SIDN scheme: in \%, split into (assembly+AMG setup+linear solve) for PPE and solid momentum balance using linear elasticity~(E), neo-Hookean~(NH) or Holzapfel--Gasser--Ogden~(HGO) material models for the solid phase and Newtonian~(N) or Carreau~(C) fluids.}
	\label{tab:AAA_compare_iter_steps}
\end{table}

Timings for the computationally relevant steps within the SIDN coupling scheme are further listed in Table~\ref{tab:AAA_compare_iter_steps}, where we notice the increased (relative) time spent in those steps. Also, this demonstrates how little time is spent on the semi-implicit steps, being the mesh motion equation ($\ve{d}_f$), fluid momentum balance ($\ve{u}_f$), Leray projection ($\psi$) and viscosity projection, the last one of which is not even listed since the time spent is less than $0.5\%$. The pressure boundary projection ($\zeta$) is performed in each coupling step, but still does not lead to a major computational load. Solving the PPE ($p_f$) and solid momentum balance ($\ve{d}_s$) has the largest influence on the overall time spent in the SIDN scheme, where one may further optimise element integration routines and linear solvers. Nonetheless, iteration counts are nicely bounded, indicating good performance of the chosen setup. Parallel scalability of the overall algorithm is not investigated at this point, but aspect of future investigations.

\section{Concluding remarks}

Within this work, we presented a family of coupling schemes involving incompressible viscous flows interacting with three-dimensional solids. The standard Navier--Stokes equations were solved by decoupling velocity and pressure via an equivalent time-splitting or split-step scheme in arbitrary Lagrangian-Eulerian formulation, based on a pressure Poisson equation with consistent boundary conditions. Consequently, substantial parts of the algorithm can be treated semi-implicitly in an added-mass-stable way without degrading accuracy. The mesh motion equation, fluid momentum balance equation, Leray projection (included for improved mass conservation) and viscosity projection are only solved once per time step. Generalised Newtonian rheological laws are exchanged effortlessly, while also allowing for equal order, $C^0$-continuous interpolation with Lagrangian finite elements. Moreover, we considered a general setup including various constitutive models applied for the solid, demonstrating the flexibility of the framework and robustness of the coupling algorithm. The accuracy of the scheme is assessed comparing to analytical solutions in space, applying $Q_2/Q_1$ and $Q_1/Q_1$ finite element pairs, and in time, combining BDF-2 with generalised-$\alpha$ time integration schemes. The framework is then further tested in a benchmark example in the context of arterial blood flow and an idealised abdominal aortic aneurysm, highlighting its great performance even with black-box solvers and preconditioners.

\revised{At this point, some open questions remain to be further analysed in the future: First, we note that owing to the extrapolation of the fluid pressure in time needed in the split-step scheme for the fluid problem, the restrictions on the time step size grow stronger with increasing the temporal order of accuracy. Thus, using high-order time stepping schemes beyond the presented second-order time integrators might reduce the maximum allowed step size drastically. %Second, note that the time-splitting scheme applied for the fluid only delivers the same convergence rate in the $H^1$ norm of the velocity error as a $Q_2/Q_1$ pair in a coupled velocity-pressure formulation, but not in the $L^2$ norm \textcolor{red}{(I would remove this sentence. I think it might bring more confusion than clarity)}.
Moreover, when using even higher polynomial orders, the splitting error of the scheme itself might dominate spatial higher-order accuracy as is observed, e.g., for projection-schemes. Second, we want to mention that the Robin variants of the coupling scheme did not perform as expected in our studies. Depending on the parameters, the considered Robin variants can substantially decrease the number of coupling steps needed, but the correct settings remained elusive, demanding further investigation to not spoil overall stability. That being said, the acceleration scheme used herein is Aitken's relaxation for simplicity, but much more advanced schemes can be applied straight-forwardly. In this regard, first tests with the interface quasi-Newton inverse least-squares method \cite{Degroote2009,Degroote2013} are very promising.}

Future and ongoing work is centered around scalability tests and algorithmic optimisation to further reduce runtimes. From a modelling perspective, other constitutive equations and mixed/hybrid formulations for the solid phase will be included, harnessing the partitioned design. Moreover, we are testing semi-implicit coupling to thrombus formation models, which is a central element in various cardiovascular conditions.

\section*{Acknowledgements}
The authors gratefully acknowledge Graz University of Technology for the financial support of the Lead-project: Mechanics, Modeling and Simulation of Aortic Dissection.

\section*{Competing interests}
The authors declare that they have no known competing financial interests or personal relationships that could have appeared to influence the work reported in this paper.

\section*{Appendix}
Proof of equivalence of systems \eqref{eqn:fluid_splitstep_mom}--\eqref{eqn:fluid_splitstep_BC_pressure_neumann} and \eqref{eqn:fluid_u_p_mom}--\eqref{eqn:fluid_u_p_robin} for sufficiently regular $p_f, \ve{u}_f, \ve{g}_f, \ve{t}_f, \ve{h}_f$ (Theorem~\ref{theorem:equivalence}).

\begin{proof}
	We start by showing first that 
	\eqref{eqn:fluid_u_p_mom}--\eqref{eqn:fluid_u_p_robin} imply \eqref{eqn:fluid_splitstep_mom}--\eqref{eqn:fluid_splitstep_BC_pressure_neumann}, where the stress term in the momentum balance equation \eqref{eqn:fluid_u_p_mom} is simply rewritten according to~\eqref{eqn:fluid_full_sigma_f} to obtain \eqref{eqn:fluid_splitstep_mom}. The Dirichlet conditions on the pressure~\eqref{eqn:fluid_splitstep_BC_pressure_dirichlet_onGNf} is obtained by \revised{taking the scalar product of} the traction boundary condition~\eqref{eqn:fluid_BC_full_traction} and $\nf$ first
	\begin{align*}
		\nf \cdot \left(
		\te{\sigma}_f \nf - \ve{t}_f
		\right) = \ve{0}
		\quad
		\Leftrightarrow
		\quad
		\nf \cdot \left(2\mu_f\nabla^S\ve{u}_f \nf \right) - \nf\cdot\ve{t}_f &= \nf \cdot \left(p_f\te{I}\nf\right) = p_f
	\end{align*}
	and then subtracting $\mu_f \nabla\cdot\ve{u}_f$ from the left hand side, which is admissible due to \mbox{$\nabla\cdot\ve{u}_f = 0$}. Analogously, a boundary condition based on the Robin condition \eqref{eqn:fluid_u_p_robin} is derived.
	The Neumann condition on the pressure~\eqref{eqn:fluid_splitstep_BC_pressure_neumann} results from \revised{taking the scalar product of} the momentum balance equation~\eqref{eqn:fluid_u_p_mom} and $\nf$ and using
	\begin{align}
		-\nabla \cdot \te{\sigma}_f &= -\nabla\cdot\left\{ -\left(p_f\te{I}\right) + \mu_f \left[ \nabla \ve{u}_f + \left( \nabla \ve{u}_f \right)^T\right]\right\} \nonumber\\
		& = \nabla p_f - \mu_f \nabla \left(\nabla \cdot \ve{u}_f\right) - \mu_f \Delta \ve{u}_f - 2 \nabla^S\ve{u}_f\nabla\mu_f\nonumber\\
		& = \nabla p_f - \mu_f \Delta \ve{u}_f - 2 \nabla^S\ve{u}_f\nabla\mu_f
		\label{eqn:div_sigma_rewrite}
	\end{align}
	for $\nabla \cdot \ve{u}_f = 0$ together with
	\begin{align}
		\Delta \ve{u}_f 
		\equiv\nabla\left(\nabla\cdot\ve{u}_f\right) 
		- \nabla\times\left(\nabla \times \ve{u}_f\right)
		= - \nabla\times\left(\nabla \times \ve{u}_f\right)
		\,,
		\label{eqn:laplace_u_is_minus_curlcurl_u}
	\end{align}
	which, restricted to $\GDf$, directly gives~\eqref{eqn:fluid_splitstep_BC_pressure_neumann}. Similarly, taking (minus) the divergence of the balance of linear momentum, we get
	\begin{align*}
		\ve{0}
		&=
		-
		\rho_f \nabla\cdot \left[
		\dtale{\ve{u}_f}
		+
		\nabla \ve{u}_f \left(\ve{u}_f - \ve{u}_m \right)
		\right]
		+
		\nabla\cdot \left(
		\nabla\cdot \te{\sigma}_f 
		\right) 
		\\
		&= 
		\rho_f 
		\dtale{\left(\nabla \cdot\ve{u}_f \right)}
		-
		\nabla\cdot
		\left[ 
			\rho_f \nabla \ve{u}_f \left(\ve{u}_f - \ve{u}_m\right)
			-
			2\nabla^S\ve{u}_f\nabla\mu_f
		\right]
		+ \nabla\cdot \left(
		\mu_f \Delta \ve{u}_f\right)
		-\Delta p_f
		\\
		&= 
		-
		\nabla\cdot
		\left[ 
		\rho_f \nabla \ve{u}_f \left(\ve{u}_f - \ve{u}_m\right)
		-
		2\nabla^S\ve{u}_f\nabla\mu_f
		\right]
		+ \nabla\cdot \left(
		\mu_f \Delta \ve{u}_f\right)
		-\Delta p_f
		\, ,
	\end{align*}
	but with \eqref{eqn:laplace_u_is_minus_curlcurl_u} we can rewrite
	\begin{align}
		\nabla\cdot\left(\mu_f\Delta\ve{u}_f\right) 
		=
		-\nabla\cdot\left[
		\mu_f \nabla\times\left(\nabla\times\ve{u}_f\right)
		\right]
		=
		-\left[
		\nabla\times\left(\nabla\times\ve{u}_f\right)
		\right]\cdot \nabla \mu_f
		\, , 
		\label{eqn:div_mu_laplace_u_is_minus_curl_curl_u_dot_grad_mu}		
	\end{align}
	which then finally gives Equation~\eqref{eqn:fluid_splitstep_ppe}. The additional Equation~\eqref{eqn:fluid_split_step_div_u0} is simply the continuity equation at $t=0$, thus completing the first part. To finish the proof of equivalence, we start from the other direction, i.e., show that \eqref{eqn:fluid_splitstep_mom}--\eqref{eqn:fluid_splitstep_BC_pressure_neumann} imply \eqref{eqn:fluid_u_p_mom}--\eqref{eqn:fluid_u_p_robin} by first taking minus the divergence of~\eqref{eqn:fluid_splitstep_mom} and adding it to~\eqref{eqn:fluid_splitstep_ppe}, which gives
	\begin{gather*}
	    -\rho_f \nabla \cdot \left( \dtale{\ve{u}_f}  \right)
	    +\nabla\cdot \left[ \nabla \cdot \left( 2 \mu_f \nabla^S \ve{u}_f \right) \right]
	    - \nabla \cdot \left[ 2\nabla^S\ve{u}_f\nabla \mu_f \right]
	    + \left[\nabla\times\left(\nabla\times\ve{u}_f\right)\right] \cdot \nabla \mu_f = 0
	    \, ,
	\end{gather*}
	where we can further use relations similar to~\eqref{eqn:laplace_u_is_minus_curlcurl_u} and ~\eqref{eqn:div_sigma_rewrite} to get
	\begin{gather*}
	    -\rho_f \nabla \cdot \left( \dtale{\ve{u}_f}  \right)
	    +\nabla\cdot \left[ \mu_f\nabla(\nabla\cdot\ve{u}_f) + \mu_f \Delta \ve{u}_f \right]
	    +\left[ \nabla (\nabla \cdot\ve{u}_f) - \Delta \ve{u}_f \right] \cdot \nabla \mu_f
	    =0
	    \, ,
	\end{gather*}
	where we use~\eqref{eqn:laplace_u_is_minus_curlcurl_u} and~\eqref{eqn:div_mu_laplace_u_is_minus_curl_curl_u_dot_grad_mu}, but without inserting $\nabla\cdot\ve{u}_f=0$, to finally arrive at
	\begin{gather}
	    \dtale{(\nabla \cdot \ve{u}_f)}
	    - \nabla \cdot \left[ 2 \frac{\mu_f}{\rho_f} \nabla (\nabla \cdot \ve{u}_f) \right]
	    = 0
	    \, ,
	    \label{eqn:auxiliary_heat_eqn}
	\end{gather}
	being a heat equation in the new variable $\Phi := \nabla \cdot \ve{u}_f$ in ALE form. Neumann conditions for~\eqref{eqn:auxiliary_heat_eqn} are obtained by \revised{taking the scalar product of}~\eqref{eqn:fluid_splitstep_mom} and $\ve{n}_f$ and adding the result to~\eqref{eqn:fluid_splitstep_BC_pressure_neumann}, which gives
	\begin{align*}
	    \ve{n}_f \cdot 
	    \left\{ 
	        - \nabla \cdot (2 \mu_f \nabla^S \ve{u}_f) 
	        - \mu_f \left[ \nabla \times (\nabla \times \ve{u}_f) \right]
	        + 2 \nabla^S \ve{u}_f
	    \right\} = 0
	    \, .
	\end{align*}
	Further using once again ~\eqref{eqn:laplace_u_is_minus_curlcurl_u} and ~\eqref{eqn:div_sigma_rewrite}, we end up with
	\begin{align}
		0 = \ve{n}_f \cdot 
	         \left[ -2 \mu_f \nabla (\nabla \cdot \ve{u}_f )\right]
		   \quad \therefore \quad 
		   \ve{n}_f \cdot \nabla \Phi = 0 \quad \text{on } \GDf \, .
		   \label{eqn:auxiliary_heat_eqn_BC_dirichlet}
	\end{align}
	Dirichlet conditions for~\eqref{eqn:auxiliary_heat_eqn} are constructed from~\eqref{eqn:fluid_splitstep_BC_pressure_dirichlet_onGNf} and~\eqref{eqn:fluid_splitstep_BC_pressure_dirichlet_onGRf}, which after inserting the definition of $\ve{t}_f$ and $\ve{h}_f$, respectively, directly result in
	\begin{align}
		- \mu_f \nabla \cdot \ve{u}_f = 0
		\quad \therefore \quad
		\Phi = 0 \quad \text{on } \partial\Oft\setminus\GDf
		\,.
		\label{eqn:auxiliary_heat_eqn_BC_neumann}
	\end{align}
	Given these zero Neumann and Dirichlet conditions for $\Phi$, and assuming the geometric conservation law is fulfilled, i.e., a constant quantity is preserved on the moving grid (cf. \cite{Forster2006,Guillard2000,Farhat2001,Boffi2004}) we obtain $\Phi\equiv0$ as the only admissible solution to~\eqref{eqn:auxiliary_heat_eqn}--\eqref{eqn:auxiliary_heat_eqn_BC_neumann}. Consequently, the modified system also inherently enforces incompressibility and we conclude the proof by quickly noting that equivalence of~\eqref{eqn:fluid_splitstep_mom} and~\eqref{eqn:fluid_u_p_mom} is easily shown using~\eqref{eqn:div_sigma_rewrite} again.
\end{proof}

\bibliographystyle{unsrtnat}
\bibliography{references}

\end{document}